\pdfoutput=1
\documentclass[acmsmall,screen,authorversion]{acmart}

\usepackage[utf8]{inputenc}
\usepackage[T1]{fontenc}
\usepackage{textcomp}
\usepackage[american]{babel}
\usepackage{mathtools}
\usepackage{xcolor}
\usepackage{paralist}
\usepackage{stackengine}
\usepackage{wrapfig}
\usepackage{fontawesome}
\usepackage{nicefrac}
\usepackage{xspace}
\usepackage{tabularx}
\usepackage{array}
\usepackage[normalem]{ulem}

\newcommand{\techreport}[1]{}
\newcommand{\nontechreport}[1]{#1}
\makeatletter
\@ifclasswith{acmart}{authorversion}{
  \renewcommand{\techreport}[1]{#1}
  \renewcommand{\nontechreport}[1]{}
}{}
\makeatother

\DeclareFontFamily{U}{MnSymbolC}{}
\DeclareSymbolFont{MnSyC}{U}{MnSymbolC}{m}{n}
\DeclareMathSymbol{\Diamonddot}{\mathbin}{MnSyC}{"7E}
\DeclareFontShape{U}{MnSymbolC}{m}{n}{
    <-6>  MnSymbolC5
   <6-7>  MnSymbolC6
   <7-8>  MnSymbolC7
   <8-9>  MnSymbolC8
   <9-10> MnSymbolC9
  <10-12> MnSymbolC10
  <12->   MnSymbolC12}{}

\usepackage{listings}
\usepackage{upquote} 
\lstdefinelanguage{SPL}{
  morekeywords={acc, method, struct,if,else,returns,procedure,def,requires,ensures,:=,var,
    new,old,free,implicit,modifies,call,locals,assume,assert,choose,havoc,ghost,
    predicate,function,invariant,while, return,atomic, split, type, field, result,
    define, datatype, domain, axiom, val, lock, unlock, not, restart, private, public, match, with, case},
  deletekeywords={union,int},
  numbers=left,
  xleftmargin=2em,
  escapeinside={@}{@},
  numberstyle=\tiny,
  basicstyle=\footnotesize\ttfamily,
  columns=flexible,
  morecomment=[s][\color{green!60!black}]{/*}{*/},
  morecomment=[l][\color{green!60!black}]{//},
  moredelim=[is][\underbar]{`}{'},
  mathescape=true,
}
\lstdefinelanguage{compactSPL}{
  language=SPL,
  aboveskip=0pt,
  belowskip=5pt,
  xleftmargin=1.5em,
  numbersep=5pt,
    moredelim=**[is][\color{colorBug}]{|<}{>|},
}

\lstdefinestyle{proof}{
  language=Java,
  basicstyle=\linespread{1.15}\footnotesize\ttfamily,
  commentstyle=\color{gray}\ttfamily,
  tabsize=2,
  numberbychapter=false,
  keepspaces=true,
  numbers=left,
  numberstyle=\scriptsize,
  numbersep=7pt,
  firstnumber=last,
  mathescape=true,
  keywords={},
  xleftmargin=3pt,
}
\lstdefinestyle{inline}{
  basicstyle=\relscale{.9}\ttfamily,
  keywords={},
}
\usepackage{relsize}
\newcommand{\code}[2][]{\lstinline[style=inline,#1]!#2!}

\usepackage{tikz}
\usetikzlibrary{matrix,arrows,positioning,calc,fit,backgrounds,patterns,snakes}

\tikzset{%
  array/.style={matrix of nodes,nodes={draw, minimum size=5mm, anchor=center},column sep=-\pgflinewidth, row sep=-\pgflinewidth, nodes in empty cells,anchor=center},
  ptr/.style={*->, shorten <=-(1.8pt+1.4\pgflinewidth)},
  edge/.style={->, thick},
  dedge/.style={<->, dashed},
  fedge/.style={->, dashed},
  unode/.style={circle, draw=black, thick, minimum size=8mm},
  mnode/.style={circle, draw=black, thick, fill=gray!20, minimum size=8mm, font=\scriptsize},
  stackVar/.style={circle, fill=none, inner sep=0pt, minimum size=8mm, font=\normalsize},
  gnode/.style={circle, draw=black, thick, minimum size=8mm},
  pnode/.style={circle, draw=black, thick, minimum size=8mm},
  rnode/.style={draw=black, thick, minimum size=8mm},
  lbl/.style={circle, fill=none, inner sep=0pt, minimum size=8mm},
  dnode/.style={circle, draw=black, thick, dotted, minimum size=8mm},
  inflow/.style={circle, fill=none, inner sep=0pt, minimum size=5mm, font=\normalsize},
  phantomNode/.style={circle, fill=none, inner sep=0pt, minimum
    size=0pt},
  treenode/.style = {align=center, inner sep=0pt, text centered, scale=.8, circle, thick, draw=black, minimum width=5mm},
  treeptr/.style = {line width=1.2pt, draw=colorTree, ->, >=stealth},
  listptr/.style = {line width=1pt, densely dashed, draw=colorList, -, >=stealth},
  trees/.style = { level/.style={sibling distance = 1.0cm, level distance = 1.15cm} },
  flattrees/.style = { trees, level/.style={sibling distance = 1.5cm, level distance = .7cm} },
  interval/.style = {above,rotate=42,anchor=south west,inner sep=0pt,font=\footnotesize},
  hinterval/.style = {font=\footnotesize},
}

\usepackage{pftools}

\definecolor{colorUnklar}{RGB}{150,30,190}
\definecolor{colorTodo}{RGB}{200,30,30}
\definecolor{colorHighlight}{RGB}{30,30,200}
\definecolor{colorMyGreen}{RGB}{80,170,0}
\definecolor{colorMyRed}{RGB}{180,20,20}
\definecolor{colorMyBlue}{RGB}{40,40,220}
\definecolor{colorMyPink}{RGB}{220,40,220}
\newcommand{\makePurple}[1]{\textcolor{purple}{#1}}
\newcommand{\makeBlue}[1]{\textcolor{colorMyBlue}{#1}}

\newcommand{\makePink}[1]{\textcolor{colorMyPink}{#1}}

\newcommand{\makeTeal}[1]{\textcolor{teal}{\addtolength{\jot}{-1pt}#1}}

\colorlet{colorTree}{colorMyBlue}
\colorlet{colorList}{purple}
\colorlet{colorBug}{purple}
\newcommand{\makeBug}[1]{{\color{colorBug}#1}}

\definecolor{colorLogic}{RGB}{20,40,180}
\newcommand{\makeColorLogic}[1]{\textcolor{colorLogic}{#1}}
\newcommand{\setColorLogic}{\color{colorLogic}}

\usepackage{mathpartir}
\makeatletter
\newcommand{\rulelabel}[2]{%
   \protected@write \@auxout {}{\string \newlabel {#1}{{#2}{\thepage}{#2}{#1}{}} }%
   \hypertarget{#1}{}
}
\makeatother
\newcommand{\mkrulelabel}[1]{\rulelabel{rule:#1}{\textsc{(#1)}}}
\newcommand{\infrule}[3][]{\ifthenelse{\equal{#1}{}}{\inferrule{#2}{#3}}{\inferrule[\textsc{(#1)}\mkrulelabel{#1}]{#2}{#3}}}

\usepackage{stackengine}
\stackMath
\newcommand\xxrightarrow[1]{\raisebox{-.85pt}{\ensuremath{\smash{\mathrel{%
  \setbox2=\hbox{\stackon{\scriptstyle#1}{\scriptstyle#1}}%
  \stackon[-4.0pt]{%
    \xrightarrow{\makebox[\dimexpr\wd2\relax]{}}%
  }{%
   \scriptstyle#1\,%
  }%
}}}}}

\usepackage{cleveref}
\crefformat{section}{\S#2#1#3}
\crefmultiformat{section}{\S#2#1#3}{ and~\S#2#1#3}{, \S#2#1#3}{, and~\S#2#1#3}
\crefrangeformat{section}{\S#3#1#4 to~\S#5#2#6}

\newcommand{\tinyskip}{\vspace{1pt plus 1pt minus 1pt}}
\newcommand{\smartparagraph}[1]{\smallskip\noindent{\sffamily\emph{#1.}}\ }

\newcommand{\boxedInline}[1]{{\setlength{\fboxsep}{.75\fboxsep}\boxed{#1}}}

\newcommand{\mymathtt}[1]{\text{\relscale{.9}\ttfamily#1}}

\def\prallspacing{\mskip 2mu plus 2mu minus 3mu}
\newcommand{\prall}[1]{{\prallspacing#1\prallspacing}}

\newcommand{\true}{\mathit{true}}
\newcommand{\false}{\mathit{false}}
\renewcommand{\emptyset}{\varnothing}

\newcommand{\setcompact}[1]{\{#1\}}
\newcommand{\set}[1]{\{\,#1\,\}}
\newcommand{\setnd}[1]{\{#1\}}
\newcommand{\setcond}[2]{\set{#1\;\mid\;#2}}

\newcommand{\powerset}[1]{\mathbb{P}(#1)}
\newcommand{\nat}{\mathbb{N}}

\newcommand{\pto}{\rightharpoonup}

\newcommand{\ite}[3]{#1\:?\;#2\::\:#3}

\newcommand{\bnf}{\;\mid\;}

\newcommand{\defeq}{\triangleq} 
\newcommand{\defebnf}{\Coloneqq}

\newcommand{\acom}{\mymathtt{com}} 
\newcommand{\astmt}{\mymathtt{st}} 
\newcommand{\cskip}{\mymathtt{skip}}
\newcommand{\seqof}[2]{#1;#2}
\newcommand{\choiceof}[2]{#1+#2}
\newcommand{\loopof}[1]{{#1}^{*}}

\newcommand{\setcom}{\mathtt{COM}}

\newcommand{\setstates}{\Sigma}

\newcommand{\localindex}{\mathsf{L}}
\newcommand{\sharedindex}{\mathsf{G}}

\newcommand{\setshared}{\Sigma_{\sharedindex}}
\newcommand{\setlocal}{\Sigma_{\localindex}}
\newcommand{\sharedemp}{\emp_{\sharedindex}}
\newcommand{\localemp}{\emp_{\localindex}}
\newcommand{\acfg}{\mathsf{cf}} 
\newcommand{\aconfig}{\acfg} 
\newcommand{\apc}{\mathsf{pc}}
\newcommand{\alocal}{\mathsf{l}}
\newcommand{\ashared}{\mathsf{g}}
\newcommand{\alocalseq}{\lambda}
\newcommand{\asharedseq}{\gamma}
\newcommand{\acomp}{\sigma}

\newcommand{\firstof}[1]{\mathsf{first}(#1)}

\newcommand{\lastof}[1]{\mathsf{last}(#1)}
\newcommand{\astate}{\mathsf{s}}
\newcommand{\astatep}{\mathsf{t}}

\newcommand{\semCom}[1]{\llbracket#1\rrbracket}

\newcommand{\pcStepOf}[3]{#1\,\xxrightarrow{\vphantom{pt}#2}\,#3}
\newcommand{\progStepRel}{\rightarrow}

\newcommand{\anop}{\oplus}

\newcommand{\abort}{\textsf{abort}}

\newcommand{\stateunit}{1}

\newcommand{\statemult}{\mathop{*}}

\newcommand{\sharedmult}{\mathop{{\mstar}_{\sharedindex}}}
\newcommand{\localmult}{\mathop{{\mstar}_{\localindex}}}

\newcommand{\statemultdef}{\mathop{\#}}

\newcommand{\apred}{\mathit{p}}
\newcommand{\aguard}{\mathit{g}}

\newcommand{\apredp}{\mathit{q}}
\newcommand{\apredpp}{\mathit{o}}
\newcommand{\apredppp}{\mathit{r}}

\newcommand{\aninvpred}{\mathit{inv}}

\newcommand{\semof}[1]{\llbracket#1\rrbracket}
\newcommand{\semOf}[1]{\llbracket#1\rrbracket}
\newcommand{\csemOf}[1]{\llbracket\mkern-7mu\semOf{#1}\mkern-7mu\rrbracket}
\newcommand{\csemof}[1]{\llbracket\mkern-7mu\semOf{#1}\mkern-7mu\rrbracket}

\newcommand{\acpred}{{\mathit{a}}}
\newcommand{\acpredp}{\mathit{b}}
\newcommand{\acpredpp}{\mathit{c}}
\newcommand{\acpredppp}{\mathit{d}}

\newcommand{\OBLname}{\mathsf{OBL}}
\newcommand{\FULname}{\mathsf{RCT}}

\newcommand{\mstar}{\mathop{*}}
\newcommand{\MSTAR}{\,\mstar\,}
\DeclareMathOperator*{\bigmstar}{\scalerel*{\ast}{\sum}}
\newcommand{\sepimp}{\mathrel{-\!\!*}}

\newcommand{\emp}{\mathsf{emp}}

\newcommand{\hoareof}[3]{\set{#1}\:#2\:\set{#3}}
\newcommand{\hoareOf}[3]{\set{#1}\:#2\:\set{#3}}

\newcommand{\subModels}{\models}

\newcommand{\semCalc}{\Vdash}
\newcommand{\semcalc}{\Vdash}

\newcommand{\past}{\Diamonddot}

\newcommand{\pastof}[1]{\past#1}
\newcommand{\pastOf}[1]{\past{\left(#1\right)}}
\newcommand{\pastOF}[1]{\past{\bigl(\mkern+1mu#1\mkern+1mu\bigr)}}
\newcommand{\nowof}[1]{\_{#1}}
\newcommand{\nowOf}[1]{\_\left(#1\right)}

\newcommand{\theInterferenceVar}{\mathbb{X}}

\newcommand{\weakhypof}[3]{\mathit{h}(#1, #2, #3)}
\newcommand{\hypof}[3]{\ahyp(#1, #2, #3)}
\newcommand{\proghypof}[4]{\mathit{hcf}(#1, #2, #3, #4)}
\newcommand{\hypholdsof}[2]{#1\, \checkmark\, #2}
\newcommand{\weakpastof}[1]{\nowof{\!\!\pastof{#1}}}
\newcommand{\weakpastOf}[1]{\weakpastof{\left(#1\right)}}

\newcommand{\initset}[2]{\mathsf{Init}_{#1, #2}}
\newcommand{\acceptset}[1]{\mathsf{Acc}_{#1}}
\newcommand{\reachset}[1]{\mathsf{Reach}(#1)}

\newcommand{\theInterference}{\mathbb{I}}
\newcommand{\theSelfInterference}{\mathbb{S}}

\newcommand{\thePredicates}{\mathbb{P}}

\newcommand{\isInterferenceFreeOf}[2][\theInterference]{\boxast_{#1}\,#2}

\newcommand{\nullptr}{\nil}

\newcommand{\keysel}{\mymathtt{key}}
\newcommand{\marksel}{\mymathtt{mark}}

\newcommand{\selof}[2]{#1\mkern+2mu{.}\mkern+2mu#2}

\newcommand{\keyof}[1]{\selof{#1}{\keysel}}
\newcommand{\markof}[1]{\selof{#1}{\marksel}}

\newcommand{\astateseq}{\sigma}
\newcommand{\astateseqp}{\tau}

\newcommand{\contentsof}[1]{\mathsf{C}(#1)}
\newcommand{\keysetof}[1]{\mathsf{KS}(#1)}
\newcommand{\insetof}[1]{\mathsf{IS}(#1)}

\newcommand{\annot}[1]{\makeTeal{\bigl\{\,#1\,\bigr\}}}
\newcommand{\annotml}[1]{\makeTeal{\left\{\,\begin{aligned}#1\end{aligned}\,\right\}}}
\newcommand{\MC}[2][]{\makeBlue{\ifthenelse{\equal{#1}{}}{#2}{\underbracket{#2}_{\text{\textnormal{#1}}}}}}
\newcommand{\MCS}[2][]{\makeBlue{\ifthenelse{\equal{#1}{}}{#2}{\smash{\underbracket{#2}_{\text{\textnormal{#1}}}}}}}
\newcommand{\MF}[2][]{\makePurple{\ifthenelse{\equal{#1}{}}{#2}{\underbracket{#2}_{\text{\textnormal{#1}}}}}}
\newcommand{\MH}[2][]{\makePink{\ifthenelse{\equal{#1}{}}{#2}{\underbracket{#2}_{\text{\textnormal{#1}}}}}}

\newcommand{\anode}{\mathit{x}}
\newcommand{\anodep}{\mathit{y}}

\newcommand{\inflow}{\mathit{in}}

\newcommand{\nil}{\texttt{nil}}

\usepackage{scalerel}

\newcommand{\inflowfld}{\mymathtt{in}}

\newcommand{\vres}{\mathit{res}}

\newcommand{\keyvarof}[1]{\mathit{key}(#1)}

\newcommand{\inflowvarof}[1]{\inflow(#1)}
\newcommand{\markvarof}[1]{\mathit{mark}(#1)}

\newcommand{\inv}{\mathsf{Inv}}

\newcommand{\old}[1]{{}^\backprime\mkern-1mu#1}

\newcommand{\nodeof}[1]{\mathsf{N}(#1)}

\newcommand{\somenodes}{\mathit{N}}
\newcommand{\somenodesp}{\mathit{M}}

\newcommand{\atoolname}[1]{\code{#1}\xspace}

\newcommand{\plankton}{\atoolname{plankton}}

\usepackage{pifont}
\newcommand{\rawsymbolYes}{\smash{\ding{51}}}
\newcommand{\rawsymbolNo}{\smash{\ding{55}}}

\newcommand{\commandof}[1]{\mymathtt{2com}(#1)}
\newcommand{\stmtof}[1]{\mymathtt{2stmt}(#1)}
\newcommand{\governed}[2]{\mathsf{Gov}(#1)}
\newcommand{\governeddef}{\governed{\theInterference}{\theSelfInterference}}
\newcommand{\governeddefnb}{\governed{\theInterference}{\theSelfInterference}}

\newcommand{\semcalcti}[1][]{\semcalc_{#1\mathit{ti}}}
\newcommand{\semcalctinf}{\semcalc_{\mathit{ti}, \mathit{nf}}}

\newcommand{\theHyp}{\mathbb{H}}
\newcommand{\ahyp}{\mathit{h}}

\newcommand{\histof}[2]{\langle#1\rangle#2}

\newcommand{\ahist}{\mathit{hst}}
\newcommand{\sethist}{\mathsf{HST}}
\newcommand{\emphist}{\mathsf{emphst}}
\newcommand{\sethisteps}{\sethist_{\varepsilon}}

\newcommand{\enrichof}[2]{\mathsf{enrich}(#1, #2)}

\newcommand{\decori}[3]{#1\times(#2\rightsquigarrow#3)}
\newcommand{\decorc}[3]{#1\times(#2\rightsquigarrow#3)}
\newcommand{\isonflowpath}[1]{\textit{indegree-one}(#1)}

\newcommand{\anodeval}{\mathit{v}}
\newcommand{\anodevalp}{\mathit{u}}
\newcommand{\avalval}{\mathit{t}}

\stackMath
\newcommand\xxmapsto[1]{\raisebox{-.85pt}{\ensuremath{\smash{\mathrel{%
  \setbox2=\hbox{\stackon{\scriptstyle#1}{\scriptstyle#1}}%
  \stackon[-4.0pt]{%
    \xmapsto{\makebox[\dimexpr\wd2\relax]{}}%
  }{%
   \scriptstyle#1\,%
  }%
}}}}}
\DeclareDocumentCommand\fracto{ g g }{\IfValueTF{#1}{\xxmapsto{%
  ~\text{\nicefrac{\IfValueTF{#2}{#1}{1}}{\IfValueTF{#2}{#2}{#1}}}~%
}}{\xxmapsto{~1~}}}
\newcommand{\persto}{\xxmapsto{~\square~}}

\newcommand{\bigleadsto}{\scalebox{1.6}{\textbf{$\rightsquigarrow$}}}
\newcommand{\bigleadsfrom}{\reflectbox{\scalebox{1.6}{\textbf{$\rightsquigarrow$}}}}
\newcommand{\inslock}{\rotatebox[origin=c]{180}{\textcolor{colorTree}{\faLock}}}
\newcommand{\dellock}{\rotatebox[origin=c]{0}{\textcolor{colorTree}{\faLock}}}
\newcommand{\textlistptr}[1][]{%
  \raisebox{3pt}{\begin{tikzpicture}[baseline=(current bounding box.center)]
    \draw (0,0) edge[listptr,#1] (.5,0);
  \end{tikzpicture}}
}
\newcommand{\texttreeptr}{%
  \raisebox{3pt}{\begin{tikzpicture}[baseline=(current bounding box.center)]
    \path[treeptr] (0,0) -- (.5,0);
  \end{tikzpicture}}  
}

\DeclareDocumentCommand\anobl{ g }{\OBLname_{\IfValueT{#1}{#1}}}
\DeclareDocumentCommand\aful{ m g }{\FULname_{#1\IfValueT{#2}{,#2}}}
\newcommand{\asspec}{\Psi}

\newcommand{\acss}{\mathsf{CSS}}

\newcommand{\acssup}{\mathsf{UP}}
\newcommand{\abskeyset}{\mathcal{K}}
\newcommand{\abscontent}{\mathcal{C}}
\newcommand{\abscontentp}{\mathcal{C}'}
\newcommand{\absop}{\mathit{op}}

\newcommand{\semcalclin}{\semcalcti^\mathit{lin}}

\newcommand{\sharedinv}{\mathsf{SInv}}

\newcommand{\nodeinv}{\mathsf{NInv}}

\newcommand{\succsel}{\mymathtt{succ}}
\newcommand{\predsel}{\mymathtt{pred}}
\newcommand{\leftsel}{\mymathtt{left}}
\newcommand{\rightsel}{\mymathtt{right}}
\newcommand{\parentsel}{\mymathtt{parent}}
\newcommand{\treesel}{\mymathtt{treeLock}}
\newcommand{\listsel}{\mymathtt{listLock}}

\newcommand{\succof}[1]{\selof{#1}{\succsel}}
\newcommand{\predof}[1]{\selof{#1}{\predsel}}
\newcommand{\leftof}[1]{\selof{#1}{\leftsel}}
\newcommand{\rightof}[1]{\selof{#1}{\rightsel}}
\newcommand{\parentof}[1]{\selof{#1}{\parentsel}}
\newcommand{\listlockof}[1]{\selof{#1}{\listsel}}
\newcommand{\treelockof}[1]{\selof{#1}{\treesel}}

\newcommand{\ptrmin}{\mymathtt{min}}
\newcommand{\ptrmax}{\mymathtt{max}}
\newcommand{\ptrroot}{\mymathtt{min}}

\newcommand{\OBLc}[1]{\mathsf{CTN}_{#1}}
\newcommand{\OBLi}[1]{\mathsf{INS}_{#1}}
\newcommand{\OBLd}[1]{\mathsf{DEL}_{#1}}
\newcommand{\FULcid}[1]{\FULname_{#1}}
\newcommand{\FULc}[1]{\FULcid{#1}}
\newcommand{\FULi}[1]{\FULcid{#1}}
\newcommand{\FULd}[1]{\FULcid{#1}}

\newcommand{\leftvarof}[1]{\mathit{left}(#1)}
\newcommand{\rightvarof}[1]{\mathit{right}(#1)}
\newcommand{\parentvarof}[1]{\mathit{parent}(#1)}
\newcommand{\predvarof}[1]{\mathit{pred}(#1)}
\newcommand{\succvarof}[1]{\mathit{succ}(#1)}
\newcommand{\treevarof}[1]{\mathit{tlock}(#1)}
\newcommand{\listvarof}[1]{\mathit{llock}(#1)}

\newcommand{\vk}{\mathit{k}}
\newcommand{\pcurr}{y}
\newcommand{\pchild}{c}
\newcommand{\pprev}{x}
\newcommand{\pnext}{z}
\newcommand{\pparent}{p}
\newcommand{\pnode}{\pcurr}

\newcommand{\glnodeof}[1]{\mathsf{Guarded}(#1)}
\newcommand{\holds}[1]{\mathsf{Locked}(#1)}

\newcommand{\succInv}{\mathsf{SuccInv}}
\newcommand{\locateInv}{\mathsf{LocInv}}
\newcommand{\linkInv}{\mathsf{LinkInv}}
\newcommand{\hiddenIns}[1]{\mathsf{TreeIns}(#1)}
\newcommand{\hiddenDel}[1]{\mathsf{TreeDel}(#1)}

\newcommand{\mand}{\cap}

\DeclareDocumentCommand\ANNOT{ m g }{\makeTeal{\annot{#1\IfNoValueF{#2}{\MSTAR\pastOF{#2}}}}}
\DeclareDocumentCommand\ANNOTML{ m g }{\makeTeal{\annotml{#1}\IfNoValueF{#2}{\mand\pastof{\annotml{#2}}}}}

\setcopyright{rightsretained}
\acmPrice{}
\acmDOI{10.1145/3591296}
\acmYear{2023}
\copyrightyear{2023}
\acmSubmissionID{pldi23main-p657-p}
\acmJournal{PACMPL}
\acmVolume{7}
\acmNumber{PLDI}
\acmArticle{182}
\acmMonth{6}
\received{2022-11-10}
\received[accepted]{2023-03-31}

\makeatletter
\techreport{
    \let\@authorsaddresses\@empty
    \settopmatter{printccs=false,printacmref=false}
}
\makeatother

\begin{document}
	\setdefaultenum{(i)}{(a)}{(a)}{(a)}

	\title{Embedding Hindsight Reasoning in Separation Logic}

	\author{Roland Meyer}
	\orcid{0000-0001-8495-671X}
	\affiliation{%
	  \institution{TU Braunschweig}
	  \country{Germany}
	}
	\email{roland.meyer@tu-bs.de}

	\author{Thomas Wies}
	\orcid{0000-0003-4051-5968}
	\affiliation{%
	  \institution{New York University}
	  \country{USA}
	}
	\email{wies@cs.nyu.edu}

	\author{Sebastian Wolff}
	\orcid{0000-0002-3974-7713}
	\affiliation{%
	  \institution{New York University}
	  \country{USA}
	}
	\email{sebastian.wolff@cs.nyu.edu}

	\bibliographystyle{ACM-Reference-Format}
	\citestyle{acmauthoryear}

	\ccsdesc[500]{Theory of computation~Automated reasoning}
	\ccsdesc[500]{Theory of computation~Hoare logic}
	\ccsdesc[500]{Theory of computation~Separation logic}
	\ccsdesc[300]{Theory of computation~Program verification}
	\ccsdesc[100]{Theory of computation~Programming logic}

	\nontechreport{
		\keywords{Hindsight, Linearizability, Logical Ordering Tree}
	}


\begin{abstract}
\looseness=-1
Automatically proving linearizability of concurrent data structures remains a key challenge for verification. We present temporal interpolation as a new proof principle to guide automated proof search using hindsight arguments within concurrent separation logic. Temporal interpolation offers an easy-to-automate alternative to prophecy variables and has the advantage of structuring proofs into easy-to-discharge hypotheses. Additionally, we advance hindsight theory by integrating it into a program logic, bringing formal rigor and complementary proof machinery. We substantiate the usefulness of temporal interpolation by implementing it in a tool and using it to automatically verify the Logical Ordering tree. The proof is challenging due to future-dependent linearization points and complex structure overlays. It is the first formal proof of this data structure. Interestingly, our formalization revealed an unknown bug and an existing informal proof as erroneous.
\end{abstract}

	\maketitle
	\techreport{{%
		\vspace{-.5em}
		\small\noindent{\bfseries Conference Version:}\par\nobreak\noindent
		Roland Meyer, Thomas Wies, and Sebastian Wolff. 2023.
		Embedding Hindsight Reasoning in Separation Logic.
		\emph{Proc. ACM Program. Lang.} 7, PLDI, Article 182 (June 2023), 24 pages.
		\url{https://doi.org/10.1145/3591296}
		\medskip
	}}



\section{Introduction}
\label{sec:intro}

We are concerned with automatically proving linearizability, the standard correctness criterion for concurrent data structures~\cite{DBLP:journals/toplas/HerlihyW90}. A concurrent data structure is linearizable subject to a sequential specification of its methods, if each method takes effect in a single atomic step of its concurrent execution, the method's \emph{linearization point}, and satisfies the sequential specification in this step.

\nontechreport{\looseness=-1}
Concurrent separation logics~\cite{DBLP:conf/concur/FuLFSZ10,DBLP:conf/esop/GotsmanRY13,DBLP:conf/esop/SergeyNB15,DBLP:conf/ecoop/DelbiancoSNB17,DBLP:conf/sas/BellAW10,DBLP:conf/popl/ParkinsonBO07,DBLP:conf/wdag/HemedRV15,DBLP:conf/concur/VafeiadisP07,DBLP:conf/pldi/GuSKWKS0CR18,DBLP:conf/tacas/ElmasQSST10} provide a powerful toolbox of deductive reasoning techniques to verify complex concurrent data structures.
However, the proof construction heavily relies on the proof author's creativity and expertise in wielding the available tools effectively.
For instance, in order to construct the inductive invariant of the data structure, the proof author may have to devise proof-specific resource algebras to express ghost state that captures the key aspects of the computation history. This hinders proof automation due to the vast complexity of the proof space that needs to be explored. Similarly, the proofs may make use of prophecy variables~\cite{DBLP:journals/tcs/AbadiL91} to predict future-dependent linearization points~\cite{DBLP:phd/ethos/Vafeiadis08,DBLP:conf/pldi/LiangF13,DBLP:journals/pacmpl/JungLPRTDJ20}. Constructing such proofs involves backward reasoning, which is difficult to automate~\cite{DBLP:conf/cav/BouajjaniEEM17}.
It stands to reason that there is a need for guiding principles that help to structure the proof and that provide effective strategies for automated tools to prune the search space.

\looseness=-1
Hindsight theory~\cite{DBLP:conf/podc/OHearnRVYY10,DBLP:conf/wdag/Lev-AriCK15,DBLP:conf/wdag/FeldmanE0RS18,DBLP:journals/pacmpl/FeldmanKE0NRS20} provides such a guiding principle, which we refer to as \emph{temporal interpolation}.
One proves lemmas of the form: if there existed a past state that satisfied property $p$ and the current state satisfies $q$, then there must have existed an intermediate state that satisfied $o$. Such lemmas can then be applied, e.g., to prove the existence of a future-dependent linearization point in hindsight.
Hindsight is 20/20, the arguments only involve forward reasoning, which is easier to automate than, say, prophecy-based arguments.

One limitation of the existing hindsight theory is that it has only explored the general idea of temporal interpolation very narrowly. Concretely, it has been used only to prove hindsight lemmas about concurrent traversals of data structures. These are variations of statements of the form \emph{``if the current node $x$ of the traversal was reachable from the root at some point in the past ($p$), and $y$ is the successor of $x$ in the present state ($q$), then $y$ was reachable from the root at some point in the past~($o$)''}. We show that temporal interpolation applies more broadly in other contexts as well.

Another limitation is that the proof and application of these hindsight lemmas has so far been confined to meta-level linearizability arguments. As a consequence, existing hindsight proofs can lack the rigor enforced by a program logic. We show that this has resulted in at least one incorrect hindsight-based proof in the past~\cite{DBLP:journals/pacmpl/FeldmanKE0NRS20}.

\smartparagraph{Contributions}
Building on~\cite{DBLP:journals/pacmpl/MeyerWW22}, we present a concurrent separation logic that integrates temporal interpolation as a general proof rule. The logic offers the best of both worlds: it enables the intuitive reasoning of hindsight theory within the rigorous framework of a formal proof system. As in~\cite{DBLP:journals/pacmpl/MeyerWW22}, the logic's semantic foundation is based on computations rather than states, which it exposes at the syntactic level in the form of a lightweight temporal operator. This operator provides a uniform mechanism for tracking history information. This reduces the need for introducing proof-specific auxiliary ghost state and helps to prune the space of possible proofs to consider for automatic proof construction. At the same time, the logic offers all advantages of separation logic, including the ability to reason locally about state mutation and concurrency via the frame rule, and to introduce ghost state if and when needed.

The key innovation over~\cite{DBLP:journals/pacmpl/MeyerWW22} is a new proof rule that enables general hindsight reasoning via temporal interpolation.
The proof rule postulates and then applies hypotheses $\weakhypof{p}{q}{o}$ that state the correctness of the temporal interpolation. These hypotheses are collected by the main proof and then discharged in subproofs. This approach provides a proof-structuring mechanism: the subproofs can use a coarse-grained abstraction of the program behavior, which often simplifies the overall proof argument and aids automation. The nature of temporal interpolation as a proof-structuring mechanism is made formally precise in our soundness proof by showing that the proof rule can be eliminated from the logic.

To demonstrate the usefulness of our development, we have integrated temporal interpolation into \plankton~\cite{DBLP:journals/pacmpl/MeyerWW22}, an automated verifier for concurrent search structures based on separation logic.
As a case study, we have used the extended tool \cite{artifact} to automatically verify the logical-ordering (LO-)tree~\cite{DBLP:conf/ppopp/DrachslerVY14}.
The proof exercises the full power of our logic by combining a linearizability argument based on temporal interpolation with local reasoning in separation logic.
To our knowledge, there has been no formal proof of the LO-tree prior to this work (either automated or mechanized). In fact, our efforts identified one previously unreported bug in the original implementation of the data structure. Another bug was identified by \citet{DBLP:journals/pacmpl/FeldmanKE0NRS20}, who presented an informal hindsight-based proof. While the fix proposed by \citet{DBLP:journals/pacmpl/FeldmanKE0NRS20} addresses the original bug, we show that it introduces a new linearizability violation. This underscores the benefit of supporting hindsight proofs in a formal logic.

\smartparagraph{Limitations}
Our focus is on automating linearizability proofs for concurrency library implementations.
In particular, our program logic was not designed for modular verification of library clients against the proved linearizability specifications.
Moreover, \plankton is not yet fully automated: the user provides an invariant describing the properties of each node comprising the data structure in the shared heap.
Finally, the implementation of temporal interpolation in \plankton is currently geared towards reasoning about \emph{pure} future-dependent linearization points (i.e., those that do not modify the abstract state of the data structure).
We leave the handling of impure cases in the implementation as future work.
Though, we note that these cases are not prevalent in the context of concurrent search structures.

\nontechreport{A companion technical report containing additional details is available as \cite{techreport}.}


\section{Overview}
\label{sec:Overview}

\newcommand{\m}[1]{\mymathtt{#1}}
\newcommand{\nl}{n_l}
\newcommand{\nr}{n_r}
\newcommand{\ret}{v}
\newcommand{\counter}{\mymathtt{counter}}

We illustrate our approach using the idealized distributed counter shown in \Cref{fig:dist-counter}.
A counter object $c$ has an abstract state that tracks an integer value $n$ and supports two methods: \code{inc($c$)} atomically increments $n$ by $1$ and \code{read($c$)} returns $n$.
The counter is distributed in the sense that $n$ is the sum of two integer values stored in separate memory locations \code{$c$.l} and \code{$c$.r}.
The implementation of \code{inc} non-deterministically chooses one of the two locations and then atomically increments it using a \emph{fetch-and-add} (\code{FAA}) instruction.
The implementation of \code{read} non-atomically reads the values of the two memory locations and then returns their sum.

\begin{figure}
  \begin{minipage}[t]{.4\linewidth}
  \begin{lstlisting}[gobble=4,language=SPL,mathescape=true,escapechar=@,aboveskip=0pt,belowskip=-4pt]
    struct C { var l: Int; var r: Int }
    predicate counter($c$: C, $n$: Int) {
      $\exists$ $\nl$ $\nr$ ::
        $c$.l $\mapsto$ $\nl$ $\mstar$ $c$.r $\mapsto$ $\nr$ $\land$ 
        $n$ == $\nl$ + $\nr$
    }
  \end{lstlisting}
  \end{minipage}\hfill
  \begin{minipage}[t]{.27\linewidth}
  \begin{lstlisting}[gobble=4,language=SPL,mathescape=true,escapechar=@,aboveskip=0pt,belowskip=-4pt]
    $\annot{\counter(c, n)}$    
    method read($c$: C) {
      val $x$ = $c$.l @\label{line:counter-read-l}@
      val $y$ = $c$.r @\label{line:counter-read-r}@
      return $x$ + $y$
    } $\annot{\ret.\, \counter(c, n) \mstar \ret = n}$    
  \end{lstlisting}
  \end{minipage}\hfill
  \begin{minipage}[t]{.30\linewidth}    
  \begin{lstlisting}[gobble=4,language=SPL,mathescape=true,escapechar=@,aboveskip=0pt,belowskip=-4pt]
    $\annot{\counter(c, n)}$    
    method inc($c$: C) {
      if (nondet())
        FAA($c$.l, 1)
      else FAA($c$.r, 1)
    } $\annot{\counter(c, n+1)}$    
  \end{lstlisting}
\end{minipage}
\caption{Distributed counter object.\label{fig:dist-counter}}
\vspace{-1em}
\end{figure}

Our goal is to prove that the distributed counter is linearizable with respect to its sequential specification, which is given in \Cref{fig:dist-counter} as Hoare annotations expressed in separation logic.
The specification uses the predicate $\counter(c,n)$ to define the abstract state of the counter $c$ in terms of the underlying memory representation.
Here, a \emph{points-to} predicate $a \mapsto v$ expresses ownership of the memory location at address $a$ and, moreover, that this location stores value $v$.
The operator $p \mstar q$ is \emph{separating conjunction}, which expresses that $p$ and $q$ hold over disjoint memory regions. In the following, we assume an \emph{intuitionistic semantics} of these predicates, i.e. $p * \true = p$.

To prove linearizability, we need to show that each method transforms its precondition to its postcondition in a single atomic step.
Due to interferences by concurrent \code{inc} methods, the counter value may change throughout the execution of a method.
Hence, the value $n$ in the precondition of the specification does not refer to the counter's initial abstract state when the method is invoked, but rather to its abstract state at the linearization point. This semantics of the Hoare annotations corresponds to that of logically atomic triples~\cite{DBLP:conf/ecoop/PintoDG14}.
Note that the variable $v$ in the postcondition of \code{read} is bound to the method's return value.

The linearization point of \code{inc} is when \code{FAA} is executed and the desired Hoare specification follows immediately from the specification of \code{FAA}.
So we focus on the more interesting case of \code{read}.
The \code{read} method does not change the value of the counter.
Hence, it suffices to show that the returned value $x + y$ is equal to the counter value $n$ at the linearization point.
The challenge is that the linearization point depends on the future interferences of concurrent \code{inc} operations.
In fact, it may lie in a concurrently executing \code{inc}.
For example, consider the scenario where at the point when \code{read} executes Line~\ref{line:counter-read-l}, we have $c.\m{l} = c.\m{r} = 0$ and before it proceeds to Line~\ref{line:counter-read-r}, two concurrent \code{inc}s increment first $c.\m{l}$ and then $c.\m{r}$ to $1$.
That is, when \code{read} executes Line~\ref{line:counter-read-l} we have $n=0$ and when it executes Line~\ref{line:counter-read-r} we have $n=2$, yet the return value is $1$.
Nevertheless, this execution of \code{read} is linearizable because there is a time point in between when $n=1$, namely right after the linearization point of the first concurrently executing \code{inc}. Note that if the second \code{inc} incremented $c.\m{l}$ instead of $c.\m{r}$, then the return value of \code{read} would be $0$ and its linearization point would already be when it reads $c.\m{l}$. This is why the linearization point of \code{read} is future-dependent.

Intuitively, the linearizability of \code{read} follows from the fact that the two memory locations increase monotonically by increments of $1$.
So if the counter has value $n$ at some point $t$ and value $n' > n$ at some later point $t'$, then for each value $n''$ with $n \leq n'' \leq n'$ there is an intermediate state between $t$ and $t'$ where the value of the counter is $n''$.
We demonstrate how to formalize this intuitive argument in our program logic.
The logic enables temporal reasoning about computations
using \emph{past predicates}, $\past{\apred}$, which express that the \emph{state predicate} $\apred$ held true at some prior state of the computation.
Our goal is to derive that $\past{(\counter(c,x+y))}$ is true after Line~\ref{line:counter-read-r}.
This implies the existence of a linearization point for \code{read}.

The proof proceeds in two parts.
The first part proves the goal above but assumes the validity of an auxiliary hypothesis that is derived during the proof.
This hypothesis captures the intuitive reasoning used above to conclude the existence of an unobserved intermediate state due to interferences by other threads.
The second part of the proof discharges this hypothesis.

An outline of the first part of the proof is shown in \Cref{fig:dist-counter-proof}.
Throughout the proof, variables that do not occur in the program code such as $n_l$ are implicitly existentially quantified. The program logic follows a thread-modular approach that mostly uses sequential Hoare-style reasoning. The soundness of this reasoning is guaranteed by ensuring that each two consecutive atomic commands are separated by an \emph{interference-free} intermediate assertion. That is, concurrently executing threads will not affect the truth value of this assertion. In the following, we elude the details of the mechanism used to check interference freedom as it is orthogonal to our core contributions. The details of this mechanism are presented in \Cref{sec:Preliminaries}.

\begin{wrapfigure}[20]{r}{5.7775cm}
  \vspace{-4.5mm}
  \techreport{\vspace{-2mm}}
  \newcommand{\MKG}{\color{green!60!black}}
\begin{lstlisting}[language=compactSPL,aboveskip=2pt,belowskip=-4pt]
$\annot{\counter(c, n)}$    
method read($c$: C) {
  $\annot{c.\m{l} \mapsto \nl \mstar c.\m{r} \mapsto \nr}$ @\label{line:counter-proof-init}@
  val $x$ = $c$.l
  $\annotml{&c.\m{l} \mapsto \nl \mstar c.\m{r} \mapsto \nr \land x = \nl \\& \land \past{(c.\m{l} \mapsto \nl \mstar c.\m{r} \mapsto \nr)}}$ @\label{line:counter-proof-after-read-l}@
  $\annotml{&c.\m{l} \mapsto {\nl'} \mstar c.\m{r} \mapsto {\nr'} \land x = \nl \\& \land \past{(c.\m{l} \mapsto \nl \mstar c.\m{r} \mapsto \nr \land \nr \leq {\nr'})}}$ @\label{line:counter-proof-after-read-l-stable}@
  val $y$ = $c$.r
  $\annotml{&c.\m{l} \mapsto {\nl'} \mstar c.\m{r} \mapsto {\nr'} \land x = \nl \land y = {\nr'} \\& \land \past{(c.\m{l} \mapsto \nl \mstar c.\m{r} \mapsto \nr \land \nr \leq {\nr'})}}$ @\label{line:counter-proof-after-read-r}@
  //@\:@Hypothesis: @\MKG$\forall \nl\, \nr\, \nr'.\:\hoareOf{\acpred}{\theInterference^*}{\acpredp}  $@@\label{line:counter-proof-hypothesis}@
  //@\:\MKG\rlap{$\acpred \defeq c.\m{l} \mapsto \nl \mstar c.\m{r} \mapsto \nr \land \nr \leq \nr'$}@
  //@\:\MKG\rlap{$\acpredp \defeq c.\m{r} \mapsto \nr' \mkern+1mu\rightarrow\mkern+1mu \weakpastof{\mkern-2mu(\counter(c, \nl + \nr')\mkern-2mu)}$}@
  $\annotml{&\counter(c,{n'}) \land x = \nl \land y = {\nr'} \\& \land \past{(\counter(c, \nl + \nr'))}}$ @\label{line:counter-proof-temp-interp}@
  $\annot{\counter(c,n') \land \past{(\counter(c,x + y))}}$ @\label{line:counter-proof-final}@
  return $x$ + $y$
} $\annot{\ret.\, \counter(c, n) \mstar \ret = n}$
\end{lstlisting}
  \caption{Proof outline for the \code{read} method.\label{fig:dist-counter-proof}}
\end{wrapfigure}

The proof starts by unfolding the definition of $\counter(c,n)$ in the precondition, yielding the assertion on Line~\ref{line:counter-proof-init}.
After reading $c.\m{l}$ we know that $x$ is bound to the old value $n_l$ of $c.\m{l}$.
We also record the state of the counter before the read command in a past predicate $\past{(c.\m{l} \mapsto \nl \mstar c.\m{r} \mapsto \nr)}$, yielding the assertion on Line~\ref{line:counter-proof-after-read-l}.
This assertion is not interference-free because concurrent \code{inc} threads may change the values of $c.\m{r}$ and $c.\m{l}$.
We therefore weaken the assertion by introducing fresh variables $\nl'$ and $\nr'$ for these values. We leave $\nl'$ unconstrained but preserve $\nr \leq \nr'$, capturing that concurrent threads can only increase $c.\m{r}$. Since $\nr \leq \nr'$ only concerns logical variables, we can push this fact into the past predicate. The resulting interference-free assertion is shown on Line~\ref{line:counter-proof-after-read-l-stable}.

We proceed similarly for the read of $c.\m{r}$ resulting in the assertion on Line~\ref{line:counter-proof-after-read-r}. Again, this assertion is not interference-free because concurrent threads may change the value of $c.\m{r}$. We want to weaken this assertion to the interference-free assertion on Line~\ref{line:counter-proof-final}, which implies our desired goal. Observe that Line~\ref{line:counter-proof-final} follows from Line~\ref{line:counter-proof-temp-interp} using equality reasoning. So it remains to connect lines~\ref{line:counter-proof-after-read-r} and \ref{line:counter-proof-temp-interp}. First, observe that the predicate $\counter(c,n')$ is obtained from $c.\m{l} \mapsto \nl' \mstar c.\m{r} \mapsto \nr'$ by choosing $n' = \nl' + \nr'$. To derive, $\past{(\counter(c, \nl + \nr'))}$, the proof conjectures the validity of the hypothesis on Line~\ref{line:counter-proof-hypothesis}. This hypothesis is a Hoare triple of the shape $\hoareOf{p}{\theInterference^*}{q \rightarrow \weakpastof{o}}$. Here, $\rightarrow$ is logical implication and $\weakpastof{o}$ is syntactic sugar for $o \lor \pastof{o}$. The variable $\theInterference$ stands for a set of \emph{interferences} that the overall proof infers as an auxiliary output of its derivation. The set $\theInterference$ consists of pairs $(\aguard,\acom)$ where $\acom$ is any atomic command in the program that affects the thread-local or shared program state, and $\aguard$ is the intermediate assertion preceding $\acom$ in the proof. In our example, the derived interferences all come from the \code{inc} method.
They comprise the set
\[\theInterference = \{ (c.\m{l} \mapsto v, \m{FAA}(c.\m{l}, 1)), (c.\m{r} \mapsto v, \m{FAA}(c.\m{r}, 1))\}\enspace.\]
Each interference can be viewed as a guarded command that first assumes $\aguard$ and then executes $\acom$. From these guarded commands, we build the new program $\theInterference^*$ which nondeterministically executes the interferences in $\theInterference$ an arbitrary number of times. That is, $\theInterference^*$ can be viewed as abstracting the overall program. Thus, the hypothesis $\hoareOf{p}{\theInterference^*}{q \rightarrow \weakpastof{o}}$ states that if execution starts from a state that satisfies $p$ and after any number of program steps it reaches a state that satisfies $q$, then $o$ must have been true in some intermediate state.
The temporal interpolation rule allows us to derive from such a hypothesis that if the program is in a state $\astate$ that satisfies $\pastof{p} \land q$, then also $\weakpastof{o}$ holds in $\astate$. We use temporal interpolation to derive Line~\ref{line:counter-proof-temp-interp} from Line~\ref{line:counter-proof-after-read-r} using the hypothesis on Line~\ref{line:counter-proof-hypothesis}.

The second part of the proof is then to establish the validity of the hypothesis. This part can also be carried out in the logic, using the same thread-modular and local reasoning principles. Effectively, the proof boils down to finding an invariant $\aninvpred$ that is implied by $p$, implies $q \rightarrow \weakpastof{o}$, and is preserved by each of the interferences. In our example, the following invariant does the trick:
\[\aninvpred \;\defeq\; \exists \nl''\,\nr''.\; c.\m{l} \mapsto \nl'' \,\mstar\, c.\m{r} \mapsto \nr'' \land \nl \leq \nl'' \land \nl'' + \nr'' < \nl + \nr' \lor \weakpastof{(\counter(c, \nl + \nr'))}\]
Intuitively, the first disjunct of the invariant holds up to the linearization point and afterwards, the second disjunct holds.
Note that $\aninvpred$ contains a past operator and is therefore a \emph{computation predicate}, not a state predicate.

\looseness=-1
We contrast the above proof with one based on prophecy reasoning in the style of~\cite{DBLP:journals/pacmpl/JungLPRTDJ20}. Without temporal interpolation, the proof has to witness the linearization point of a \code{read} thread $t$ at the exact moment where the relevant \code{inc} thread sets $n$ to $n_l+n_r'$ for the value $n_r'$ that will be later read by $t$.
However, $n_r'$ depends on how many other \code{inc} threads will still increment \code{r} between these two points.
One can introduce a prophecy variable for $t$ that predicts the number of such increments between the points when $t$ reads \code{l} and \code{r}.
To establish the linearizability argument, the prophecy variables and linearization obligations for the unboundedely many \code{read} threads need to be shared with all \code{inc} threads that may execute concurrently.
 This involves a complex \emph{helping protocol} construction that governs the transfer of resources between threads.
This construction is reflected in the proof in the form of a more complex invariant capturing the shared state of the data structure.


\section{Preliminaries}
\label{sec:Preliminaries}
We study concurrency libraries, i.e., a single program executed by a potentially unbounded number of threads.
We give a formal account of concurrency libraries and introduce a Hoare-style proof system for verifying them.
Our formalism is based on \cite{DBLP:journals/pacmpl/MeyerWW22}.

\subsection{Programming Model}

Along the lines of abstract separation logic~\cite{DBLP:conf/lics/CalcagnoOY07,DBLP:conf/popl/Dinsdale-YoungBGPY13,DBLP:journals/jfp/JungKJBBD18}, the actual sets of states and commands are a parameter to our development.

\smartparagraph{States and Computations}
We draw \emph{states} from a separation algebra, a partial commutative monoid $(\setstates, \statemult, \emp)$
 with a set of units $\emp$ so that
\begin{inparaenum}
	\item each state $\astate\in\setstates$ has a unit $\stateunit\in\emp$ with $\astate\statemult\stateunit=\astate$, and
	\item $\stateunit\statemult\stateunit'$ is undefined for any two distinct units $\stateunit,\stateunit'\in\emp$.
\end{inparaenum}
Definedness of $\astate\statemult\astate'$ is denoted $\astate\statemultdef\astate'$.

We work over a separation algebra with a certain structure.
We expect states from $(\setstates, \statemult, \emp)$ to be composed from a \emph{global} and a \emph{local} state.
The global resp. local states are again drawn from separation algebras $(\setshared, \sharedmult, \sharedemp)$ resp. $(\setlocal, \localmult, \localemp)$.
We require that
\begin{inparaenum}
	\item states in $\setstates$ are multiplied elementwise, $(\ashared_1, \alocal_1)\statemult(\ashared_2, \alocal_2)\defeq(\ashared_1\sharedmult\ashared_2, \alocal_1\localmult\alocal_2)$ provided the resulting state is in~$\setstates$ and undefined otherwise,
	\item states $\setstates$ can be decomposed, $(\ashared_1\sharedmult\ashared_2, \alocal_1\localmult\alocal_2) \in \setstates$ implies $(\ashared_1, \alocal_1) \in \setstates$, and
	\item units $\emp$ are also composed, $\emp\defeq \sharedemp \times \localemp$.
\end{inparaenum}
It is readily checked that this is a separation algebra.

The temporal interpolation principle we propose reasons over knowledge obtained at different points in time during a computation.
To formulate it, we lift the given separation algebra $(\setstates, \statemult, \emp)$ to a separation algebra over computations $(\setstates^+, \statemult, \setstates^*.\emp)$. 
A computation is a non-empty sequence of states.
We write $\astateseq . \astateseqp$ for the concatenation of two computations $\astateseq$ and $\astateseqp$.
The multiplication of two computations $\astateseq.\astate, \astateseqp.\astatep  \in \setstates^+$ is defined, $\astateseq.\astate\statemultdef\astateseqp.\astatep$, if $\astateseq=\astateseqp$ and $\astate\statemultdef\astatep$. 
In this case, the multiplication yields $\astateseq.\astate\statemult\astateseqp.\astatep\defeq\astateseq.(\astate\mstar\astatep)$. 
The two computations share the same history, which is preserved by the multiplication.
In the current state, we use the composition given by the separation algebra.
This construction works in general, not just for our product separation algebra.  

\begin{lemma}
	If $(\setstates, \mstar, \emp)$ is a separation algebra, so is $(\setstates^+, \mstar, \setstates^*.\emp)$.
\end{lemma}

\smartparagraph{Predicates}
For clarity of the exposition, we refrain from introducing an assertion language that needs to be interpreted but work on the semantic level.  
Given a separation algebra $(\Gamma, \mstar, \emp)$, a \emph{predicate} is a set of elements from $\Gamma$. 
The predicates form a Boolean algebra $(\powerset{\Gamma}, \cup, \cap, \subseteq, \overline{\phantom{\bullet}}, \emptyset, \Gamma)$ with disjunction, conjunction, implication, negation, false, and true. 
We moreover have the standard connectives \emph{separating conjunction} $\mstar$ and \emph{separating implication} $\sepimp$:
\begin{align*}
	\apred\statemult\apredp
		\:\defeq\:
		\setcond{
			\gamma_1\statemult\gamma_2
		\!}{\!
			\gamma_1\prall{\in}\apred
			\wedge
			\gamma_2\prall{\in}\apredp
			\wedge
			\gamma_1\statemultdef\gamma_2
		}
	\quad\text{and}\quad
	\apred\sepimp\apredp
		\:\defeq\:
		\setcond{
			\gamma
		\!}{\!
			\setcompact{\gamma} \mstar \apred
			\subseteq
			\apredp
		}
	\: .
\end{align*}
A predicate $\apred$ is \emph{intuitionistic}, if $\apred\mstar\Gamma\subseteq\apred$.

In our setting, we have the separation algebra of states $(\setstates, \mstar, \emp)$ and \emph{state predicates} $\apred, \apredp, \apredpp\subseteq\setstates$.
We moreover have the separation algebra of computations $(\setstates^+, \mstar, \emp^+)$ and \emph{computation predicates} $\acpred, \acpredp, \acpredpp\subseteq \setstates^+$. 
For our temporal interpolation principle developed in \Cref{Section:TemporalInterpolation}, it suffices to consider simple computation predicates that reason about single states of the computation.
These computation predicates are derived from state predicates. 

\begin{definition}
	\label{def:comp-preds}
	From state predicates $\apred \subseteq \setstates$ we construct
	\begin{inparaenum}[(i)]
			\item the \emph{now predicate} $\nowof{\apred}\,\defeq\,\setstates^*.\apred$ and
			\item the \emph{past predicate} $\pastof{\apred}\,\defeq\, \setstates^*.\apred.\setstates^+$ and
			\item the \emph{weak past predicate} $\weakpastof{\apred}\,\defeq\, \nowof{\apred}\cup\pastof{\apred}$. 
	\end{inparaenum}
\end{definition}

Now predicates lift state predicates to hold in the last (the current) state of a computation.
Past predicates lift state predicates to hold at some time in the past of the computation.
The precise moment when the state predicate was true is not known, which means framing is not relevant for past predicates, and lead us to define the multiplication of computations as an intersection in the past. 
Intuitionism carries over from state to computation predicates. 

\begin{lemma}\label{Lemma:PreciseIntuitionistic}
	If $\apred$ is intuitionistic, so is $\nowof{\apred}$. 
	Predicate $\pastof{\apred}$ is intuitionistic. 
\end{lemma}

The predicates are compatible with the separation logic operators as follows.
\begin{lemma}
	\label{Lemma:SLOperators}
	\begin{inparaitem}[]
		\item $\nowOf{\apred\anop\apredp} = \nowof{\apred}\anop\nowof{\apredp}$ for all $\anop\in\setnd{\cap, \cup, \mstar,\sepimp}$,\,
		\item $\nowof{\overline{\apred}} {\:=\:} \overline{\nowof{\apred}}$,\,
		\item $\false {\:=\:} \nowof{\false}$,\,
		\item $\true {\:=\:} \nowof{\true}$,\,
		\item $\nowof{\apred}{\:\subseteq\:}\nowof{\apredp}$ iff $\apred{\:\subseteq\:}\apredp$,\,
		\item $\pastOf{\apred\cap\apredp} \subseteq \pastof{\apred}\cap\pastof{\apredp}$,\,
		\item $\pastOf{\apred\cup\apredp} = \pastof{\apred}\cup\pastof{\apredp}$, and
		\item $\pastof{\apred}\subseteq\pastof{\apredp}$ iff $\apred\subseteq\apredp$.
	\end{inparaitem}
\end{lemma}

\smartparagraph{Commands}
We assume a potentially infinite set of commands $(\setcom, \semCom{-})$. 
The actual set is a parameter and not relevant for our development. 
Commands $\acom\in\setcom$ transform a pre state into a post state which, due to non-determinism, need not be unique. 
This state transformer is given by the interpretation $\semCom{\acom}:\setstates\to\powerset{\setstates}$ of $\acom$. 
We lift the transformer to computations by appending the post state: $\csemOf{\acom}(\astateseq.\astate)\defeq\setcond{\astateseq.\astate.\astate'}{\astate'\in\semOf{\acom}(\astate)}$.
The transformer extends to predicates in the usual way~\cite{DBLP:books/ph/Dijkstra76}: \(
	\csemOf{\acom}(\acpred) \defeq {\textstyle \bigcup_{\astateseq\in\acpred}\;} \csemOf{\acom}(\astateseq)
\).
We expect to have a neutral command $\cskip\in\setcom$ that is interpreted as the identity. 
To model failing commands, we follow~\cite{DBLP:conf/lics/CalcagnoOY07} and assume their post state to be $\abort$, a dedicated top value in the lattice of predicates.

For the frame rule to be sound, we require the following locality:
\begin{align*}
	\label{cond-loccom}
	\forall \acpred,\acpredp,\acpredpp \subseteq \setstates^+.
	\quad
	\csemOf{\acom}(\acpred) \subseteq \acpredp
	\quad\text{implies}\quad
	\csemOf{\acom}(\acpred\mstar\acpredpp) \subseteq \acpredp\mstar\acpredpp
	\ .
	\tag{LocCom}
\end{align*}
Note that \eqref{cond-loccom} requires the computation predicate $\acpredpp$ to perform a stuttering step when being framed on the right-hand side of the latter inclusion. 
We call a computation predicate $\acpredpp\subseteq\setstates^+$ \emph{frameable}, if $\sigma.\astate\in\acpredpp$ implies $\sigma.\astate.\astate\in\acpredpp$ for all $\sigma.\astate\in\setstates^+$. 
Fortunately, all computation predicates that are constructed by union, intersection, and separating conjunction from now and past predicates are frameable.
Unless otherwise stated, we will assume that all predicates we encounter are frameable.

\smartparagraph{Concurrency libraries}
Concurrency libraries consist of an unbounded number of threads that all execute the same program $\astmt$.
Different functions would be modeled by an initial non-deterministic choice among the function bodies, which is supported in our while language together with sequential composition and repetition:
\begin{align*}
	\astmt\;\defebnf\;
		\acom
		\bnf
		\choiceof{\astmt}{\astmt}
		\bnf
		\seqof{\astmt}{\astmt}
		\bnf
		\loopof{\astmt}
		\ .
\end{align*}
A configuration $\aconfig=(\asharedseq, \apc)$ of the library comprises a global computation $\asharedseq\in\setshared^+$ and a program counter $\apc$.
The program counter maps thread identifiers $i\in\nat$ to pairs $\apc(i)=(\alocalseq, \astmt)$ containing thread-$i$-local information: a computation $\alocalseq\in\setlocal^+$ and a program fragment $\astmt$ the execution of which remains.
The transition rules among configurations are as expected: a step of thread $i$ changes the shared and the thread-$i$-local information according to the transformer of the executed command, and leaves all other threads unchanged.

Towards a Hoare-style proof system, we call a configuration $(\asharedseq, \apc)$ initial wrt. computation predicate $\acpred$ and program $\astmt$, if all threads $i$ with $\apc(i)=(\alocalseq,\astmt')$ satisfy $(\asharedseq,\alocalseq)\in\acpred$ and $\astmt'=\astmt$.
Similarly, $(\asharedseq, \apc)$ is accepting wrt. $\acpredp$, if all terminated threads with $\apc(i)=(\alocalseq,\cskip)$ satisfy $(\asharedseq,\alocalseq)\in\acpredp$.
Reachability is defined as usual.
We refer to the initial, accepting, and reachable configurations by $\initset{\acpred}{\astmt}$, $\acceptset{\acpredp}$, and $\reachset{\aconfig}$, respectively.

The correctness condition we would like to prove for concurrency libraries is whether all configurations reachable from $\acpred$-$\astmt$-initial configurations are $\acpredp$-accepting, $\reachset{\initset{\acpred}{\astmt}}\subseteq\acceptset{\acpredp}$.
In this case, we say that a Hoare triple of the form $\hoareOf{\acpred}{\astmt}{\acpredp}$ is \emph{valid}, denoted by $\subModels\hoareOf{\acpred}{\astmt}{\acpredp}$.

\subsection{Program Logic}

We use a proof system to establish the validity of Hoare triples, \Cref{Figure:ProgramLogicTI} below (ignore the \makeColorLogic{marked} parts for now).
The proof system is thread-modular~\cite{DBLP:conf/cav/BerdineLMRS08,DBLP:journals/toplas/Jones83} in nature, thus verifies a single thread in isolation.
To account for the actions of other threads which may affect the isolated thread, we ensure \emph{interference freedom}~\cite{DBLP:journals/acta/OwickiG76} of the overall proof.

Technically, the proof system establishes judgements $\thePredicates, \theInterference \semCalc \hoareOf{\acpred}{\astmt}{\acpredp}$ with the following components:
\begin{inparaenum}
	\item a Hoare triple $\hoareOf{\acpred}{\astmt}{\acpredp}$ for the isolated thread,
	\item a set $\thePredicates$ of intermediary assertions used during the proof of the Hoare triple, and
	\item a set $\theInterference$ of interferences that the isolated thread is subject to.
\end{inparaenum}
Recording the intermediary assertions allows us to separate the interference freedom check from the derivation of the Hoare triple~\cite[Section 7.3]{DBLP:conf/popl/Dinsdale-YoungBGPY13}.
We denote the interference freedom of $\thePredicates$ under $\theInterference$ by $\isInterferenceFreeOf[\theInterference]{\thePredicates}$.
The resulting proof system is sound.

\begin{theorem}[\citet{DBLP:journals/pacmpl/MeyerWW22}]
	\label{Theorem:SoundnessComput}
	$\thePredicates, \theInterference\semCalc\hoareOf{\acpred}{\astmt}{\acpredp}$
	and $\,\isInterferenceFreeOf[\theInterference]{\thePredicates}$
	and $\acpred\in\thePredicates$
	imply $\,\subModels\hoareOf{\acpred}{\astmt}{\acpredp}$.
\end{theorem}

In our development, we will use the set of computations $\csemof{\astmt}_{\theInterference}(\acpred)$ defined by extending each computation in $\acpred$ by every sequence of states encountered when executing program $\astmt$ to completion while admitting interferences from $\theInterference$.
The formal definition is the straightforward lift of $\csemof{-}$ to sequences of commands and interferences. 
A consequence of the soundness result is the following.

\begin{lemma}\label{Lemma:SoundInclusion}
	If there is a set $\thePredicates$ with $\acpred\in\thePredicates$, $\isInterferenceFreeOf[\theInterference]{\thePredicates}$, and $\thePredicates, \theInterference\semCalc\hoareOf{\acpred}{\astmt}{\acpredp}$, then $\csemof{\astmt}_{\theInterference}(\acpred)\subseteq\acpredp$. 
\end{lemma}

\smartparagraph{Interference Freedom}
The isolated thread is influenced by the actions of other, interfering threads.
We capture those actions as \emph{interferences} $(\acpredpp, \acom)$, meaning that $\acom$ may be executed by an interfering thread from a configuration satisfying $\acpredpp$.
Observe that the global portion of $\acpredpp$ imposes restrictions on when the interference may happen while the local portion of $\acpredpp$ supplies the local computation the interfering thread needs for its execution. 
From the point of view of the isolated thread with computation $(\asharedseq, \alocalseq)$, only the global portion $\asharedseq$ changes, formally:
\begin{align*}
	\csemOf{(\acpredpp, \acom)}(\asharedseq, \alocalseq)
		~\defeq~
		\setcond{(\asharedseq', \alocalseq)}{
			\exists \alocalseq_1,\alocalseq_2.~~
			(\asharedseq, \alocalseq_1)\in \acpredpp
			~\wedge~
			(\asharedseq', \alocalseq_2)\in \csemof{\acom}(\asharedseq, \alocalseq_1)
		}
	\ .
\end{align*}
The \emph{interference freedom check} wrt. a set $\theInterference$ of interferences then proceeds as follows.
It takes a computation predicate $\acpred$ and tests whether $\csemOf{(\acpredpp, \acom)}(\acpred)\subseteq \acpred$ for all $(\acpredpp, \acom)\in\theInterference$.
If this is the case, the interference does not invalidate $\acpred$ and the predicate is interference-free.
The interference freedom check extends naturally to the set of predicates $\thePredicates$.
We write $\theInterference\mstar\acpredp$ for the set of interferences $(\acpred\mstar\acpredp, \acom)$ with $(\acpred, \acom)\in\theInterference$.
We also use the notation for sets of predicates $\thePredicates$ and write $\thePredicates\mstar\acpredp$ for the set of predicates $\acpred\mstar\acpredp$ with $\acpred\in\thePredicates$.
We also remark that past information is always interference-free, because interferences append states and this does not change the past of the computation.


\section{Temporal Interpolation}
\label{Section:TemporalInterpolation}

Temporal interpolation is a reasoning principle to derive information about intermediary states that have not been observed in the program proof.
Coming back to the example of a distributed counter, if the counter value has been $n_1$ in the past and is now $n_2> n_1$, then we wish to derive that there has been a moment in which the counter has been $n$ with $n_2\geq n\geq n_1$.
Temporal interpolation will allow us to do so, although an assertion with counter value $n$ is not interference-free and hence will not be observable in the program proof.
We can actually guarantee that the moment in which the counter was $n$ is in between the past and the current state, but defer the timing aspect for now.
Another example of temporal interpolation is reachability in concurrent data structures, 
as studied by the hindsight principle which inspired this work~\cite{DBLP:conf/podc/OHearnRVYY10,DBLP:conf/wdag/FeldmanE0RS18,DBLP:journals/pacmpl/FeldmanKE0NRS20}.
If a node $n_1$ has been reachable in the past, and the node now points to $n_2$, then there has been a moment in which the node was reachable and pointed to $n_2$. 
Also this moment will not be interference-free and hence cannot be recorded in the program proof (the set of predicates $\thePredicates$).

To derive the intermediary information, temporal interpolation proves inclusions of the form
\begin{align}
	\weakpastof{\apred} \,\cap\, \nowof{\apredp} \quad\subseteq\quad \weakpastof{\apredpp}
	\ .
	\label{Equation:TINaive}
\end{align}
The inclusion indeed formulates an interpolation property for the set of computations:
if state predicate $\apred$ has been true in the past of the computation and we now have $\apredp$, then there has been a moment in which $\apredpp$ was true, and typically $\apredpp$ will be $\apred\cap\apredp$. 
Unfortunately, the inclusion will rarely hold in this generality.
The first problem is that the set of computations leading from $\apred$ to $\apredp$ is too liberal.
Rather than considering all sequences of states, we should only consider the ones generated by the program at hand.
The second problem is that even if we restrict the computations, we need to prove the inclusion.
Our technical contribution is to embed the above inclusion into a proof system in which it can justifiably be used.

To restrict the set of computations leading from $\apred$ to $\apredp$, we introduce a new predicate that reflects the influence of the program on the course of the computation.
The observation behind the definition of the predicate is that the set of interferences $\theInterference$ which we collect during the proof gives us precise information about the program behavior.
An interference $(\acpred, \acom)\in\theInterference$ not only says that a command $\acom$ is executable, it also records in predicate $\acpred$ the conditions under which the command will be executed.
Notably, these conditions refer to the shared as well as the local state, meaning the interference captures the thread-local behavior as well.
The new predicate thus employs the set of interferences as an abstraction of the overall program behavior.

To make the idea formal, we transform interferences into programs as follows: \[
	\commandof{\acpred, \acom} \,\defeq\,
		\mymathtt{atomic}\{\,\mymathtt{assume}(\acpred\mstar\true);\, \acom\,\}
	\quad~\text{and}~\quad
	\stmtof{\theInterference} \,\defeq\, 
		\bigl(
			{\textstyle \mbox{\large\ensuremath{\sum}}}_{(\acpred, \acom)\in\theInterference} ~
			\commandof{\acpred, \acom}\ 
		\!\bigr)^*
	\ .
\]
We turn an interference $(\acpred, \acom)$ into an atomic block the execution of which is guarded by an assumption.
Recall that atomic blocks are not part of our programming constructs, but the above expression will be treated as a single command with the expected semantics.
The reason we need a single command is that $\commandof{\acpred, \acom}$ should abstract command $\acom$ in the program, and that command leads to a single state change.
Also note that $\acpred$ is a predicate from the assertion language that we deliberately use within an assumption.
To be closer to programming practice, one can weaken $\acpred$ to information about the current state that can be checked over the program variables. 
We use $\acpred\mstar\true$ rather than $\acpred$ to make sure the command satisfies~\eqref{cond-loccom}.
We also call $\commandof{\acpred, \acom}$ a \emph{self-interference}.
Function~$\stmtof{-}$ lifts the construction to a set of interferences.
The resulting program repeatedly executes all self-interferences in random order. 

The new predicate $\governeddefnb$ describes the set of \emph{$\theInterference$-governed} computations, the computations in which every state change is due to an interference or a self-interference:
\[
	\governeddefnb \quad\defeq\quad \csemof{\stmtof{\theInterference}}_{\theInterference}(\Sigma)
	\ .
\]
We view here $\setstates$ as a set of computations that consist of a single state.
With this definition, we intend to replace Inclusion~\eqref{Equation:TINaive} by
\begin{align}
	\weakpastof{\apred} \,\cap\, \nowof{\apredp} \,\cap\, \governeddefnb \quad\subseteq\quad \weakpastof{\apredpp}
	\ .
	\label{Equation:TIRelative}
\end{align}
\looseness=-1
This inclusion may or may not hold depending on the set of interferences. 
To prove the inclusion for the set of interferences at hand, we define Hoare triples that take a set of interferences as a parameter. 
We justify the need for this parameterization in moment. 
A so-called \emph{hypothesis} $\ahyp$ has the form \[
	\theInterferenceVar\semcalc\hoareOf{\acpred}{\stmtof{\theInterferenceVar}}{\acpredp}
	\ .
\]
Variable $\theInterferenceVar$ will be evaluated by a set of interferences.
The hypothesis is said to \emph{hold for $\theInterference$}, denoted by $\hypholdsof{\theInterference}{\ahyp}$, if we can prove the Hoare triple with $\theInterferenceVar$ replaced by $\theInterference$:
there is a set of predicates $\thePredicates$ with $\acpred\subseteq\acpred'\in \thePredicates$ so that $\thePredicates, \theInterference\semcalc\hoareOf{\acpred'}{\stmtof{\theInterference}}{\acpredp}$ is derivable and $\isInterferenceFreeOf{\thePredicates}$. 
We elaborate on the weakening of $\acpred$ to $\acpred'$ further below.
For a set of hypothesis~$\theHyp$, we write $\hypholdsof{\theInterference}{\theHyp}$ to mean $\hypholdsof{\theInterference}{\ahyp}$ for all $\ahyp\in\theHyp$.

The hypotheses we are interested in have the shape \[
	\theInterferenceVar\semcalc\hoareOf{\nowof{\apred}}{\stmtof{\theInterferenceVar}}{\nowof{\apredp}\rightarrow\weakpastof{\apredpp}}
	\ .
\]
Since the shape is fixed, we write the hypothesis as $\weakhypof{\apred}{\apredp}{\apredpp}$.
It states that from a computation ending in $\apred$, every execution of the interferences and the self-interferences that leads to a state from $\nowof{\apredp}$ satisfies $\weakpastof{\apredpp}$.
This is precisely the information that has been missing to justify Inclusion~\eqref{Equation:TIRelative}.

\begin{lemma}
	\label{Lemma:ATISound}
	If $\,\hypholdsof{\theInterference}{\weakhypof{\apred}{\apredp}{\apredpp}}$,
	then $\weakpastof{\apred}\cap \nowof{\apredp}\cap \governeddef\subseteq\weakpastof{\apredpp}$.
\end{lemma}


\begin{figure*}[t]
	\newcommand{\semcalcticol}{\semcalcti[\setColorLogic]}
	\def\MathparLineskip{\lineskip=2mm}
	\small
	\begin{mathpar}
		\inferH{com-ti}{
			\csemof{\acom}(\acpred)\subseteq \acpredp 
		}{
			\setcompact{\acpredp}, \set{(\acpred, \acom)}, \makeColorLogic{\emptyset} \semcalcticol\hoareOf{\acpred}{\acom}{\acpredp}
		}
		\and
		\makeColorLogic{
			\inferH{temporal-interpolation}{
				\apred, \apredp\text{ intuitionistic}
				\quad~~
				\theInterference=\set{(\acpred, \cskip)}
				\quad~~
				\theHyp = \set{\weakhypof{\apred}{\apredp}{\apredpp}}
				}{
				\set{\acpred\cap\pastof{\apredpp}}, \theInterference, \theHyp \semcalcticol\hoareOf{\acpred\cap\weakpastof{\apred}\cap\nowof{\apredp}}{\cskip}{\acpred\cap\pastof{\apredpp}}
			}
		}
		\and
		\makeColorLogic{
		\inferH{temporal-interpolation-unordered}{
			\apred, \apredp\text{ intuitionistic}
			\\\\
			\theInterference=\set{(\acpred, \cskip)}
			\\
			\theHyp = \set{\weakhypof{\apred}{\apredp}{\apredpp}, \weakhypof{\apredp}{\apred}{\apredpp}}
			}{
			\set{\acpred\cap\pastof{\apredpp}}, \theInterference, \theHyp\semcalcticol\hoareOf{\acpred\cap\weakpastof{\apred}\cap\weakpastof{\apredp}}{\cskip}{\acpred\cap\pastof{\apredpp}}
		}
		}
		\and
		\inferH{consequence-ti}{
		 \thePredicates', \theInterference', \makeColorLogic{\theHyp'} \semcalcticol\hoareOf{\acpred'}{\astmt}{\acpredp'}
			\quad~~\:
			\acpred\subseteq \acpred'\!\!
			\\\\
			\thePredicates' \subseteq \thePredicates
			\quad~
			\theInterference'\subseteq \theInterference
			\quad~
			\makeColorLogic{\theHyp'\subseteq \theHyp}
			\quad~\:
			\acpredp'\!\subseteq \acpredp
			}{
			\thePredicates, \theInterference, \makeColorLogic{\theHyp}\semcalcticol\hoareOf{\acpred}{\astmt}{\acpredp}
		}
		\and
		\inferH{frame-ti}{
			\thePredicates, \theInterference, \makeColorLogic{\theHyp}\semcalcticol\hoareOf{\acpred}{\astmt}{\acpredp}
		}{
			\thePredicates\mstar\acpredpp, \theInterference\mstar\acpredpp, \makeColorLogic{\theHyp}\semcalcticol\hoareOf{\acpred\mstar \acpredpp}{\astmt}{\acpredp\mstar\acpredpp}
		}
		\and
		\inferH{seq-ti}{
			\thePredicates_1, \theInterference_1, \makeColorLogic{\theHyp_1}\semcalcticol\hoareOf{\acpred}{\astmt_1}{\acpredp}\\
			\thePredicates_2, \theInterference_2, \makeColorLogic{\theHyp_2}\semcalcticol\hoareOf{\acpredp}{\astmt_2}{\acpredpp}
		}{
			\setcompact{\acpredp}\cup\thePredicates_1\cup\thePredicates_2, \theInterference_1\cup\theInterference_2, \makeColorLogic{\theHyp_1\cup\theHyp_2}\semcalcticol\hoareOf{\acpred}{\astmt_1;\astmt_2}{\acpredpp}
		}
		\and
		\inferH{loop-ti}{
			\thePredicates, \theInterference, \makeColorLogic{\theHyp}\semcalcticol\hoareOf{\acpred}{\astmt}{\acpred}
		}{
			\setcompact{\acpred}\cup\thePredicates, \theInterference, \makeColorLogic{\theHyp}\semcalcticol\hoareOf{\acpred}{\loopof{\astmt}}{\acpred}
		}\and
		\inferH{choice-ti}{
			\thePredicates_1, \theInterference_1, \makeColorLogic{\theHyp_1}\semcalcticol\hoareOf{\acpred}{\astmt_1}{\acpredp}\\
			\thePredicates_2, \theInterference_2, \makeColorLogic{\theHyp_2}\semcalcticol\hoareOf{\acpred}{\astmt_2}{\acpredp}
		}{
			\thePredicates_1\cup\thePredicates_2, \theInterference_1\cup\theInterference_2, \makeColorLogic{\theHyp_1\cup\theHyp_2}\semcalcticol\hoareOf{\acpred}{\choiceof{\astmt_1}{\astmt_2}}{\acpredp}
		}
	\end{mathpar}%
	\vspace{-6mm}\normalsize
	\caption{%
		Program logic from \citet{DBLP:journals/pacmpl/MeyerWW22} with our \makeColorLogic{extension} of hypotheses and temporal interpolation.
		We denote the former by $\semcalc$ and the latter by $\semcalcti$ (the subscript is short for ``temporal interpolation'').%
		\label{Figure:ProgramLogicTI}%
	}
	\vspace{-1mm}
\end{figure*}

We incorporate temporal interpolation into the separation logic presented in \Cref{sec:Preliminaries} by means of the new proof rule \ruleref{temporal-interpolation} given in \Cref{Figure:ProgramLogicTI}.
It draws a conclusion as in Equation~\eqref{Equation:TINaive} at the expense of recording a hypothesis $\weakhypof{\apred}{\apredp}{\apredpp}$.
There are a few things worth noting.
The rule does not expect the predicate $\governeddef$ to be present in the premise.
The soundness result will show that any program proof can be strengthend to maintain the set of governed computations, and we can therefore leave this set implicit.
We draw the conclusion after a $\cskip$ command, which turns the weak past predicate $\weakpastof{\apredpp}$ from the hypothesis into a proper past predicate $\pastof{\apredpp}$.
This is needed to harmonize the implicit treatment of $\governeddef$ with framing.
However, one can easily avoid the $\cskip$ by applying the rule to the preceding command.
The state predicates $\apred$ and $\apredp$ should be intuitionistic.
This is also related to framing.
Rule~\ruleref{temporal-interpolation-unordered} is a variant in which we do not know whether $\apred$ or $\apredp$ has been observed first and we rely on both hyptheses.

\looseness=-1
The hypotheses spawned by \ruleref{temporal-interpolation} have to be discharged against the full set of interferences collected for the overall program.
This is the reason why we work with hypotheses as parameterized Hoare triples rather than ordinary Hoare triples: in the moment we interpolate, we do not yet know the full set of interferences.
Instead, we may only have a fraction of the program (and hence the interferences) at hand.
It is also the reason why the separation logic judgements given in \Cref{Figure:ProgramLogicTI} maintain a set $\theHyp$ of hypotheses, and the rules are modified to join these sets.
We are not allowed to forget a hypothesis while building up the correctness judgement for the overall program.

We elaborate on why we weaken $\acpred$ to $\acpred'$ in the definition of $\hypholdsof{\theInterference}{\ahyp}$.
The purpose of \ruleref{temporal-interpolation} is to derive $\weakpastof{\apredpp}$ from $\weakpastof{\apred}\cap\nowof{\apredp}$.
Typically, $\apred$ occurs within a weak past predicate, because it is not interference-free.
This means no interference-free set of predicates $\thePredicates$ can prove the hypothesis $\hoareOf{\nowof{\apred}}{\stmtof{\theInterferenceVar}}{\nowof{\apredp}\rightarrow\weakpastof{\apredpp}}$.
A way out would be to prove the hypothesis for a weaker predicate $\apred\subseteq \apred'$ and replace the predicate $\weakpastof{\apred}$ in the main proof by $\weakpastof{\apred'}$.
Unfortunately, the predicates $\apred$ that require temporal interpolation not only fail the interference freedom test, it also seems to be impossible to weaken them to interference-free state predicates.
All we can do is weaken them by introducing past information.
Consider the example of a distributed counter given in \Cref{sec:Overview}.
There, $\apred$ is the predicate $c.\m{l} \mapsto \nl \mstar c.\m{r} \mapsto \nr \land \nr \leq \nr'$.
We weaken it to the invariant $\aninvpred$ defined as $\exists \nl''\,\nr''.\; c.\m{l} \mapsto \nl'' \,\mstar\, c.\m{r} \mapsto \nr'' \land \nl \leq \nl'' \land \nl'' + \nr'' < \nl + \nr' \lor \weakpastof{(\counter(c, \nl + \nr'))}$.
Although we have $\nowof{\apred}\subseteq \aninvpred$, the invariant does not have the shape $\nowof{\apred'}$.
This means the invariant does not lead to a hypothesis $\hypof{\apred'}{\apredp}{\apredpp}$ as required for temporal interpolation.
By weakening the condition of when $\hypof{\apred}{\apredp}{\apredpp}$ holds, we bridge the gap between $\nowof{\apred}$ and $\aninvpred$.

Hypotheses require an ordinary program proof, using a method of choice.
Yet, their shape suggests an invariance-based proof strategy: since program $\stmtof{\theInterference}$ repeats self-interferences $\commandof{\acpred, \acom}$, it suffices to find a predicate that is stable under these commands, contains the precondition, and entails the postcondition.
Call $\aninvpred\subseteq\setstates^+$ an \emph{inductive invariant for $\theInterference$} if $\csemof{\commandof{\acpred, \acom}}(\aninvpred)\subseteq\aninvpred$ for all $(\acpred, \acom)\in\theInterference$ and $\isInterferenceFreeOf[\theInterference]{\aninvpred}$.
We say that $\aninvpred$ \emph{proves $\weakhypof{\apred}{\apredp}{\apredpp}$}, if $\nowof{\apred}\subseteq\aninvpred$ and $\aninvpred\cap\nowof{\apredp}\subseteq \weakpastof{\apredpp}$.

\begin{lemma}[Strategy]
	\label{Lemma:ProofStrategy}
	Let $\aninvpred$ be an inductive invariant for $\theInterference$ proving $\weakhypof{\apred}{\apredp}{\apredpp}$.
	Then $\hypholdsof{\theInterference}{\weakhypof{\apred}{\apredp}{\apredpp}}$. 
\end{lemma}

\subsection{Soundness}
We show that every proof in the new program logic of \Cref{Figure:ProgramLogicTI} gives rise to a proof in the program logic of \Cref{sec:Preliminaries}, provided the hypotheses hold for the overall set of interferences.
Also successful interference freedom checks will carry over.
This means we can take full advantage of temporal interpolation, 
trusting that a traditional program proof will exist which discharges all hypotheses along the way.
Temporal interpolation can therefore be understood as a way of structuring and shortening traditional program proofs that involve temporal reasoning.
Technically, soundness shows that any derivation in the new program logic can be strengthened by an intersection with $\governeddefnb$.
This allows us to replace \ruleref{temporal-interpolation} by \ruleref{consequence-ti} relying on \Cref{Lemma:ATISound}.

\begin{theorem}[Soundness]
	\label{Theorem:Soundness}
	Consider a derivation $\thePredicates, \theInterference, \theHyp \semcalcti\hoareOf{\acpred}{\astmt}{\acpredp}$ with $\acpred\prall\in\thePredicates$, $\isInterferenceFreeOf[\theInterference]{\thePredicates}$, and~$\hypholdsof{\theInterference}{\theHyp}$.  
	Then $\thePredicates\cap \governeddef, \theInterference \semCalc\hoareOf{\acpred\cap\governeddef}{\astmt}{\acpredp\cap \governeddef}$ with $\acpred\cap\governeddef\prall\in\thePredicates\cap\governeddef$ and $\isInterferenceFreeOf[\theInterference]{(\thePredicates\cap\governeddef)}$. 
\end{theorem}

The difficulty in proving the theorem is the interplay between the intersection we intend to add and the frame rule.
Therefore, our first step is to eliminate the frame rule and show that whenever a correctness statement can be derived, then it can be derived without \ruleref{frame-ti}.
Let $\semcalctinf$ denote the restriction of $\semcalcti$ that avoids \ruleref{frame-ti}. 

\begin{lemma}[\ruleref{frame-ti} elimination]
	\label{Lemma:FrameElimination}
	$\thePredicates, \theInterference, \theHyp \semcalcti\hoareOf{\acpred}{\astmt}{\acpredp}$ iff $\,\thePredicates, \theInterference, \theHyp \semcalctinf\hoareOf{\acpred}{\astmt}{\acpredp}$.
\end{lemma}

At the heart of the lemma is the fact that the frame rule commutes with the remaining rules of the program logic.
This allows us to organize proofs in such a way that the frame rule is applied right after \ruleref{com-ti}.
A combination of \ruleref{com-ti} and \ruleref{frame-ti}, in turn, can be captured by \ruleref{com-ti} alone.
The difficult case is \ruleref{temporal-interpolation}, for the proof of which we rely on the following identity.

\begin{lemma}\label{Lemma:InterplayCapStarPast}
	$\acpredp\mstar\acpredpp\cap \pastof{\apredpp}=(\acpredp\cap \pastof{\apredpp})\mstar\acpredpp$. 
\end{lemma}

With the previous result, the derivation that makes use of temporal interpolation can be assumed to be \ruleref{frame-ti}-free.
We now show that also \ruleref{temporal-interpolation} can be eliminated, provided we strengthen the correctness statement by the governed computations.

\begin{lemma}
	\label{Lemma:HistoryTracking}
	If $\,\thePredicates, \theInterference', \theHyp \semcalctinf\hoareOf{\acpred}{\astmt}{\acpredp}$ is derivable, then for all $\,\theInterference$ with $\,\theInterference'\subseteq \theInterference$ and $\hypholdsof{\theInterference}{\theHyp}$ we have $\thePredicates\cap \governeddef, \theInterference \semcalc \hoareOf{\acpred\cap\governeddef}{\astmt}{\acpredp\cap \governeddef}$. 
\end{lemma}

The previous lemmas allow us to prove \Cref{Theorem:Soundness}.
For interference freedom, note that the governed computations are interference-free, $\isInterferenceFreeOf[\theInterference]{\governeddef}$, and we have $\isInterferenceFreeOf[\theInterference]{\thePredicates}$ by the assumption.
The intersection of two interference-free predicates is interference-free.


\section{Temporal Interpolation for Linearizability}
\label{sec:linearizability}

We present an extension of our program logic from \Cref{Section:TemporalInterpolation} to verify linearizability.
The approach is akin to atomic triples \cite{DBLP:conf/ecoop/PintoDG14}, except that we do not aim to support compositional reasoning about clients against atomic specifications of libraries. Instead, we only focus on verifying library implementations. We use update tokens that encode a method's obligation to execute a linearization point.
Once the method executes a command that resembles the linearization point, the update token is traded into a receipt token certifying successful linearization.
This also prevents the method from having further linearization points since tokens are not duplicable and thus no more tokens can be traded.
Here, we focus on concurrent search structures (CSS), however, the approach applies more generally. Sequential specifications $\asspec$ of concurrent search structure methods $\absop$ and key $k$ take the following form: \[
	\asspec~=~\hoareof{ \abscontent.~ \acss(\abscontent) }{ \absop(k) }{ v.~ \exists \abscontentp.~ \acss(\abscontentp) \mstar \acssup(\abscontent,\abscontentp,k,v) }
	\ .
\]
Here, $\abscontent$ and $\abscontentp$ are the logical contents of the structure before and after the operation takes effect.
The predicate $\acss(\abscontent)$ ties the physical state of the structure to $\abscontent$.
How the method call $\absop(k)$ changes the contents is prescribed by the relation $\acssup(\abscontent,\abscontentp,k,v)$.

The linearizability obligation is denoted by $\anobl{\asspec}$ and the receipt token by $\aful{\asspec}{v}$, 
and we drop $\asspec$ if it is clear from the context.
Receipts are parameterized in the result value of the operation to reconcile the actual return value with the one prescribed by $\asspec$.
For concurrent search structures, the sequential specifications of the methods \code{contains($\vk$)}, \code{insert($\vk$)}, and \code{delete($\vk$)} are as expected and we denote their obligations by $\OBLc{\vk}$, $\OBLi{\vk}$, and $\OBLd{\vk}$ (their receipts are just $\FULcid{v}$).


\begin{figure}
	\def\MathparLineskip{\lineskip=2.5mm}
	\small
	\begin{mathpar}
		\inferH{com-lin-void}{
			\thePredicates, \theInterference, \theHyp
			\semcalcti
			\hoareOf{\acpred}{\acom}{\acpredp}
			\\\\
			\acpred\subseteq\acss(\abscontent)
			\\
			\acpredp\subseteq\acss(\abscontent)
		}{
			\thePredicates, \theInterference, \theHyp
			\semcalclin
			\hoareOf{\acpred}{\acom}{\acpredp}
		}
		\and
		\inferH{com-lin-pure}{
			\acpred\subseteq\pastOF{\acss(\abscontent)\mstar\acssup(\abscontent,\abscontent,y,v)}
			\\\\
			\thePredicates = \aful{v}\mstar\acpred
			\\
			\theInterference = \decori{\set{(\acpred,\cskip)}}{\anobl}{\aful{v}}
		}{
			\thePredicates, \theInterference, \theHyp
			\semcalclin
			\hoareOf{\anobl\mstar\acpred}{\cskip}{\aful{v}\mstar\acpred}
		}
		\and
		\inferH{com-lin-impure}{
			\acpred\subseteq\acss(\abscontent)
			\\
			\thePredicates, \theInterference, \theHyp
			\semcalcti
			\hoareOf{\acpred}{\acom}{\acpredp}
			\\
			\acpredp\subseteq\acss(\abscontentp)\cap\acssup(\abscontent,\abscontentp,y,v)
		}{
			\thePredicates\mstar\aful{v},~
			\decori{\theInterference}{\anobl}{\aful{v}},~
			\theHyp
			\semcalclin
			\hoareOf{\anobl\mstar\acpred}{\acom}{\aful{v}\mstar\acpredp}
		}
	\end{mathpar}
	\vspace{-5mm}\normalsize
	\caption{%
		Proof rules for commands that ensure proper handling of the linearizability tokens $\anobl$ and $\aful{}$.%
		\label{fig:linearizability-rules}%
	}
	\vspace{-2pt}
\end{figure}

To deal with the tokens in a proof, we lift the proof system $\semcalcti$ from \Cref{Section:TemporalInterpolation} to a new proof system $\semcalclin$ which inherits all the rules of $\semcalcti$ except for Rule~\ruleref{com-ti}.
Rule~\ruleref{com-ti} is replaced by the three new rules from \Cref{fig:linearizability-rules}.
The rules extract the tokens, invoke $\semcalcti$, and then add the tokens back. However, in the process, they potentially transform the tokens if a linearization point is registered.
That is, the updates of tokens are handled by $\semcalclin$ rather than $\semcalcti$. 
To do this, we lift the program semantics $\csemOf{\acom}$ in a trivial way: the ghost component of the state is simply ignored.
However, for temporal interpolation to remain sound, we need to capture the effect of ghost state updates in the interferences.
So, we decorate commands $\decorc{\acom}{\anobl}{\aful{v}}$. 
Then, decorating an interference $(\acpred,\acom)$ decorates the command and adds the required token to the premise, $\decori{(\acpred,\acom)}{\anobl}{\aful{v}}=(\acpred\mstar\anobl,\decorc{\acom}{\anobl}{\aful{v}})$. 
With this, we are ready for the proof rules of $\semcalclin$.

\looseness=-1
Rule~\ruleref{com-lin-void} deals with commands that do not alter the logical contents of the structure.
Consequently, they maintain the current obligation/receipt token.
Rule~\ruleref{com-lin-impure} trades an obligation for a receipt if the executed command is the linearization point, that is, if it updates the logical contents of the structure according to the sequential specification.
If a command changes the logical contents but does not satisfy the specification or has no obligation token, the proof fails.
Rule~\ruleref{com-lin-pure} also trades an obligation for a receipt.
However, the rule does so in hindsight.
That is, there is no need to perform the trade at the very moment the sequential specification is satisfied, it can be done later if a past predicate can certify the existence of the linearization point.
It is this rule that sets our approach apart from atomic triples \cite{DBLP:conf/ecoop/PintoDG14}.
We allow for this retrospective linearization only if the linearization point is pure, i.e., does not alter the logical contents of the structure.
The reason is this: such pure linearization points can be used by arbitrarily many threads to linearize whereas impure linearization points require a one-to-one correspondence to threads.
The approach can be extended to support impure linearization points.
\techreport{We discuss this in \Cref{sec:RDCSS} and demonstrate it in a proof for the RDCSS data structure~\cite{DBLP:conf/wdag/HarrisFP02}.}

\begin{theorem}
	\label{thm:proof-implies-linearizability}
	If there are $\thePredicates,\theInterference,\theHyp$ with
	$\thePredicates,\theInterference,\theHyp\semcalclin\hoareOf{\acss(\bullet)\mstar\anobl{\asspec}}{\astmt}{\vres.~\acss(\bullet)\mstar\aful{\asspec}{\vres}}$ and
	$\acss(\bullet)\mstar\anobl{\asspec}\in\thePredicates$ and
	$\isInterferenceFreeOf[\theInterference]{\thePredicates}$ and
	$\hypholdsof{\theInterference}{\theHyp}$, then
	$\astmt$ is linearizable wrt. $\asspec$.
\end{theorem}


\section{Case Study: the LO-Tree}
\label{sec:lotree}

We substantiate the usefulness of the developed program logic by verifying the linearizability of a challenging concurrent data structure: the the logical-ordering (LO-)tree~\cite{DBLP:conf/ppopp/DrachslerVY14}.
We identify and fix bugs in the original implementation from \citet{DBLP:conf/ppopp/DrachslerVY14} as well as in the correction attempt by \citet{DBLP:journals/pacmpl/FeldmanKE0NRS20}.


\subsection{The LO-Tree in a Nutshell}
\label{sec:lotree:overview}

\smartparagraph{Overview}
The LO-tree~\cite{DBLP:conf/ppopp/DrachslerVY14} is a self-balancing binary search tree implementing a set data type.
Self-balancing refers to the tree periodically restructuring itself to maintain a low height in order to speed up accesses.
The restructuring mechanism in the LO-tree are standard tree rotations.
For an example rotation consider \Cref{fig:lotree:rotation}.
There, node $13$ experiences a \emph{right rotation}: its left child $7$ takes the position of node $13$ and node $13$ becomes the right child of $7$.
The formerly right subtree of $7$ becomes the left subtree of $13$.
The resulting tree is a binary search tree again.


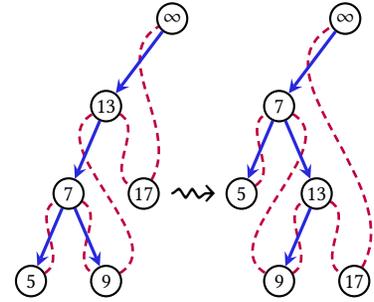
\begin{wrapfigure}[14]{r}{5cm}
	\vspace{-2.5mm}
	\begin{minipage}{2cm}
		\begin{tikzpicture}[trees,baseline=(node17.center)]
			\node (root) [] {}
				child {
					node (node13) [treenode] {$13$} edge from parent[draw=none]
						child [treeptr] {
							node (node7) [treenode] {$7$} 
								child { node (node5) [treenode] {$5$}}
								child { node (node9) [treenode] {$9$}}
						}
						child [treeptr] { node (node17) [treenode] {$17$} edge from parent[draw=none]}
				}
			;
			\node (rootN) at (root) [treenode,xshift=11mm] {$\infty$};
			\draw[treeptr] (rootN) -- (node13);
			\begin{scope}[on background layer]
				\draw[listptr] (node5) edge[out=30,in=210] (node7);
				\draw[listptr] (node7) edge[out=-30,in=150] (node9);
				\draw[listptr] (node9) edge[out=30,in=210] (node13);
				\draw[listptr] (node13) edge[out=-30,in=150] (node17);
				\draw[listptr] (node17) edge[out=60,in=210] (rootN);
			\end{scope}
		\end{tikzpicture}
	\end{minipage}
	\hspace{-.5mm}
	\raisebox{-18pt}{\bigleadsto}
	\hspace{-2mm}
	\begin{minipage}{2.4cm}
		\begin{tikzpicture}[trees,baseline=(node5.center)]
			\node (root) [] {}
				child {
					node (node7) [treenode] {$7$} edge from parent[draw=none]
						child [treeptr] {
							node (node5) [treenode] {$5$}}
						child [treeptr] { node (node13) [treenode] {$13$} 
							child { node (node9) [treenode] {$9$}}
							child { node (node17) [treenode] {$17$} edge from parent[draw=none]}}
				}
			;
			\node (rootN) at (root) [treenode,xshift=11mm] {$\infty$};
			\draw[treeptr] (rootN) -- (node7);
			\begin{scope}[on background layer]
				\draw[listptr] (node5) edge[out=30,in=210] (node7);
				\draw[listptr] (node7) edge[out=-30,in=150] (node9);
				\draw[listptr] (node9) edge[out=30,in=210] (node13);
				\draw[listptr] (node13) edge[out=-30,in=150] (node17);
				\draw[listptr] (node17) edge[out=60,in=210] (rootN);
			\end{scope}
		\end{tikzpicture}
	\end{minipage}
	\caption{%
		A right rotation of the node storing $13$.
		While the \textcolor{colorTree}{tree layout} changes, the \textcolor{colorList}{logical ordering} remains unaffected.
		\label{fig:lotree:rotation}
	}
\end{wrapfigure}

In a concurrent setting, rotations pose a major challenge.
To avoid performance bottlenecks, one wishes to traverse the tree without synchronization, e.g., without acquiring locks that prevent rotations from happening.
Without synchronization, however, one cannot prevent traversals to \emph{go astray} in the presence of rotations.
In \Cref{fig:lotree:rotation}, if a tree traversal searching for node $5$ arrives at node $13$ and node $13$ experiences the right rotation before the tree traversal continues, then the tree traversal will never reach node $5$ but end up at node $9$.
For the implementation to be linearizable, it must detect this and be able to find node $5$ despite the rotation.

The LO-tree solves the problem by organizing the nodes in a doubly-linked list, the eponymous logical ordering.
In fact, it is this list which dictates the contents of the LO-tree.
The tree structure is merely an overlay to that list which helps to speed up accesses.
In \Cref{fig:lotree:rotation}, the logical ordering \textlistptr contains all nodes in ascending order while the tree overlay \texttreeptr does not yet contain node $17$.
Hence, the previous tree traversal, which arrives at node $9$ on its way to node $5$, can follow the logical ordering backward to find $5$.
Similarly, a tree traversal searching for $17$ arrives at node $13$ and then follows the logical ordering forward to find it.

\smartparagraph{Implementation}
We link the above ideas to the implementation of the LO-tree in \Cref{fig:lotree:impl}~(ignore the \textcolor{teal}{proof outline annotations} for now).
The nodes of the tree are represented by the struct type \code{Node}.
Each node stores an integer $\keysel$ and a Boolean $\marksel$ as well as several pointers and locks.
The $\marksel$ field is used to indicate that the node is being or has been removed from the tree.
For the doubly-linked logical ordering list each node stores a forward $\succsel$ and a backward $\predsel$ pointer.
To synchronize mutations of the list, there is a lock $\listsel$.
For the tree overlay, each node stores pointers $\leftsel$ and $\rightsel$ to its children and a pointer $\parentsel$ to its parent.
Tree mutations are synchronized with a lock $\treesel$.
There are two sentinel nodes $\ptrmin$ resp. $\ptrmax$ storing values $-\infty$ resp. $\infty$.
The initial logical ordering consists of these two nodes.
The root of the tree is $\ptrmax$.

The user-facing API of the LO-tree consists of the three methods of a concurrent search structure: \code{contains}, \code{insert}, and \code{delete}.
The methods return a Boolean indicating success of the operation.
Methods \code{insert} and \code{remove} use fine-grained locking to synchronize mutators.
Both methods rely on the helper method \code{locate($\vk$)} which finds (and locks) the position in the logical ordering to which value $\vk$ belongs.
This position can be thought of as the interval between two successive nodes~$\pprev$ and~$\pnext$, $\succof{\pprev}=\pnext$, so that $\vk$ is logically ordered between the two or in $\pnext$, $\vk\in(\keyof{\pprev},\keyof{\pnext}]$.
To arrive at this location, a straightforward binary tree traversal is used, as implemented by \code{traverse($\vk$)}.
Since the traversal may yield $\pprev$ or $\pnext$ depending on the tree structure, the remaining node is determined using $\predsel$/$\succsel$ of the logical ordering.
To account for the tree traversal going astray due to rotations, \code{locate} validates the found position.
More precisely, it checks for $\vk\in(\keyof{\pprev},\keyof{\pnext}]$ and ensures that $\pprev$ is unmarked, i.e., still part of the logical ordering.
The validation happens after locking $\listsel$ of $\pprev$ so that the position cannot be invalidated by concurrent mutators.

Insertions of value $\vk$ proceed as follows.
They first \code{locate} the position $\pprev,\pnext$ in the logical ordering where $\vk$ should be inserted.
The returned position also reveals whether $\vk$ is already present in the logical ordering.
If so, the insertion fails and returns $\false$.
Otherwise, a new node $\pnode$ is inserted in between $\pprev$ and $\pnext$.
The new node's $\predsel$ and $\succsel$ are pointed to $\pprev$ and $\pnext$, respectively.
Then, $\pnode$ is inserted into the logical ordering.
It is first inserted into the forward ordering by pointing $\succof{\pprev}$ to $\pnode$.
Only after this, it is inserted into the backward ordering by pointing $\predof{\pnext}$ to $\pnode$.
This order deviates from the original version \cite{DBLP:conf/ppopp/DrachslerVY14} for reasons we explain in \Cref{sec:lotree:bugs}.
Finally, $\pnode$ is inserted into the tree by a call to \code{performTreeInsertion($\pnode$, $\pparent$)}.
This call expects the node $\pparent$ that is the parent of $\anode$.
The parent $\pparent$ is determined before $\anode$ is inserted into the logical ordering by \code{prepareTreeInsertion($\pprev$, $\pnext$)}, which does not alter the logical ordering nor the tree but may acquire locks.
We do not got into the details of the tree modifications as they are orthogonal to our linearizability proof.
Finally, $\true$ is returned by \code{insert}.

Deletions of value $\vk$ are similar to insertions.
They \code{locate} the position $\pprev,\pcurr$ where $\vk$ resides.
If $\keyof{\pcurr}\neq\vk$, then $\vk$ is not present and the deletion fails, returning $\false$.
Otherwise, it acquires $\pcurr$'s $\listsel$ and reads $\pcurr$'s  successor $\pnext$.
To remove $\pcurr$, it is marked by setting $\markof{\pcurr}=\true$, unlinked from the backward logical ordering by setting $\predof{\pnext}=\pprev$, and then unlinked from the forward logical ordering by setting $\succof{\pprev}=\pnext$.
Afterwards, $\pcurr$ is removed from the tree using \code{performTreeDeletion($\pcurr$)} which expects \code{prepareTreeDeletion($\pcurr$)} has been called before $\pcurr$ was marked.
Similar to insertions, \code{prepareTreeDeletion} does not alter the logical ordering nor the tree but may acquire locks.
Again, we elide \code{performTreeDeletion} and \code{prepareTreeDeletion} as they are unimportant for our discussion.

Unlike the above mutations, the \code{contains($\vk$)} method is wait-free, in particular it does not acquire locks.
It traverses the tree, follows $\predsel$ pointers, and finally follows $\succsel$ pointers to check whether there is an unmarked node containing $\vk$.
In addition to the original version \cite{DBLP:conf/ppopp/DrachslerVY14}, we need to follow $\predsel$ pointers at least until the first unmarked node to guarantee that $\vk$ is found indeed, see \Cref{sec:lotree:bugs}.


\begin{figure}[p]
	\def\columnWidthLeft{.48\textwidth}
	\def\columnWidthRight{.48\textwidth}
	\def\columnLeftSkip{6.6pt}
	\def\columnRightSkip{6pt}
	\newcommand{\MYMSTAR}{{}\mstar{}}
	\newcommand{\MYMAND}{{}\mand{}}
	\newcommand{\mkNL}{\\&\hspace{1.2em}}
	\newcommand{\travInvOf}[1]{\inv(\abscontent,\somenodes\cup\anodeval) \MYMSTAR #1=\anodeval}
	\newcommand{\succInvOf}[1]{\succInv(\abscontent,\abscontentp\!\!,\somenodes,\somenodesp,#1,\anodeval)}
	\lstdefinelanguage{xcompactSPL}{language=compactSPL,numbersep=3pt,xleftmargin=1.05em}
%
%
\begin{lstlisting}[language=xcompactSPL,numbers=none,xleftmargin=.25em,belowskip=-2pt]
struct Node { int key; bool mark; Lock treeLock, listLock; Node* left, right, parent, pred, succ;$\;$}
val min = new Node { key = -$\infty$; mark = false; }; val max = new Node { key = $\infty$; mark = false; }
min.pred, min.succ := max, max; max.pred, max.succ := min, min
\end{lstlisting}%
	\begin{minipage}[t]{\columnWidthLeft}
%
%
\begin{lstlisting}[language=xcompactSPL,numbers=none,xleftmargin=.25em,belowskip=\columnLeftSkip]
$\makeTeal{\begin{aligned}[t]&
	\locateInv(\abscontent,\somenodes,\pprev,\pnext) \defeq
		\inv(\abscontent,\somenodes\prall{\cup}\pprev\prall{\cup}\pnext) \MYMSTAR \keyvarof{\pprev}\prall{<}\vk\prall{\leq}\keyvarof{\pnext}
		\mkNL
		\MYMSTAR \holds{\pprev} \MYMSTAR \vk\prall{\in}\keysetof{\pnext} \MYMSTAR
		\pnext=\succvarof{\pprev} \MYMSTAR \neg\markvarof{\pprev}
\end{aligned}}$
$\makeTeal{\begin{aligned}[t]&
	\linkInv(\abscontent,\somenodes,\pprev,\pcurr,\pnext) \defeq
		\inv(\abscontent,\somenodes\cup\pprev\cup\pcurr\cup\pnext) \MYMSTAR \insetof{\pprev}\neq\emptyset
		\mkNL\MYMSTAR
		\holds{\pprev} \MYMSTAR \holds{\pcurr} \MYMSTAR \pcurr=\succvarof{\pprev} \MYMSTAR \pnext=\succvarof{\pcurr}
		\mkNL\MYMSTAR
		\neg\markvarof{\pprev} \MYMSTAR \keyvarof{\pprev}\prall{<}\vk\prall{=}\keyvarof{\pcurr}\prall{<}\keyvarof{\pnext}
\end{aligned}}$
$\makeTeal{\begin{aligned}[t]&
	\succInv(\abscontent,\abscontentp\!,\somenodes,\somenodesp,\anode,\anodeval) \defeq
		\inv(\abscontent,\somenodes\cup\anodeval) \MYMSTAR \anode=\anodeval
		\mkNL\MYMAND
		\pastOF{ \old{\inv(\abscontentp\!\!,\somenodesp\cup\anodeval)} \MYMSTAR \vk\in\old{\insetof{\anodeval}} }
\end{aligned}}$
\end{lstlisting}
%
%
\begin{lstlisting}[language=xcompactSPL,belowskip=\columnLeftSkip]
$\ANNOT{\exists \abscontent,\somenodes.~ \inv(\abscontent,\somenodes) }$
method `traverse'($\vk$: Int): Node {
	val $\pcurr$ = max
	while (true) { $\ANNOT{ \inv(\abscontent,\somenodes\cup\pcurr) }$
		val $\pchild$ = $\vk$ < $\pcurr$.key ? $\pcurr$.left : $\pcurr$.right
		if ($\pcurr$.key == $\vk$ || $\pchild$ == NULL) return $\pcurr$
		$\ANNOT{ \inv(\abscontent,\somenodes\cup\pchild) }$ $\pcurr$ := $\pchild$
}	}
$\ANNOT{ \pcurr.~\exists\abscontent,\somenodes.~ \inv(\abscontent,\somenodes\cup\pcurr) }$
\end{lstlisting}
%
%
\begin{lstlisting}[language=xcompactSPL,belowskip=\columnLeftSkip]
$\ANNOT{\exists \abscontent,\somenodes.~ \OBLc{\vk} \MYMSTAR \inv(\somenodes) \MYMSTAR -\infty < \vk < \infty }$
method `contains'($\vk$: Int): Bool {
	val $\pcurr$ = traverse($\vk$)   $\label{code:lotree:contains:treetraversal}$
	$\ANNOT{ \OBLc{\vk} \MYMSTAR \travInvOf{\pcurr} }$   $\label{code:lotree:contains:treetraversal-post}$
	while ($\vk$ < $\pcurr$.key) { val $\pprev$ = $\pcurr$.pred; $\pcurr$ := $\pprev$ } $\label{code:lotree:contains:go-pred}$
	$\ANNOT{ \OBLc{\vk} \MYMSTAR \travInvOf{\pcurr} \MYMSTAR \keyvarof{\anodeval}\leq\vk }$ $\label{code:lotree:contains:go-pred-post}$ |<
	while ($\makeBug{\pcurr}$.mark) { val $\makeBug{\pprev}$ = $\makeBug{\pcurr}$.pred; $\makeBug{\pcurr}$ := $\makeBug{\pprev}$ } >| $\label{code:lotree:contains:fix}\label{code:lotree:fix-contains}$
	$\ANNOT{ \OBLc{\vk} \MYMSTAR \succInvOf{\pcurr} }$   $\label{code:lotree:contains:go-succ-pre}$
	while ($\pcurr$.key < $\vk$) { val $\pnext$ = $\pcurr$.succ; $\pcurr$ := $\pnext$ } $\label{code:lotree:contains:go-succ}$
	$\ANNOT{ \OBLc{\vk} \MYMSTAR \succInvOf{\pcurr} \MYMSTAR \vk\leq\keyvarof{\anodeval} }$   $\label{code:lotree:contains:go-succ-post}$
	val $\vres$ = $\pcurr$.key == $\vk$ @\makeBug{\begingroup\def\ULthickness{.75pt}\sout{\textcolor{black}{\&\& !$\pcurr$.mark}}\endgroup}@ $\label{code:lotree:obsolete-check}\label{code:lotree:contains:validate}$
	$\ANNOTML{
		& \OBLc{\vk} \MYMSTAR \inv(\abscontent,\somenodes) \MYMSTAR \vres=\avalval 
		\\& \MYMSTAR 
		\pastOf{
			\old{\inv(\abscontentp\!\!,\somenodesp)} \MYMSTAR
			\avalval \Leftrightarrow \vk\prall{\in}\abscontentp
		}
	}$   $\label{code:lotree:contains:validate-annot}$
	$\ANNOT{ \FULc{\vres} \MYMSTAR \inv(\abscontent,\somenodes) }$ // hindsight   $\label{code:lotree:contains:final}$
	return $\vres$   $\label{code:lotree:contains:return}$
}
$\ANNOT{\vres.~ \exists \abscontent,\somenodesp.~ \FULc{\vres} \MYMSTAR \inv(\abscontent,\somenodes) }$
\end{lstlisting}
%
%
\begin{lstlisting}[language=xcompactSPL,belowskip=0pt]
$\ANNOT{\exists \abscontent,\somenodes.~ \inv(\abscontent,\somenodes) \MYMSTAR -\infty < \vk < \infty }$
method `locate'($\vk$: Int): Node * Node {
	val $\pcurr$ = traverse($\vk$)   $\label{code:lotree:locate:tree}$
	val $\pprev$ = $\pcurr$.key < $\vk$ ? $\pcurr$ : $\pcurr$.pred   $\label{code:lotree:locate:prev}$
	lock($\pprev$.listLock)   $\label{code:lotree:locate:lock}$
	val $\pnext$ = $\pprev$.succ   $\label{code:lotree:locate:next}$
	$\ANNOT{ \inv(\abscontent,\somenodes\cup\pprev\cup\pnext) \MYMSTAR \holds{\pprev} \MYMSTAR \pnext=\succvarof{\pprev} }$   $\label{code:lotree:locate:annot}$
	if ($\pprev$.key$\;$<$\; \vk \;$<=$\;\pnext$.key && !$\pprev$.mark) return $\pprev$,$\;\pnext$   $\label{code:lotree:locate:return}$
	unlock($\pprev$.listLock); restart   $\label{code:lotree:locate:restart}$
}
$\ANNOT{ \pprev,\pnext.~\exists\abscontent,\somenodes.~ \locateInv(\abscontent,\somenodes,\pprev,\pnext) }$   $\label{code:lotree:locate:post}$
\end{lstlisting}
	\end{minipage}
	\hfill
	\begin{minipage}[t]{\columnWidthRight}
%
%
\begin{lstlisting}[language=xcompactSPL,belowskip=\columnRightSkip]
$\ANNOT{\exists \abscontent,\somenodes.~ \OBLi{\vk} \MYMSTAR \inv(\abscontent,\somenodes) \MYMSTAR -\infty < \vk < \infty }$
method `insert'($\vk$: Int): Bool {
	val $\pprev$, $\pnext$ = locate($\vk$)   $\label{code:lotree:insert:locate}$
	if ($\pnext$.key == $\vk$) {   $\label{code:lotree:insert:check}$
		$\ANNOT{ \OBLi{\vk} \MYMSTAR \locateInv(\abscontent,\somenodes,\pprev,\pnext) \MYMSTAR \vk\in\abscontent }$
		$\ANNOT{ \FULi{\false} \MYMSTAR \locateInv(\abscontent,\somenodes,\pprev,\pnext) }$
		unlock($\pprev$.listLock); return false
	}
	$\ANNOT{ \OBLi{\vk} \MYMSTAR \locateInv(\abscontent,\somenodes,\pprev,\pnext) \MYMSTAR \keyvarof{\pnext}\neq\vk\notin\abscontent }$ $\label{code:lotree:insert:pre-preparetree}$
	val $\pparent$ = prepareTreeInsertion($\pprev$, $\pnext$)   $\label{code:lotree:insert:preparetree}$
	val $\pnode$ = new Node { key := $\vk$; mark := false;    $\label{code:lotree:insert:malloc}$
	        pred := $\pprev$; succ := $\pnext$; parent := $\pparent$ }
	$\ANNOTML{ &\OBLi{\vk} \MYMSTAR \locateInv(\abscontent,\somenodes,\pprev,\pnext) \MYMSTAR \hiddenIns{\pprev,\pnext} \\& \MYMSTAR \keyvarof{\pnode}=\vk\notin\abscontent }$ |<
	$\makeBug{\pprev}$.succ := $\makeBug{\pnode}$ >| // logical insertion   $\label{code:lotree:fix-insert-succ}$ $\label{code:lotree:insert:linksucc}$
	$\ANNOTML{ &\FULi{\true} \MYMSTAR \locateInv(\abscontent,\somenodes,\pprev,\pnode) \MYMSTAR \hiddenIns{\pprev,\pnext} \MYMSTAR \\& \keyvarof{\pnode}=\vk<\keyvarof{\pnext} }$ |<
	$\makeBug{\pnext}$.pred := $\makeBug{\pnode}$ >|   $\label{code:lotree:fix-insert-pred}$ $\label{code:lotree:insert:linkpred}$
	unlock($\pprev$.listLock)
	$\ANNOT{ \FULi{\true} \MYMSTAR \inv(\abscontent,\somenodes) \MYMSTAR \hiddenIns{\pprev,\pnext} }$
	performTreeInsertion($\pnode$, $\pparent$)$\label{code:lotree:insertIntoTree}$; return true
}
$\ANNOT{\vres.~ \exists\abscontent,\somenodesp.~ \FULi{\vres} \MYMSTAR \inv(\abscontent,\somenodes) }$
\end{lstlisting}
%
%
\begin{lstlisting}[language=xcompactSPL,belowskip=0pt]
$\ANNOT{\exists \abscontent,\somenodes.~ \OBLd{\vk} \MYMSTAR \inv(\abscontent,\somenodes) \MYMSTAR -\infty < \vk < \infty }$
method `delete'($\vk$: Int): Bool {
	val $\pprev$, $\pcurr$ = locate($\vk$)
	if ($\pcurr$.key != $\vk$) {   $\label{code:lotree:delete:check}$
		$\ANNOT{ \OBLd{\vk} \MYMSTAR \locateInv(\abscontent,\somenodes,\pprev,\pcurr) \MYMSTAR \vk\notin\abscontent }$
		$\ANNOT{ \FULd{\false} \MYMSTAR \locateInv(\abscontent,\somenodes,\pprev,\pcurr) }$
		unlock($\pprev$.listLock); return false
	}
	lock($\pcurr$.listLock)   $\label{code:lotree:delete:lock}$
	prepareTreeDeletion($\pcurr$)   $\label{code:lotree:delete:preparetree}$
	val $\pnext$ = $\pcurr$.succ
	$\ANNOT{ \OBLd{\vk} \MYMSTAR \linkInv(\abscontent,\somenodes,\pprev,\pcurr,\pnext) \MYMSTAR \vk\prall{\in}\abscontent \MYMSTAR \hiddenDel{\pcurr} }$
	$\pcurr$.mark := true   $\label{code:lotree:delete:mark}$
	$\pnext$.pred := $\pprev$   $\label{code:lotree:delete:unlinkpred}$
	$\ANNOT{ \OBLd{\vk} \MYMSTAR \linkInv(\abscontent,\somenodes,\pprev,\pcurr,\pnext) \MYMSTAR \vk\prall{\in}\abscontent \MYMSTAR \hiddenDel{\pcurr} }$
	$\pprev$.succ := $\pnext$ // logical deletion   $\label{code:lotree:delete:unlinksucc}$
	$\ANNOTML{ &\FULd{\true} \MYMSTAR \linkInv(\abscontent,\somenodes,\pprev,\pcurr,\pnext) \MYMSTAR \vk\notin\abscontent \\& \MYMSTAR \hiddenDel{\pcurr} }$
	unlock($\pcurr$.listLock); unlock($\pprev$.listLock)
	$\ANNOT{ \FULi{\true} \MYMSTAR \inv(\abscontent,\somenodes\cup\pcurr) \MYMSTAR \hiddenDel{\pcurr} }$
	performTreeDeletion($\pcurr$)$\label{code:lotree:removeFromTree}$; return true
}
$\ANNOT{\vres.~ \exists\abscontent,\somenodesp.~ \FULi{\vres} \MYMSTAR \inv(\abscontent,\somenodes\cup\pnode) }$
\end{lstlisting}
	\end{minipage}%
	\vspace{-2mm}
	\caption{%
		Implementation, \makeBug{bug fixes}, and \textcolor{teal}{linearizability proof outline} of the LO-tree~\cite{DBLP:conf/ppopp/DrachslerVY14}.
		The proof of \code{contains} requires hindsight reasoning to handle the future-dependent linearization point, see~\Cref{sec:lotree:proof-contains}.%
		\label{fig:lotree:impl}
	}
\end{figure}


\subsection{Bugs and their Fixes}
\label{sec:lotree:bugs}

The original version of the LO-tree \cite{DBLP:conf/ppopp/DrachslerVY14} has two bugs which we fixed in \Cref{fig:lotree:impl}.
\nontechreport{Due to space constraints, we refer to the technical report \cite{techreport} for more details.}
\techreport{See \Cref{appendix:lo-bugs-full} for more details.}

\smartparagraph{Bug 1: Duplicate Values}
A subtle quirk of the LO-tree is the fact that an insertion of value $\vk$ may be unaware of a concurrent deletion of $\vk$ because the tree traversal of the insertion experienced a rotation but still ended up in the right position for the insertion (the validation in \code{locate} succeeds).
Successful validation requires that the deletion already removed $\vk$ from the logical ordering.
So, the insertion can proceed and insert $\vk$ into the logical ordering and into the tree.
If the deletion has not yet removed the old marked version of $\vk$, then the tree contains two nodes with value $\vk$ that disagree on the mark bit.
Hence, rotations influence the result of \code{contains($\vk$)}---it is not linearizable.

Our implementation from \Cref{fig:lotree:impl} fixes the above problem by adding \Cref{code:lotree:fix-contains}: the logical ordering is followed backward ($\predsel$ fields) at least until an unmarked node is encountered.
This ensures that the final result is not \emph{confused} by concurrent deletions.
Other than that \code{contains} proceeds as originally devised by \citet{DBLP:conf/ppopp/DrachslerVY14}.
Interestingly, adding \Cref{code:lotree:fix-contains} renders the mark bit check on \Cref{code:lotree:obsolete-check} superfluous.

\smartparagraph{Bug 2: Insertion Order}
\looseness=-1
\citet{DBLP:journals/pacmpl/FeldmanKE0NRS20} identified another bug in the \code{insert} method.
In the original version \cite{DBLP:conf/ppopp/DrachslerVY14}, new nodes are inserted first into the backward logical ordering and then into the forward one (compared to \Cref{fig:lotree:impl}, \Cref{code:lotree:fix-insert-succ,code:lotree:fix-insert-pred} are reversed).
To see why this is problematic, assume an insertion of a new node $\pnode$ with value $\vk$ between nodes $\pprev$ and $\pnext$ already linked $\predof{\pnext}$ to $\pnode$ but $\succof{\pprev}$ is still pointing to $\pnext$.
Then, \code{contains($\vk$)} will find $\pnode$ only if the tree traversal takes it to nodes that appear after $\pnext$ in the logical order.
For earlier nodes, \code{contains} will only follow $\succsel$ fields which cannot yet reach $\pnode$.
It is easy to see that this violates linearizability.

We fixed this bug by changing the order in which $\pnode$ is linked into the logical ordering, cf. \Cref{code:lotree:fix-insert-succ,code:lotree:fix-insert-pred}.
\citet{DBLP:journals/pacmpl/FeldmanKE0NRS20} apply the same fix.\techreport{\footnote{%
	The code they give \cite[Figure 2]{DBLP:journals/pacmpl/FeldmanKE0NRS20} contains the erroneous linking order.
	Their proof arguments \cite[Case (i) on Page 18]{DBLP:journals/pacmpl/FeldmanKE0NRS20}, however, suggests that this is an oversight and is meant to be the correct linking order.
	This has been confirmed by one of the authors.
}}
However, they also change \code{insert} to link new nodes first into the tree overlay and then into the logical ordering (without modifying \code{contains}).
This violates linearizability:
if a new node $\pnode$ with value $\vk$ is inserted into the tree but not yet into the logical ordering, \code{contains} will find $\vk$ if and only if it is not affected by concurrent rotations.\footnote{%
	This is a mistake in the proof of the LO-tree by \citet{DBLP:journals/pacmpl/FeldmanKE0NRS20}. We do not make any claims regarding the soundness of their meta theory.
}


\subsection{Local Reasoning Principle}
\label{sec:lotree:locality}

\techreport{\smartparagraph{Local Reasoning}}
While our program logic from \Cref{sec:linearizability} tells us how to establish linearizability, it leaves us with a hard task: show that a command does or does not alter the contents of the structure.
The contents is defined inductively over the data structure graph.
To localize the reasoning about this inductive quantity, we build on the keyset framework \cite{DBLP:journals/tods/ShashaG88,DBLP:conf/pldi/KrishnaPSW20,DBLP:series/synthesis/2021Krishna}.

Suppose the global data structure graph consists of a set of nodes $\somenodes$.
We will define a predicate $\inv(\abscontent,\abskeyset,\somenodes,\somenodesp)$ that describes the resources and properties of a subregion $\somenodesp \subseteq \somenodes$ in the graph.
Here, $\abscontent$ will be the \emph{logical contents} of the subregion, which is the union of the logical contents $\contentsof{\anode}$ of all nodes $\anode \in \somenodesp$.
The set $\abskeyset$ is the \emph{keyset} of the region $\somenodesp$, which consists of all those keys that \emph{could be} in $M$.
We require the invariant to guarantee $\abscontent \subseteq \abskeyset$.
The keyset will be defined inductively over the graph structure as we explain below.
We then define the invariant $\acss(\abscontent)$ of the entire structure as follows:
$\acss(\abscontent) \defeq \exists\somenodes.~\inv(\abscontent,(-\infty,\infty),\somenodes,\somenodes)$.

To enable local reasoning, we aim for a definition of $\inv$ that yields the following compositionality: \[
	\inv(\abscontent,\abskeyset,\somenodes,\somenodesp\uplus\somenodesp')
	~\iff~
	\exists\,\abscontent_1,\abscontent_2,\abskeyset_1,\abskeyset_2.~\,
		\begin{aligned}[t]
			&\inv(\abscontent_1,\abskeyset_2,\somenodes,\somenodesp)\MSTAR\inv(\abscontent_2,\abskeyset_2,\somenodes,\somenodesp')
			\MSTAR{}\\
			&\abscontent=\abscontent_1\uplus\abscontent_2 \MSTAR \abskeyset=\abskeyset_1\uplus\abskeyset_2
	\ .
		\end{aligned}
\]
That is, the predicate allows us to decompose the graph arbitrarily into disjoint subregions $M$ and $M'$ and compose them back together. In particular, separating conjunction will guarantee that the keysets (and hence the logical contents) of disjoint subregions will also be disjoint.

For proofs, this means that we can focus our reasoning on appropriate fragments $\inv(\abscontent,\abskeyset,\somenodes,\somenodesp)$ with a small set $\somenodesp$.
When reasoning about updates we can focus on the fragment $\somenodesp$ that contains only those nodes whose fields or keysets change.
As we will see, three nodes will suffice to handle the LO-tree.
Also, $\inv$ enables a local-to-global lifting of the specification $\acssup(\abscontent,\abscontentp,\vk,v)$ of our search structure methods.
For example, if we have identified a fragment of the form $\inv(\abscontent,\abskeyset,\somenodes,\set{\anode})$ with $\vk\in\abskeyset$, then $\vk\in\abscontent=\contentsof{\anode}$ iff $\vk$ is in the logical contents of the entire structure.

\smartparagraph{Flows}
To obtain a definition of $\inv$ with the desired properties, we build on the flow framework~\cite{DBLP:journals/pacmpl/KrishnaSW18,DBLP:conf/esop/KrishnaSW20} which enables local reasoning about inductively-defined quantities of graphs. We sketch the main ideas for our specific application of the flow framework to keysets.

Each node is augmented with a ghost quantity called \emph{inset}.
Intuitively, the inset of a node $\anode$ is the set $\insetof{\anode}$ of all keys $\vk$, such that a thread searching for $\vk$ will traverse $\anode$.
That $\anode$ is traversed means that the search eventually considers $\anode$; the search may or may not continue from there.
The keyset $\keysetof{\anode}$ of $\anode$ is the subset of $\insetof{\anode}$ for which the traversal will terminate at $\anode$.
For the LO-tree, the inset is $\insetof{\ptrroot}=[-\infty,\infty]$ for the root node of the logical ordering and for the remaining nodes it is obtained as a solution to the following recursive equation:
\[
	\insetof{\anode} ~~=~~ \mbox{\Large\ensuremath{\bigcup_\anodep}} ~~ \ite{\succof{\anodep}=\anode\;}{\;\insetof{\anodep} \cap (\keyof{\anodep},+\infty]}{\emptyset}
	\ .
\]
We then define $\keysetof{\anode}=\insetof{\anode}\cap[-\infty,\keyof{\anode}]$.
The inset propagates via $\succsel$ links only, because it is the list of $\succsel$ links that makes up the logical contents of the LO-tree, as alluded to in \Cref{sec:lotree:overview}.
With this, we formally express the logical contents of node $\anode$ by $\contentsof{\anode}\defeq\set{\keyof{\anode}}\cap\keysetof{\anode}$.

\looseness=-1
To express insets in a separation algebra, the flow framework adds an additional ghost resource component.
The technical details are not relevant for our discussion.
In our proofs, we use the separation algebras proposed by \citet{DBLP:journals/pacmpl/MeyerWW22} and defer the interested reader there.
What is important here, is that the above definitions guarantee that the keysets of subregions are always disjoint.


\subsection{The Structural Invariant}
\label{sec:lotree:invariant}
We use standard separation logic assertions to represent the semantic predicates used so far.
\techreport{%
	In particular, we use boxed assertions $\boxed{A}$ to denote that $A$ is interpreted in the shared rather than the local state \cite{DBLP:conf/concur/VafeiadisP07,DBLP:phd/ethos/Vafeiadis08}.
	Moreover, we use fractional permissions~\cite{DBLP:conf/sas/Boyland03} for points-to predicates~$\fracto{1}{n}$ to allow reads but prevent interfering updates to lock-protected resources.
	We also use persistent points-to~\cite{DBLP:conf/cpp/VindumB21} predicates~$\persto$ to easily share knowledge about immutable fields.
}\nontechreport{%
	In particular, we use
	\begin{inparaenum}[(i)]
	 	\item boxed assertions $\boxed{A}$ to denote that $A$ refers to shared state \cite{DBLP:conf/concur/VafeiadisP07,DBLP:phd/ethos/Vafeiadis08},
	 	\item fractional permissions~\cite{DBLP:conf/sas/Boyland03} for points-to predicates~$\fracto{1}{n}$ to allow reads but prevent interfering updates to lock-protected resources, and
		\item persistent points-to~\cite{DBLP:conf/cpp/VindumB21} predicates~$\persto$ to easily share knowledge about immutable fields.
	\end{inparaenum}
}

We define a predicate $\nodeof{\anode}$ for the shared resources of a node $\anode$.
For simplicity, we assume that proofs are implicitly existentially closed.
This enables the naming convention where a use of $f(\anode)$ in the outer proof context refers to the value of field \code{f} as defined within $\nodeof{\anode}$.
We define:
\begin{gather*}
	\nodeof{\anode} \defeq\, \boxed{\:\begin{aligned}[t]
		\strut&\keyof{\anode}~\persto~\keyvarof{\anode} \MSTAR
		\selof{\predvarof{\anode}}{\keysel}~\persto~\keyvarof{\predvarof{\anode}} \MSTAR
		\selof{\succvarof{\anode}}{\keysel}~\persto~\keyvarof{\succvarof{\anode}}
		\\\MSTAR&
		\predof{\anode}~\fracto~\predvarof{\anode} \MSTAR
		\succof{\anode}~\fracto{2}~\succvarof{\anode} \MSTAR
		\markof{\anode}~\fracto{2}~\markvarof{\anode} \MSTAR
		\selof{\anode}{\inflowfld}~\fracto~\inflowvarof{\anode}
		\\\MSTAR&
		\listlockof{\anode}~\fracto{2}~\listvarof{\anode} \MSTAR
		(\listvarof{\anode}=0\sepimp\glnodeof{\anode}) \MSTAR 
		\treelockof{\anode}~\fracto~\treevarof{\anode}
		\\\MSTAR&
		\leftof{\anode}~\fracto~\leftvarof{\anode} \MSTAR
		\rightof{\anode}~\fracto~\rightvarof{\anode} \MSTAR
		\parentof{\anode}~\fracto~\parentvarof{\anode}
	\end{aligned}}
	\\[1pt]
	\glnodeof{\anode} \;\defeq~
		\listlockof{\anode}~\fracto{2}~\listvarof{\anode} \MSTAR
		\succof{\anode}~\fracto{2}~\succvarof{\anode} \MSTAR
		\markof{\anode}~\fracto{2}~\markvarof{\anode} 
\end{gather*}
Field $\inflowfld$ is the ghost field storing the node's inflow (cf. \Cref{sec:lotree:locality}).
We use fractional permissions for the fields $\listsel$, $\succsel$, and $\marksel$.
The $\listsel$ protects $\succsel$ which is why $\nodeof{\anode}$ has a full permission for $\succsel$ only if $\listsel$ is unlocked.
Otherwise, there is half a permission, the other half is transferred to the local state of the locking thread.
The setup for $\marksel$ is similar.

As noted above, the lock protects the resources $\glnodeof{\anode}$ whose ownership is transferred from the shared state to the local state of the thread acquiring the lock.
To make this precise, we define \(
	\holds{\anode} ~\defeq~ \listlockof{\anode}~\fracto{2}~1 \MSTAR \glnodeof{\anode}
\) and obtain the following behavior of locks: 
\begin{align*}
	\hoareof{\nodeof{\anode}}{~~\,\text{\code{lock($\anode$.listLock)}}~~\,}{\nodeof{\anode}\mstar\holds{\anode}}
	\\\text{and}\qquad
	\hoareof{\nodeof{\anode}\mstar\holds{\anode}}{~\text{\code{unlock($\anode$.listLock)}}~}{\nodeof{\anode}}
\end{align*}
For the first Hoare triple, note that its pre condition does not require $\listlockof{\anode}$ to be unlocked, $\listvarof{\anode}=0$.
This is established by \code{lock} as it blocks until $\listlockof{\anode}$ can be acquired.
The post condition realizes the ownership transfer: $\holds{\anode}$ contains the protected resources $\glnodeof{\anode}$ in the local state while maintaining the node's shared resources $\nodeof{\anode}$.

With the resources of individual nodes set up, we are ready to state the invariant of the LO-tree:
{\newcommand{\MKNL}{\\&\hspace{-1.5cm}\mstar~}%
\begin{align*}
	\inv(\abscontent,\abskeyset,\somenodes,\somenodesp)
		&~\defeq~~
		\sharedinv(\abscontent,\abskeyset,\somenodes,\somenodesp)
		~\MSTAR~
		{\bigmstar}_{\anode\in\somenodesp}~\nodeof{\anode}\MSTAR\nodeinv(\somenodes,\somenodesp,\anode)
	\\
	\sharedinv(\abscontent,\abskeyset,\somenodes,\somenodesp)
		&~\defeq~~
		\ptrmin,\ptrmax\in\somenodes
		~\MSTAR~
		\nullptr\notin\somenodes
		~\MSTAR~
		\somenodesp\subseteq\somenodes
		~\MSTAR~
		\abscontent=\contentsof{\somenodesp}
		~\MSTAR~
		\abskeyset=\keysetof{\somenodesp}
	\\
	\nodeinv(\somenodes,\somenodesp,\anode)
		&~\defeq~~
		\contentsof{\anode}\subseteq\keysetof{\anode}
		\MSTAR
		\predvarof{\anode},\succvarof{\anode}\in\somenodes
		\MSTAR
		\leftvarof{\anode},\rightvarof{\anode}\in\somenodes\prall{\cup}\setcompact{\nullptr}
		\tag{I1}
		\label{eq:lotree:inv:contents-closed}
		\MKNL
		\bigl(\anode=\ptrmin \Rightarrow \neg\markvarof{\anode} \mstar \keyvarof{\anode}=-\infty\bigr)
		\MSTAR
		\bigl(\anode=\ptrmax \Rightarrow \neg\markvarof{\anode} \mstar \keyvarof{\anode}=\infty\bigr)
		\tag{I2}
		\label{eq:lotree:inv:minmax}
		\MKNL
		\bigl(\neg\markvarof{\anode} \Rightarrow \insetof{\anode}\neq\emptyset\bigr)
		\MSTAR
		\bigl(\insetof{\anode}\neq\emptyset \Rightarrow [\keyvarof{\anode},\infty]\subseteq\insetof{\anode}\bigr)
		\MSTAR
		\isonflowpath{\anode}
		\tag{I3}
		\label{eq:lotree:inv:flow}
		\MKNL
		\keyvarof{\predvarof{\anode}}<\keyvarof{\anode}<\keyvarof{\succvarof{\anode}}
		\tag{I4}
		\label{eq:lotree:inv:sorted}
\end{align*}}%
The invariant follows the form and satisfies the properties laid out in \Cref{sec:linearizability,sec:lotree:locality}.
Its main part is the node-$\anode$-local invariant $\nodeinv$, which restricts the resources held by the overall invariant $\inv$.
The properties are as follows.
\begin{inparaitem}
	\item[{\eqref{eq:lotree:inv:contents-closed}}]
		The contents of a node are governed by its keyset.
		Moreover, the invariant is closed under following pointer fields of $\anode$.
		Observe that we require the overall invariant containing full $\somenodes$ to be closed, not the fragment comprising $\somenodesp$.
	\item[{\eqref{eq:lotree:inv:minmax}}]
		Nodes $\ptrmin$ resp. $\ptrmax$ are unmarked and store values $-\infty$ resp. $\infty$.
	\item[{\eqref{eq:lotree:inv:flow}}]
		Unmarked nodes have a non-empty inset which contains all values greater or equal to the node's own value.
		Moreover, nodes receive inset from at most one node, meaning that the $\succsel$ list between $\ptrmin$ and $\ptrmax$ is a path.
		The abstract predicate $\isonflowpath{\anode}$ can be expressed using flows.
	\item[{\eqref{eq:lotree:inv:sorted}}]
		Nodes are sorted in the sense that a node's predecessor (successor) stores a lesser (greater) key.
\end{inparaitem}
It is worth pointing out that \eqref{eq:lotree:inv:sorted} is a node-$\anode$-local property indeed, because $\nodeof{\anode}$ holds the required resources.

We may simply write $\inv(\abscontent,\somenodesp)$ instead of $\inv(\abscontent,\abskeyset,\somenodes,\somenodesp)$ if $\somenodes$ is clear from the context.


\subsection{Proof Outline}
\label{sec:lotree:proof-outline}

The proof outline can be found in \Cref{fig:lotree:impl}.
While the proof for \code{insert} and \code{delete} requires mostly standard reasoning, it reveals the interference that other threads are subjected to.
The hindsight reasoning for method \code{contains} is performed relative to this interference.

\looseness=-1
Using our proof system $\semcalclin$, we give a proof \emph{template} of the LO-tree:
we do not make any assumptions about the operations manipulating the tree overlay other than them being memory-safe.

\subsubsection{Locating Nodes}
Recall from \Cref{sec:lotree:overview} that \code{insert} and \code{delete} use the helper \code{locate} to find the position $\pprev,\pnext$ to which a given key $\vk$ belongs.
Node $\pprev$ is the result of a tree traversal, \Cref{code:lotree:locate:tree}.
Since we elide the mechanics of the tree overlay, we only know that the resulting pointer is non-$\nullptr$---this little information suffices.
Next, $\pprev$ is locked, \Cref{code:lotree:locate:lock}.
This provides us with the protected resources, $\glnodeof{\pprev}$. 
They guarantee that $\succof{\pprev}$ and $\markof{\pprev}$ cannot change due to interference.
Reading $\succof{\pprev}$, \Cref{code:lotree:locate:next}, binds $\pnext$ to $\succvarof{\pprev}$.
Hence, the validation of position $\pprev,\pnext$ on \Cref{code:lotree:locate:return} results in the interference-free knowledge that $\pprev$ is unmarked, $\pnext$ is the successor of $\pprev$, and that $\vk$ indeed belongs in-between $\pprev$ and $\pnext$, $\vk\in(\keyvarof{\pprev},\keyvarof{\pnext}]$.
This together with the obtained resources forms the predicate $\locateInv(\somenodes,\pprev,\pnext)$, formally defined in \Cref{fig:lotree:impl}, and is the post condition of \code{locate} on \Cref{code:lotree:locate:post}.
Later, we will use the fact that $\locateInv(\somenodes,\pprev,\pnext)$ implies $\vk\in\keysetof{\pnext}$.
To see this, invoke invariant~\eqref{eq:lotree:inv:flow} for the unmarked $\pprev$.
We get $[\keyvarof{\pprev},\infty]\subseteq\insetof{\pprev}$.
The keys $(\keyvarof{\pprev},\infty]$ distributes via $\succof{\pprev}$ as inset to $\pnext$ according to \Cref{sec:lotree:locality}.
Hence, $\vk\in\keysetof{\pnext}=(\keyvarof{\pprev},\keyvarof{\pnext}]$.

\subsubsection{Insertions}
An Insertion of key $\vk$ first calls \code{locate} to find the position $\pprev,\pnext$ to which $\vk$ belongs.
The position reveals if $\vk$ is already contained because $\vk\in\keysetof{\pnext}$ as inferred above.
If $\keyvarof{\pnext}=\vk$, then $\vk\in\contentsof{\pnext}$ and thus $\vk\in\abscontent$.
That is, if the conditional in \Cref{code:lotree:insert:check} succeeds, the specification of an unsuccessful insertion is met.
We trade the obligation $\OBLi{\vk}$ for the receipt $\FULi{\false}$.


\begin{figure}
	\centering
	\definecolor{colorMyInsDelFlow}{RGB}{40,40,160}
	\definecolor{colorMyInsDelKeyset}{RGB}{40,140,20}
	\newcommand{\colFlow}{\color{colorMyInsDelFlow}}
	\newcommand{\colKeyset}{\color{colorMyInsDelKeyset}}
	\begin{tikzpicture}[trees]
		\begin{scope}[shift={(0,0)}]
			\node (x) [treenode] at (0,0) {$\pprev$};
			\node (n) [treenode] at (.75,-.75) {$\pnode$};
			\node (z) [treenode] at (1.5,0) {$\pnext$};
			\draw (x.20) edge[listptr,solid,->] node[interval]{\colFlow$(\keyof{\pprev},\infty]$} (z.160);
			\draw (z.200) edge[listptr,->] (x.-20);
			\draw (n.120) edge[listptr,->] (x.-45);
			\draw (n) edge[listptr,solid,->] (z);
			\draw (x.160) edge[listptr,solid,<-] node[interval]{\colFlow$[\keyof{\pprev},\infty]$} ++(-.75,0);
			\draw (x.200) edge[listptr,->] ++(-.75,0);
			\draw (z.20) edge[listptr,solid,->] node[interval]{\colFlow$(\keyof{\pnext},\infty]$} ++(.75,0);
			\draw (z.-20) edge[listptr,<-] ++(.75,0);
			\node[interval] at (x.north east) {\colKeyset$[\keyof{\pprev},\keyof{\pprev}]$};
			\node[interval] at (z.north east) {\colKeyset$(\keyof{\pprev},\keyof{\pnext}]$};
			\node[hinterval,anchor=north west,inner sep=0pt] at (n.south east) {\colKeyset$\emptyset$};
		\end{scope}

		\node (lta) at (3.3,.2) {\bigleadsto};
		\node[yshift=8pt] at (lta) {\footnotesize\Cref{code:lotree:insert:linksucc}};
		\node (lfa) at (3.3,-.2) {\bigleadsfrom};
		\node[yshift=-8pt] at (lfa) {\footnotesize\Cref{code:lotree:delete:unlinksucc}};

		\begin{scope}[shift={(5,0)}]
			\node (x) [treenode] at (0,0) {$\pprev$};
			\node (n) [treenode] at (.75,-.75) {$\pnode$};
			\node (z) [treenode] at (1.5,0) {$\pnext$};
			\draw (x.-80) edge[listptr,solid,->] node[left,hinterval,pos=.75]{\colFlow$(\keyof{\pprev},\infty]$} (n.165);
			\draw (z.200) edge[listptr,->] (x.-20);
			\draw (n.120) edge[listptr,->] (x.-45);
			\draw (n.20) edge[listptr,solid,->] node[right,hinterval,pos=.22]{\colFlow$(\keyof{\pnode},\infty]$} (z.-110);
			\draw (x.160) edge[listptr,solid,<-] node[interval]{\colFlow$[\keyof{\pprev},\infty]$} ++(-.75,0);
			\draw (x.200) edge[listptr,->] ++(-.75,0);
			\draw (z.20) edge[listptr,solid,->] node[interval]{\colFlow$(\keyof{\pnext},\infty]$} ++(.75,0);
			\draw (z.-20) edge[listptr,<-] ++(.75,0);
			\node[interval] at (x.north east) {\colKeyset$[\keyof{\pprev},\keyof{\pprev}]$};
			\node[interval] at (z.north east) {\colKeyset$(\keyof{\pprev},\keyof{\pnext}]$};
			\node[hinterval,anchor=north west,inner sep=0pt] at (n.south east) {\colKeyset$(\keyof{\pprev},\keyof{\pnode}]$};
		\end{scope}

		\node (ltb) at (8.3,.2) {\bigleadsto};
		\node[yshift=8pt] at (ltb) {\footnotesize\Cref{code:lotree:insert:linkpred}};
		\node (lfb) at (8.3,-.2) {\bigleadsfrom};
		\node[yshift=-8pt] at (lfb) {\footnotesize\Cref{code:lotree:delete:unlinkpred}};

		\begin{scope}[shift={(10,0)}]
			\node (x) [treenode] at (0,0) {$\pprev$};
			\node (n) [treenode] at (.75,-.75) {$\pnode$};
			\node (z) [treenode] at (1.5,0) {$\pnext$};
			\draw (x.-80) edge[listptr,solid,->] node[left,hinterval,pos=.75]{\colFlow$(\keyof{\pprev},\infty]$} (n.165);
			\draw (z.215) edge[listptr,->] (n.60);
			\draw (n.120) edge[listptr,->] (x.-45);
			\draw (n.20) edge[listptr,solid,->] node[right,hinterval,pos=.22]{\colFlow$(\keyof{\pnode},\infty]$} (z.-110);
			\draw (x.160) edge[listptr,solid,<-] node[interval]{\colFlow$[\keyof{\pprev},\infty]$} ++(-.75,0);
			\draw (x.200) edge[listptr,->] ++(-.75,0);
			\draw (z.20) edge[listptr,solid,->] node[interval]{\colFlow$(\keyof{\pnext},\infty]$} ++(.75,0);
			\draw (z.-20) edge[listptr,<-] ++(.75,0);
			\node[interval] at (x.north east) {\colKeyset$[\keyof{\pprev},\keyof{\pprev}]$};
			\node[interval] at (z.north east) {\colKeyset$(\keyof{\pnode},\keyof{\pnext}]$};
			\node[hinterval,anchor=north west,inner sep=0pt] at (n.south east) {\colKeyset$(\keyof{\pprev},\keyof{\pnode}]$};
		\end{scope}
	\end{tikzpicture}
	\vspace{-6mm}
	\caption{%
		Physical linking (\Cref{code:lotree:insert:linksucc,code:lotree:insert:linkpred}) and unlinking (\Cref{code:lotree:delete:unlinksucc,code:lotree:delete:unlinkpred}) of node $\pnode$.
		The arrows \textlistptr[->] resp. \textlistptr[->,solid] indicate $\predsel$ resp. $\succsel$ pointers.
		The {\colFlow intervals} on $\succsel$ links denote the insets, the {\colKeyset intervals} on nodes denote their keysets.
		Mark bits ($\pnode$ is marked prior to the unlinking) and acquired locks are not depicted.
		\label{fig:illustration-insertion}
		\label{fig:illustration-deletion}
		\label{fig:illustration-linkage}
	}
	\vspace{-2mm}
\end{figure}

Otherwise, $\vk$ is inserted into the structure.
To do that, a new node $\pnode$ containing $\vk$ is allocated in \Cref{code:lotree:insert:malloc}.
The $\predsel$ and $\succsel$ fields are set to $\pprev$ and $\pnext$, respectively.
It remains to link $\pnode$ into the logical ordering, as depicted in \Cref{fig:illustration-insertion}.
First, \Cref{code:lotree:insert:linksucc} redirects $\succof{\pprev}$ to $\pnode$.
This is the linearization point: $\pnode$ receives the inset $(\keyvarof{\pprev},\infty]$ from $\pprev$ so that we get $\contentsof{\pnode}=\set{\vk}$.
Hence, the update turns $\abscontent$ into $\abscontent\cup\set{\vk}$ so that $\OBLi{\vk}$ can be traded for $\FULi{\true}$.
Next, \Cref{code:lotree:insert:linkpred} redirects $\predof{\pnext}$ to $\pnode$.
The command has no effect on the logical contents $\abscontent$ which is why we need no $\OBLi{\vk}$ to proceed.
It is readily checked that the update maintains the node-local invariants of the nodes $\pprev,\pnode,\pnext$.

Our proof outline does not consider the methods for inserting the new node $\pnode$ into the tree overlay.
We simply assume that \code{prepareTreeInsertion} in \Cref{code:lotree:insert:preparetree} produces an interference-free predicate $\hiddenIns{\pprev,\pnext}$ that is maintained by the updates of the logical ordering in \Cref{code:lotree:insert:linksucc,code:lotree:insert:linkpred} and consumed by the later \code{performTreeInsertion} in \Cref{code:lotree:insertIntoTree}.

\subsubsection{Deletions}
Deletions are similar to insertions (see \Cref{fig:illustration-deletion}).
We omit the details.

\subsubsection{Contains}
\label{sec:lotree:proof-contains}

The proof in \Cref{fig:lotree:impl} uses implicitly existentially quantified symbolic variables $\anodeval,\anodevalp,\avalval$ to share knowledge between now and past predicates.
We cannot use program variables for this purpose because their values change during computation, meaning they may be valuated differently in now and past predicates.
To further avoid confusion between now and past states, we write $\old{e}$ to replace in expression $e$ all symbolic variables like $\markof{\anodeval}$ with $\old{\markof{\anodeval}}$.
We think of $\old{e}$ as the \emph{old} version and use it under past operators.
For example, in $\nowof{\inv(\abscontent,\set{\anodeval})}\mstar\pastof{\old{\inv(\abscontentp,\set{\anodeval})}}$ we would use $\insetof{\anodeval}$ resp. $\old{\insetof{\anodeval}}$ to clearly refer to the inset of $\anodeval$ in the current resp. past state.
The proof of \code{contains($\vk$)} has these five stages:

%
\begin{wrapfigure}[15]{r}{.44\textwidth}
	\vspace{-4mm}
	\newcommand{\mkCG}{\color{green!60!black}}
	\newcommand{\MYMSTAR}{{}\mstar{}}
	\newcommand{\succInvOf}[1]{\succInv(\abscontent,\abscontentp\!\!,\somenodes,\somenodesp,#1,\anodeval)}
	\begin{lstlisting}[language=compactSPL,belowskip=-2mm]
$\ANNOT{ \OBLc{\vk} \MYMSTAR \succInvOf{\pcurr} }$
while ($\pcurr$.key < $\vk$)@\:@{
	val $\pnext$ = $\pcurr$.succ
	$\ANNOTML{ &\OBLc{\vk} \MYMSTAR \succInvOf{\pcurr} \\& \MYMSTAR \keyvarof{\anodeval}<\vk \MYMSTAR \pnext=\anodevalp=\succvarof{\anodeval} }$   $\label{code:lotree:hindsight-succ:inter-pre}$
	//@\;\mkCG$\weakhypof{\apred}{\apredp}{\apred\cap\apredp}$\;@with $\label{code:lotree:hindsight-succ:inter-hyp}$
	// @\;\mkCG$\apred\defeq\inv(\abscontentp\!\!,\somenodesp\prall{\cup}\anodeval) \MYMSTAR \insetof{\anodeval}\neq\emptyset$@
	// @\;\mkCG$\apredp\defeq\inv(\abscontent,\somenodes\prall{\cup}\anodeval) \MYMSTAR \succvarof{\!\anodeval}\prall{=}\anodevalp \MYMSTAR \keyvarof{\!\anodeval}\prall{<}\vk$@
	@\color{teal}skip@ $\label{code:lotree:hindsight-succ:inter-skip}$
	$\ANNOT{ \OBLc{\vk} \MYMSTAR \succInvOf{\pnext} }$   $\label{code:lotree:hindsight-succ:inter-post}$
	$\pcurr$ := $\pnext$
}
$\ANNOT{ \OBLc{\vk} \MYMSTAR \succInvOf{\pcurr} \MYMSTAR \vk\leq\keyvarof{\anodeval} }$
\end{lstlisting}
	\caption{
		Detailed proof outline and temporal interpolation for \Cref{code:lotree:contains:go-succ-pre,code:lotree:contains:go-succ,code:lotree:contains:go-succ-post}.
		\label{fig:lotree:hindsight-succ}
	}
\end{wrapfigure}
%
%
(1)
The tree traversal, \Cref{code:lotree:contains:treetraversal}, finds a starting node $\pcurr$ for traversing the logical ordering.
The only guarantee for $\pcurr$ is that it is non-$\nullptr$, \Cref{code:lotree:contains:treetraversal-post}.

(2)
The logical ordering is traversed by following $\predsel$ fields as long as $\vk$ is less than the key in the traversed node, \Cref{code:lotree:contains:go-pred}.
The resulting node $\pcurr$ is non-$\nullptr$ by~\eqref{eq:lotree:inv:contents-closed}.
Moreover, we obtain the interference-free fact $\vk \geq \keyof{\pcurr}$, \Cref{code:lotree:contains:go-pred-post}.

(3)
The traversal continues to follow $\predsel$ pointers until an unmarked node is reached, \Cref{code:lotree:contains:fix}.
By invariant~\eqref{eq:lotree:inv:contents-closed}, the resulting node $\pcurr$ is non-$\nullptr$.
That $\pcurr$ is unmarked means that its inset is at least $[\keyvarof{\pcurr},\infty]$ by invariant~\eqref{eq:lotree:inv:flow}.
Moreover, $\vk \geq \keyof{\pcurr}$ from the previous stage is preserved due to~\eqref{eq:lotree:inv:sorted}.
Together, this implies that $\vk$ is in $\pcurr$'s inset.
This fact is not interference-free because $\pcurr$ is not locked.
To preserve it, we turn it into a past predicate, \Cref{code:lotree:contains:go-succ-pre}.

(4)
The traversal follows $\succsel$ pointers as long as $\vk$ is greater than the key in the traversed node, \Cref{code:lotree:contains:go-succ}.
Using temporal interpolation (details below), we conclude that also the reached node $\pcurr$ had $\vk$ in its inset at some point.
Note that this together with $\vk\leq\keyof{\pcurr}$ from \Cref{code:lotree:contains:go-succ-post} means $\vk\in\keysetof{\pcurr}$ in some past state.
So $\vk$ was in the structure at this past state iff $\vk=\keyvarof{\pcurr}$. 

(5)
Using temporal interpolation (details below), we derive from the past contents and the current $\keysel$ field of $\pcurr$ whether or not $\vk$ has been logically contained, \Cref{code:lotree:contains:validate-annot}.
This past state is, in fact, the linearization point.
We retrospectively linearize, \Cref{code:lotree:contains:final}, before returning.

We turn to the details of the temporal interpolation that goes into stages (4) and (5).

\tinyskip
\emph{Temporal Interpolation in Stage~(4).}
The proof outline for the loop from \Cref{code:lotree:contains:go-succ} is given in \Cref{fig:lotree:hindsight-succ}.
The temporal interpolation needed here is this:
that $\pcurr$ had flow in the past and its $\succsel$ field currently points to $\pnext$ and its $\keysel$ field currently is less than $\vk$ means that all three facts were true simultaneously at some point.
Intuitively, this is the case because $\pcurr$ has a non-empty inset whenever $\succof{\pcurr}$ is changed and because $\keyvarof{\pcurr}$ is never changed.
Technically, we show the hypothesis $\weakhypof{\apred}{\apredp}{\apred\cap\apredp}$ on \Cref{code:lotree:hindsight-succ:inter-hyp} with
\begin{align*}
	\apred ~\defeq~ \inv(\abscontentp,\somenodesp\cup\anodeval) \MSTAR \insetof{\anodeval}\neq\emptyset
	~~\quad\text{and}\quad~~
	\apredp ~\defeq~ \inv(\abscontent,\somenodes\cup\anodeval) \MSTAR \succvarof{\anodeval}=\anodevalp \MSTAR \keyvarof{\anodeval}<\vk
	\ .
\end{align*}
The symbolic variables $\anodeval$ resp. $\anodevalp$ are bound to $\pcurr$ resp. $\pprev$ by the outer proof context; we use $\anodeval$/$\anodevalp$ instead of $\pcurr$/$\pprev$ as they are logically pure and thus do not change their valuation.
To prove the hypothesis, we establish \(
	\thePredicates,\theInterference\semcalc\hoareof{\acpred}{\stmtof{\theInterference}}{\nowof{\apredp}\rightarrow\weakpastOf{\apred\cap\apredp}}
\) and \(
	\isInterferenceFreeOf[\theInterference]{\thePredicates}
\)
for some set $\thePredicates$ of predicates with $\nowof{\apred}\subseteq\acpred\in\thePredicates$~(cf. \Cref{Section:TemporalInterpolation}).
We cannot simply use $\acpred=\nowof{\apred}$ because $\apred$ is not interference-free.
Instead, we use $\acpred\defeq\nowof{\apredp}\rightarrow\weakpastOf{\apred\cap\apredp}$.
It is easy to see that $\acpred$ is weaker than $\nowof{\apred}$, $\nowof{\apred}\subseteq\acpred$.
Note that $\acpred$ is the invariant that the hypothesis proof strategy from \Cref{Lemma:ProofStrategy} asks for.

Next, we show that $\acpred$ is interference-free, i.e., $\csemOf{(\acpredpp,\acom)}(\acpred)\subseteq\acpred$ for all interferences $(\acpredpp,\acom)\in\theInterference$ of the LO-tree.
For an interference $(\acpredpp,\acom)$ to invalidate $\acpred$ it must change the truth of $\apredp$ in the current state.
If the truth of $\apredp$ is changed to $\false$, then $\acpred$ is vacuously true.
Otherwise, the interference changes $\succvarof{\anodeval}$ to $\anodevalp$ ($\keyvarof{\anodeval}$ is not changed by any interference).
This means $\acom$ stems from \Cref{code:lotree:insert:linksucc} in \code{insert} or \Cref{code:lotree:delete:unlinksucc} in \code{delete}.
In both cases we know from the proof (\Cref{fig:illustration-linkage,fig:lotree:impl}) that $\anodeval$ has a non-empty inset after the interfering update.
Concretely, this means $\csemOf{(\acpredpp,\acom)}(\acpred)\subseteq\nowof{\apred}$.
Because we already established $\nowof{\apred}\subseteq\acpred$, we obtain the interference-freedom of $\acpred$, as required.

It remains to show that $\acpred$ is invariant under the self-interferences $\stmtof{\theInterference}$.
To see this, observe that $\acpred$ concerns only the global state, not the local state.
Hence, the self-interferences invalidate $\acpred$ iff the interferences of other threads do so.
Since the latter is not the case, nothing needs to be shown.

With the hypothesis proved, we obtain $\pastOf{\apred\cap\apredp}$ from Rule~\ruleref{temporal-interpolation}.
The rule is applied to a command which we make explicit in the form of $\cskip$ on \Cref{code:lotree:hindsight-succ:inter-skip}.
One can avoid this $\cskip$ by applying the rule together with the previous command.
Finally, we invoke invariant \eqref{eq:lotree:inv:flow} under the past predicate to obtain $[\old{\keyvarof{\anodeval}},\infty]\subseteq\old{\insetof{\anodeval}}$.
By definition, this means that $\old{\succvarof{\anodeval}}=\anodevalp$ receives $\old{\insetof{\anodeval}}\setminus(\old{\keyvarof{\anodeval}},\infty]$.
Because $\old{\keyvarof{\anodeval}}<\vk$, this means $\vk\in\old{\insetof{\anodevalp}}$.
Altogether, we arrive at the desired assertion on \Cref{code:lotree:hindsight-succ:inter-post}, namely $\pastOf{\old{\inv(\abscontentp,\somenodesp\cup\anodevalp)} \mstar \vk\in\old{\insetof{\anodevalp}}}$.

\tinyskip
\emph{Temporal Interpolation in Stage~(5)}
We proceed in two steps.
First, we prove that $\weakhypof{\apred}{\apredp}{\apred\cap\apredp}$ holds for arbitrary $\apred$ and $\apredp\defeq\keyvarof{\anodeval}=\avalval$.
As before, we use \Cref{Lemma:ProofStrategy} with invariant ${\acpred\defeq\nowof{\apredp}\rightarrow\weakpastOf{\apred\cap\apredp}}$.
Since $\keyvarof{\anodeval}$ is immutable, $\acpred$ is immediately stable under (self-)interferences.
This justifies to move facts about the $\keysel$ freely between now and past states.

\looseness=-1
Towards the assertion on \Cref{code:lotree:contains:validate-annot}, assume $\keyvarof{\anodeval}=\vk$.
We move this fact into the past predicate from \Cref{code:lotree:contains:go-succ-post} using the above argument.
The result is: $\pastOf{\old{\inv(\abscontentp,\somenodesp)}\mstar\vk\in\old{\insetof{\anodeval}}\mstar\vk=\old{\keyvarof{\anodeval}}}$.
This means that $\vk$ was contained in the structure in the past: $\pastOf{\old{\inv(\abscontentp,\somenodesp)}\mstar\vk\in\old{\contentsof{\anodeval}}\subseteq\abscontentp}$.
This conclusion uses the fact that $\vk\in\old{\insetof{\anodeval}}\mstar\vk=\old{\keyvarof{\anodeval}}$ implies $\vk\in\old{\keysetof{\anodeval}}$.
The case for $\keyvarof{\anodeval}\neq\vk$ is similar.
Overall, rewriting both cases into one yields the desired assertion, \Cref{code:lotree:contains:validate-annot}.
Finally, this allows us to retrospectively linearize as the past predicate witnesses a past state where $\vk$ was resp. was not in the structure as reflected by the return value.
This concludes the linearizability proof.


\subsection{Proof Automation}
\label{sec:lotree:automation}

We substantiate our claims that temporal interpolation and the resulting proof system for linearizability aid automated proof construction. To this end, we adapted the \plankton tool \cite{DBLP:journals/pacmpl/MeyerWW22}.
\plankton is a verifyer for non-blocking data structures that constructs proofs in the program logic from \Cref{sec:Preliminaries} extended by rules for linearizability akin to those from \Cref{sec:linearizability}.
To be more precise, \plankton takes as input the implementation under scrutiny together with a candidate node invariant, like $\nodeinv(\somenodes,\somenodesp,\anode)$ from \Cref{sec:lotree:invariant}.
It then performs an exhaustive proof search.

We extended \plankton to use our new proof rules from \Cref{Figure:ProgramLogicTI,fig:linearizability-rules}, in particular Rule~\ruleref{temporal-interpolation}.
Our implementation \cite{artifact} applies temporal interpolation only for hypotheses of the form $\hypof{\apred}{\apredp}{\apred\cap\apredp}$ and only if it is able to discharge the hypothesis using \Cref{Lemma:ProofStrategy} with invariant $\nowof{\apredp}\rightarrow\weakpastof{(\apred\cap\apredp)}$.
This eager approach ensures that we do not pollute the proof search with temporal interpolations that are doomed to fail because their hypotheses do not hold.
Note that this is possible despite a potentially incomplete interference set as \plankton restarts proof construction whenever a new interference is discovered.
Altogether, our implementation establishes linearizability results along \Cref{Theorem:Soundness}.

We used our tool to verify automatically the LO-tree from \Cref{fig:lotree:impl}.
Similarly to the presented proof, we did not use the actual implementation of the helper functions modifying the tree overlay.
Instead, we used \emph{most general stubs}, functions that change the tree overlay arbitrarily (leaving the logical ordering list unchanged).
The node invariant we specified is the one from \Cref{sec:lotree:invariant}.
With this, \plankton is able to fully automatically construct a linearizability proof for the LO-tree within twenty minutes~(see \Cref{table:benchmarks}).
We stress that this includes fully automatic applications of temporal interpolation, which are strictly necessary to prove the LO-tree linearizable.

We also compared our new version of \plankton against the original version form \citet{DBLP:journals/pacmpl/MeyerWW22}.
See \Cref{table:benchmarks} for the results: temporal interpolation incurs a slow down of factor $3.15$ in the worst case and factor $2$ on average. We believe that this slowdown is justified by the reasoning power brought by temporal interpolation.
We consider a more extensive evaluation of our implementation future work.
As of now, \plankton's proof construction is limited by orthogonal concerns (e.g. imprecise joins, the handling of updates with non-local effects) that still limit its applicability.


\begin{table*}%
	\caption{Runtime comparison of our novel temporal interpolation proof rule with the tool from \citet{DBLP:journals/pacmpl/MeyerWW22}. The experiments were conducted on an Apple M1 Pro.}%
	\label{table:benchmarks}%
		\center%
		\newcommand{\thespacer}{\quad~}
		\newcommand{\cellFactor}[1]{\thespacer\(\times#1\)}
		\newcommand{\cellYes}[1]{\thespacer\makebox[1cm][r]{\(#1\)}\hspace{1.5mm}\makebox[.4cm][l]{\rawsymbolYes}\thespacer}
		\newcommand{\cellNo}{\thespacer\makebox[1cm][r]{---~}\hspace{1.5mm}\makebox[.4cm][l]{\!\rawsymbolNo}\thespacer}
		\begin{tabularx}{\textwidth}{Xccc}%
			\toprule
			Benchmark
				& \citet{DBLP:journals/pacmpl/MeyerWW22}
				& This Paper
				& \thespacer Factor
				\\
			\midrule
			Fine-Grained set
				& \cellYes{44s}
				& \cellYes{45s}
				& \cellFactor{1.02}
				\\
			Lazy set
				& \cellYes{1m\,21s}
				& \cellYes{2m\,13s}
				& \cellFactor{1.65}
				\\
			FEMRS tree {(no maintenance)}
				& \cellYes{2m\,22s}
				& \cellYes{3m\,50s}
				& \cellFactor{1.62}
				\\
			Vechev\&Yahav~2CAS set
				& \cellYes{1m\,09s}
				& \cellYes{1m\,15s}
				& \cellFactor{1.08}
				\\
			Vechev\&Yahav~CAS set
				& \cellYes{0m\,52s}
				& \cellYes{2m\,20s}
				& \cellFactor{2.70}
				\\
			ORVYY set
				& \cellYes{0m\,54s}
				& \cellYes{1m\,36s}
				& \cellFactor{1.79}
				\\
			Michael set
				& \cellYes{3m\,06s}
				& \cellYes{6m\,53s}
				& \cellFactor{2.22}
				\\
			Michael set {(wait-free search)}
				& \cellYes{3m\,42s}
				& \cellYes{6m\,53s}
				& \cellFactor{1.86}
				\\
			Harris set
				& \cellYes{18m\,14s}
				& \cellYes{57m\,20s}
				& \cellFactor{3.15}
				\\
			Harris set {(wait-free search)}
				& \cellYes{19m\,54s}
				& \cellYes{43m\,00s}
				& \cellFactor{2.16}
				\\
			LO-tree (maintenance stubs)
				& \cellNo
				& \cellYes{16m\,43s}
				& \thespacer ---
				\\
			\bottomrule
	\end{tabularx}
\end{table*}


\section{Related Work}
\label{sec:related}

\looseness=-1
The hindsight principle~\cite{DBLP:conf/podc/OHearnRVYY10,DBLP:conf/wdag/Lev-AriCK15,DBLP:conf/wdag/FeldmanE0RS18,DBLP:journals/pacmpl/FeldmanKE0NRS20} and our temporal interpolation have relatives in classical program verification~\cite{DBLP:books/daglib/0080029,DBLP:series/txcs/Schneider97}.
So-called causality formulas, in our notation written as $\nowof{\apred}\rightarrow\pastof{\apredp}$, express that $\apredp$ is a prerequisite for seeing~$\apred$.
Temporal interpolation is more general in that it may take past information into account in order to infer the existence of an intermediary state.
Yet, the past invariance proof principle by \citet[\S4.1]{DBLP:books/daglib/0080029} inspired an application of \ruleref{temporal-interpolation} in the RDCSS proof \techreport{(\Cref{sec:RDCSS})}\nontechreport{\cite[\S{}F]{techreport}} to derive a contradiction in a case distinction.
The careful identification of verification conditions by \citet{DBLP:books/daglib/0080029} has also lead us to the definition of hypotheses that can be proven in isolation.
What sets our work apart is that we incorporate temporal interpolation into a modern program logic with powerful reasoning techniques~\cite{DBLP:journals/jfp/JungKJBBD18} such as framing~\cite{DBLP:conf/csl/OHearnRY01}, atomic triples~\cite{DBLP:conf/ecoop/PintoDG14}, and general separation algebras~\cite{DBLP:conf/lics/CalcagnoOY07}, in particular flows~\cite{DBLP:journals/pacmpl/KrishnaSW18,DBLP:conf/esop/KrishnaSW20}.

There are first tools that automate linearizability proofs based on hindsight reasoning.
The \texttt{poling} tool~\cite{DBLP:conf/cav/ZhuPJ15} implements the hindsight lemma in the formulation of~\citet{DBLP:conf/podc/OHearnRVYY10}.
The \plankton tool~\cite{DBLP:journals/pacmpl/MeyerWW22} automates a restricted form of hindsight reasoning that can be expressed via state-independent variables shared between a past and the current state.
However, it did not support general temporal interpolation prior to our extension.
Without this extension, the tool would have been unable to verify the LO-tree and other structures that require more complex hindsight reasoning.

We are not the first to study program logics defined over computations instead of states.
History-based local rely-guarantee~\cite{DBLP:conf/concur/FuLFSZ10,DBLP:conf/esop/GotsmanRY13} has an elaborate assertion language whose temporal operators are carefully harmonized with the rules of the program logic.
Our approach builds on the logic proposed by~\citet{DBLP:journals/pacmpl/MeyerWW22} from which it inherits the notion of past predicates over computations.
We introduce temporal interpolation by means of a new proof rule.
The soundness result shows that the proof rule can be eliminated, and hence is really a mechanism for structuring complex proofs.
This means that, in principle, all of our proofs can also be expressed in the logic of~\citet{DBLP:journals/pacmpl/MeyerWW22}. Doing so, however, requires one to repeat the soundness arguments within each program proof anew.
In particular, this (i) requires reasoning about the governed computations explicitly and (ii) thwarts the use of the frame rule.
Realistically, this would make the proofs intractable, even manual ones.
The conclusions~\cite{DBLP:journals/pacmpl/MeyerWW22} can draw directly about the past of the computation are all based on immutability arguments, and compared to what we propose here this is a very weak form of hindsight reasoning.
Notably, the version of \plankton presented in~\cite{DBLP:journals/pacmpl/MeyerWW22} cannot handle the example from \Cref{sec:Overview} nor the LO-tree from \Cref{sec:lotree}.
Comparing to other computation-based separation logics, we note that the formalization of computations matters: definitions based on interleaving products~\cite{DBLP:conf/sas/BellAW10} or the union of sets of events \cite{DBLP:conf/esop/SergeyNB15,DBLP:conf/ecoop/DelbiancoSNB17} seem to be less suited for temporal interpolation.

Prophecies were introduced to separation logic by~\citet{DBLP:phd/ethos/Vafeiadis08} and formalized by~\citet{DBLP:conf/tamc/ZhangFFSL12} to structural prophecies that foresee the actions of one thread, a restriction overcome by~\citet{DBLP:journals/pacmpl/JungLPRTDJ20}.
Temporal interpolation conducts full subproofs in the presence of interferences.
However, it is in the nature of Owicki-Gries, and has been observed early on~\cite{DBLP:journals/acta/OwickiG76}, that interferences may require auxiliary variables to increase precision.
What seems to make prophecies more difficult to use is the need to reason about the computation backward, against the control flow~\cite{DBLP:conf/cav/BouajjaniEEM17}.
This is shared with simulation and refinement-based proofs~\cite{DBLP:conf/pldi/LiangF13,DBLP:conf/popl/TuronTABD13}, where backward reasoning is known to be complete~\cite{DBLP:conf/cav/SchellhornWD12}.

Our proofs use standard techniques like boxed assertions~\cite{DBLP:conf/concur/VafeiadisP07,DBLP:phd/ethos/Vafeiadis08}, fractional permissions~\cite{DBLP:conf/sas/Boyland03}, and persistent points-to predicates~\cite{DBLP:conf/cpp/VindumB21}.
Combining these techniques is no contribution of ours.
In fact, they were already combined in the original \plankton tool from \citet{DBLP:journals/pacmpl/MeyerWW22}, although the use of fractional permissions and persistent points-to predicates has not been discussed there (probably due to their focus on lock-free implementations).


\begin{acks}
This work is funded in parts by the \grantsponsor{GS100000001}{National Science Foundation}{http://dx.doi.org/10.13039/100000001} under grant~\grantnum{GS100000001}{1815633} and by an Amazon Research Award.
The third author is supported by a Junior Fellowship from the Simons Foundation (855328, SW).
\end{acks}

\section*{Data-Availability Statement}

Our extended version of \plankton and the dataset (\Cref{table:benchmarks}) analysed in the present paper are available in the Zenode repository \cite{artifact}, \url{https://zenodo.org/record/7829982/}.

	\bibliography{bib,dblp}

	\techreport{
		\appendix

\clearpage


\section{Bugs in the LO-Tree and their Fixes}
\label{appendix:lo-bugs-full}

The original version of the LO-tree \cite{DBLP:conf/ppopp/DrachslerVY14} contains two bugs which we fixed in \Cref{fig:lotree:impl}.
The fist bug concerns \code{contains}: concurrent insertions and deletions of a value $\vk$ in the original implementation may lead to duplicates of $\vk$ in the tree that do not agree on their $\marksel$ field, making \code{contains} produce non-linearizable results.
The second bug concerns \code{insert}: the original sequence in which new nodes are linked into the logical ordering leads \code{contains} to miss values and thus produce non-linearizable results.
This bug has been reported by \citet{DBLP:journals/pacmpl/FeldmanKE0NRS20}.

\smartparagraph{Bug 1: Duplicate Values}
A subtle quirk of the LO-tree is the fact that an insertion of value $\vk$ may be unaware of a concurrent deletion of $\vk$ because the tree traversal of the insertion experienced a rotation but still ended up in the right position for the insertion (the validation in \code{locate} succeeds).
Successful validation requires that the deletion already removed $\vk$ from the logical ordering.
As a consequence, the insertion can proceed and insert $\vk$ into the logical ordering and into the tree.
If the deletion has not yet removed the old marked version of $\vk$, then the tree contains two nodes with value $\vk$ that disagree on the mark bit.
Hence, the result of \code{contains($\vk$)} is influenced by rotations.


\begin{figure}[b]
	{\begingroup
		\centering
		\newcommand{\MYXSHIFT}{15mm}
		\newcommand{\MYSTEP}[1]{%
			\hspace{-3mm}%
			\raisebox{-4.5pt}{\stackon{\bigleadsto}{\text{}}}%
			\hspace{-3mm}%
		}
		\newcommand{\MYMARK}[1]{%
			\node at (#1.south east) [yshift=-2.5pt] {\scalebox{1.2}{\textcolor{colorList}{\rawsymbolNo}}};
		}
		\newcommand{\MYLOCK}[1]{%
			\node at (#1.north) [xshift=0pt,yshift=1pt] {{\color{colorTree}\tiny\dellock}};
		}
		\newcommand{\MYLOCKP}[1]{%
			\node at (#1.south) [xshift=0pt,yshift=-1pt] {{\color{colorTree}\tiny\inslock}};
		}
		\newcommand{\MYLAB}[2][]{%
			\node (lab) at (#2) [#1] {{\small\code{insert($77$)}}};
			\draw[->,>=stealth,line width=.6pt,decorate,decoration={coil,aspect=0,amplitude=1pt, segment length=1.5mm,post length=1mm,pre length=1mm}] (lab) -- (node42);
		}
		\begin{tikzpicture}[flattrees, baseline=(node42.center)]
			\node (root) [] {}
				child {
					node (node5) [treenode] {$5$} edge from parent[draw=none]
						child { node {} edge from parent[draw=none] }
						child [treeptr] {
							node (node42) [treenode] {$42$}
								child { node {} edge from parent[draw=none] }
								child {
									node (node77) [treenode] {$77$} 
										child { node {} edge from parent[draw=none] }
										child {
											node (node99) [treenode] {$99$} 
										}
								}
						}
				}
			;
			\node (rootN) at (root) [treenode,xshift=\MYXSHIFT] {$\infty$};
			\draw[treeptr,densely dotted] (rootN) -- (node5);
			\MYLAB{.5,-2.75}
			\begin{scope}[on background layer]
				\draw[listptr] (node5) edge[out=-0,in=180,looseness=2] (node42);
				\draw[listptr] (node42) edge[out=-0,in=180,looseness=2] (node77);
				\draw[listptr] (node77) edge[out=-0,in=180,looseness=2] (node99);
				\draw[listptr] (node99) edge[out=60,in=180,looseness=1.0] (rootN);
			\end{scope}
		\end{tikzpicture}
		\MYSTEP{a}
		\begin{tikzpicture}[flattrees, baseline=(node77.center)]
			\node (root) [] {}
				child {
					node (node5) [treenode] {$5$} edge from parent[draw=none]
						child { node {} edge from parent[draw=none] }
						child [treeptr] {
							node (node77) [treenode] {$77$}
								child {
									node (node42) [treenode] {$42$} 
								}
								child {
									node (node99) [treenode] {$99$} 
								}
						}
				}
			;
			\node (rootN) at (root) [treenode,xshift=\MYXSHIFT] {$\infty$};
			\draw[treeptr,densely dotted] (rootN) -- (node5);
			\MYLAB{1.15,-2.75}
			\begin{scope}[on background layer]
				\draw[listptr] (node5) edge[out=-0,in=180,looseness=1.2] (node42);
				\draw[listptr] (node42) edge[out=0,in=180,looseness=2] (node77);
				\draw[listptr] (node77) edge[out=-0,in=180,looseness=2] (node99);
				\draw[listptr] (node99) edge[out=60,in=180,looseness=1.2] (rootN);
			\end{scope}
		\end{tikzpicture}
		\MYSTEP{b}
		\begin{tikzpicture}[flattrees, baseline=(node77.center)]
			\node (root) [] {}
				child {
					node (node5) [treenode] {$5$} edge from parent[draw=none]
						child { node {} edge from parent[draw=none] }
						child [treeptr] {
							node (node77) [treenode] {$77$}
								child {
									node (node42) [treenode] {$42$} 
								}
								child {
									node (node99) [treenode] {$99$} 
								}
						}
				}
			;
			\node (rootN) at (root) [treenode,xshift=\MYXSHIFT] {$\infty$};
			\draw[treeptr,densely dotted] (rootN) -- (node5);
			\MYMARK{node77}
			\MYLOCK{node5}
			\MYLOCK{node77}
			\MYLOCK{node99}
			\MYLAB{1.15,-2.75}
			\begin{scope}[on background layer]
				\draw[listptr] (node5) edge[out=-0,in=180,looseness=1.2] (node42);
				\draw[listptr] (node42) edge (node99);
				\draw[listptr] (node99) edge[out=60,in=180,looseness=1.2] (rootN);
			\end{scope}
		\end{tikzpicture}
		\MYSTEP{c}
		\begin{tikzpicture}[flattrees, baseline=(node77.center)]
			\node (root) [] {}
				child {
					node (node5) [treenode] {$5$} edge from parent[draw=none]
						child { node {} edge from parent[draw=none] }
						child [treeptr] {
							node (node77) [treenode] {$77$}
								child {
									node (node42) [treenode] {$42$} 
										child { node {} edge from parent[draw=none] }
										child {
											node (copy) [treenode] {$77$}
										}
								}
								child {
									node (node99) [treenode] {$99$} 
								}
						}
				}
			;
			\node (rootN) at (root) [treenode,xshift=\MYXSHIFT] {$\infty$};
			\draw[treeptr,densely dotted] (rootN) -- (node5);
			\MYMARK{node77}
			\MYMARK{node77}
			\MYLOCK{node5}
			\MYLOCK{node77}
			\MYLOCK{node99}
			\MYLOCKP{node42}
			\MYLOCKP{copy}
			\begin{scope}[on background layer]
				\draw[listptr] (node5) edge[out=-0,in=180,looseness=1.2] (node42);
				\draw[listptr] (node42) edge[out=-0,in=180,looseness=2] (copy);
				\draw[listptr] (copy) edge[out=0,in=180,looseness=2] (node99);
				\draw[listptr] (node99) edge[out=60,in=180,looseness=1.2] (rootN);
			\end{scope}
		\end{tikzpicture}
	\endgroup}
	\caption{%
		Sequence of events to produce a tree overlay containing two versions of value $77$, a marked one and an unmarked one.
		Which version is found by \code{contains} depends on whether or not the tree traversal experiences rotations, yielding non-linearizable results.
		\label{fig:lotree:bug-duplicate}
	}
\end{figure}
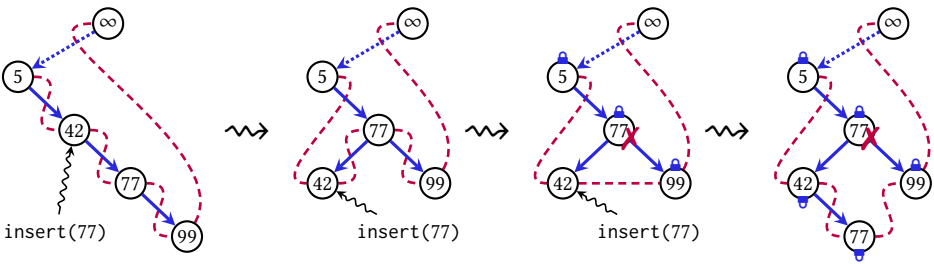

We make the malicious scenario precise.
To that end, consider \Cref{fig:lotree:bug-duplicate}.
In the first (leftmost) state, node $77$ is logically contained in the data structure.
Moreover, there is an \code{insert($77$)} underway whose tree traversal is currently at node $42$.
The second state is the result of a left rotation on node $42$.
Next, a \code{delete($77$)} starts.
Its tree traversal finds node $77$, and subsequently marks and unlinks it from the logical ordering.
Before the deletion removes node $77$ from the tree (\Cref{code:lotree:removeFromTree}), the insertion continues.
Ominously, the insertion manages to proceed: \code{locate} is able to validate that $77$ should be inserted between node $42$ and its now successor $99$.
The insertion will insert $77$ into the logical ordering and into the tree as a child of node $42$ which is situated in the left subtree of the marked node $77$ whose deletion is stalled.
This is the last (rightmost) state in \Cref{fig:lotree:bug-duplicate}.
Note that this scenario is not prevented by the $\treesel$s acquired by \code{prepareTreeDeletion} and \code{prepareTreeInsertion} (in \Cref{fig:lotree:bug-duplicate}, the $\treesel$s held by \code{delete} resp. \code{insert} according to \citet{DBLP:conf/ppopp/DrachslerVY14} are marked with ${\footnotesize\dellock}$ resp. ${\footnotesize\inslock}$).
In the last state, the original implementation of \code{contains($77$)}, which coincides with the one from \Cref{fig:lotree:impl} without \Cref{code:lotree:fix-contains}, produces non-linearizable results.
The problem is this: without rotations, \code{traverse($77$)} will return the marked node $77$.
Hence, \code{traverse($77$)} will not follow the logical ordering and returns $\false$ because the found node is marked.
This is not linearizable.
To see why, consider the execution of the above scenario:
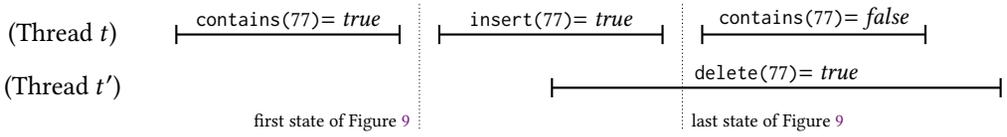
\[
	\begin{tikzpicture}
		\draw[|-|,thick] (0,0) -- node[above,yshift=-.7mm] {\small\code{contains($77$)}$=\true$} (3,0);
		\draw[|-|,thick] (3.5,0) -- node[above,yshift=-.7mm] {\small\code{insert($77$)}$=\true$} (6.5,0);
		\draw[|-|,thick] (7,0) -- node[above,yshift=-.7mm] {\small\code{contains($77$)}$=\false$} (10,0);
		\draw[|-|,thick] (5,-.7) -- node[above,yshift=-.7mm] {\small\code{delete($77$)}$=\true$} (11,-.7);
		\node at (-1.5,0) {(Thread $t$)};
		\node at (-1.5,-.7) {(Thread $t'$)};
		\draw[densely dotted] (3.25,.4) -- (3.25,-1.25);
		\node at (3.25,-1.15) [anchor=east, align=right] {{\scriptsize first state of \Cref{fig:lotree:bug-duplicate}}};
		\draw[densely dotted] (6.75,.4) -- (6.75,-1.25);
		\node at (6.75,-1.15) [anchor=west, align=left] {{\scriptsize last state of \Cref{fig:lotree:bug-duplicate}}};
	\end{tikzpicture}
\]
The first \code{contains($77$)} of thread $t$ is executed in the first state of \Cref{fig:lotree:bug-duplicate}, certifying that $77$ is indeed in the data structure.
Then, we perform the insertion and deletion of $77$ concurrently as described above.
After the insertion is finished, the same thread starts \code{contains($77$)} which returns $\false$, again as described above.
There are three possible linearizations of that execution: \[
	\begin{array}{llll}
		\mymathtt{contains($77$)}=\true
			;& \makeBug{\mymathtt{insert($77$)}=\true}
			;& \mymathtt{contains($77$)}=\false
			;& \mymathtt{delete($77$)}=\true
		\\
		\mymathtt{contains($77$)}=\true
			;& \makeBug{\mymathtt{insert($77$)}=\true}
			;& \mymathtt{delete($77$)}=\true
			;& \mymathtt{contains($77$)}=\false
		\\
		\mymathtt{contains($77$)}=\true
			;& \mymathtt{delete($77$)}=\true
			;& \mymathtt{insert($77$)}=\true
			;& \makeBug{\mymathtt{contains($77$)}=\false}
	\end{array}
\]
It is easy to see that all linearizations \makeBug{violate} the sequential specification of a set data type, meaning that the implementation is not linearizable.

The above linearizations also reveal that we can alleviate the problem by making the second \code{contains} return $\true$ so that the last linearization complies with the sequential specification of a set data type.
Our implementation from \Cref{fig:lotree:impl} achieves this by adding \Cref{code:lotree:fix-contains}: after the tree traversal, the logical ordering is followed ($\predsel$ fields) until an unmarked node is encountered.
This ensures that the final result is not \emph{confused} by concurrent deletions.
Once at an unmarked node, \code{contains} proceeds as devised by \citet{DBLP:conf/ppopp/DrachslerVY14}.

\smartparagraph{Bug 2: Insertion Order}
\looseness=-1
\citet{DBLP:journals/pacmpl/FeldmanKE0NRS20} identified another bug in the \code{insert} method.
In the original version by \citet{DBLP:conf/ppopp/DrachslerVY14}, new nodes are inserted first into the backward logical ordering and then into the forward one (compared to \Cref{fig:lotree:impl}, \Cref{code:lotree:fix-insert-succ,code:lotree:fix-insert-pred} are reversed).
To see why this is problematic, assume an insertion of a new node $\pnode$ with value $\vk$ between nodes $\pprev$ and $\pnext$ already linked $\predof{\pnext}$ to $\pnode$ but $\succof{\pprev}$ is still pointing to $\pnext$.
Then, \code{contains($\vk$)} will find $\pnode$ only if the tree traversal takes it to nodes that appear after $\pnext$ in the logical order.
For earlier nodes, \code{contains} will only follow $\succsel$ fields which cannot yet reach $\pnode$.
It is easy to see that this violates linearizability.

We fixed this bug by changing the order in which $\pnode$ is linked into the logical ordering, cf. \Cref{code:lotree:fix-insert-succ,code:lotree:fix-insert-pred}.
\citet{DBLP:journals/pacmpl/FeldmanKE0NRS20} apply the same fix.
However, they also change \code{insert} to link new nodes first into the tree overlay and then into the logical ordering (without modifying \code{contains}).
This violates linearizability:
if a new node $\pnode$ with value $\vk$ is inserted into the tree but not yet into the logical ordering, method \code{contains} will find $\vk$ if and only if it is not affected by concurrent rotations.
This gives rise to a linearizabilty violation similar to the one above.


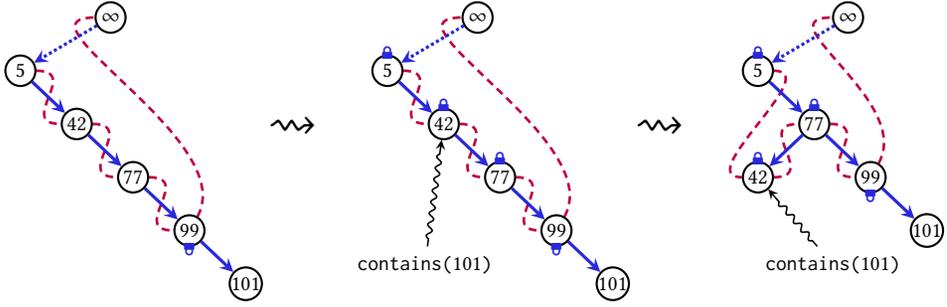
\begin{figure}[tb]
	{\begingroup
		\centering
		\newcommand{\MYXSHIFT}{15mm}
		\newcommand{\MYSTEP}[1]{%
			\hspace{-0mm}%
			\raisebox{-4.5pt}{\stackon{\bigleadsto}{\text{}}}%
			\hspace{-0mm}%
		}
		\newcommand{\MYMARK}[1]{%
			\node at (#1.south east) [yshift=-2.5pt] {\scalebox{1.2}{\textcolor{colorList}{\rawsymbolNo}}};
		}
		\newcommand{\MYLOCK}[1]{%
			\node at (#1.north) [xshift=0pt,yshift=1pt] {{\color{colorTree}\tiny\dellock}};
		}
		\newcommand{\MYLOCKP}[1]{%
			\node at (#1.south) [xshift=0pt,yshift=-1pt] {{\color{colorTree}\tiny\inslock}};
		}
		\newcommand{\MYLAB}[2][]{%
			\node (lab) at (#2) [#1] {{\small\code{contains($101$)}}};
			\draw[->,>=stealth,line width=.6pt,decorate,decoration={coil,aspect=0,amplitude=1pt, segment length=1.5mm,post length=1mm,pre length=1mm}] (lab) -- (node42);
		}
		\begin{tikzpicture}[flattrees, baseline=(node42.center)]
			\node (root) [] {}
				child {
					node (node5) [treenode] {$5$} edge from parent[draw=none]
						child { node {} edge from parent[draw=none] }
						child [treeptr] {
							node (node42) [treenode] {$42$}
								child { node {} edge from parent[draw=none] }
								child {
									node (node77) [treenode] {$77$} 
										child { node {} edge from parent[draw=none] }
										child {
											node (node99) [treenode] {$99$} 
												child { node {} edge from parent[draw=none] }
												child {
													node (node101) [treenode] {$101$} 
												}
										}
								}
						}
				}
			;
			\node (rootN) at (root) [treenode,xshift=\MYXSHIFT] {$\infty$};
			\draw[treeptr,densely dotted] (rootN) -- (node5);
			\MYLOCKP{node99}
			\begin{scope}[on background layer]
				\draw[listptr] (node5) edge[out=-0,in=180,looseness=2] (node42);
				\draw[listptr] (node42) edge[out=-0,in=180,looseness=2] (node77);
				\draw[listptr] (node77) edge[out=-0,in=180,looseness=2] (node99);
				\draw[listptr] (node99) edge[out=60,in=180,looseness=1.0] (rootN);
			\end{scope}
		\end{tikzpicture}
		\MYSTEP{a}
		\begin{tikzpicture}[flattrees, baseline=(node42.center)]
			\node (root) [] {}
				child {
					node (node5) [treenode] {$5$} edge from parent[draw=none]
						child { node {} edge from parent[draw=none] }
						child [treeptr] {
							node (node42) [treenode] {$42$}
								child { node {} edge from parent[draw=none] }
								child {
									node (node77) [treenode] {$77$} 
										child { node {} edge from parent[draw=none] }
										child {
											node (node99) [treenode] {$99$} 
												child { node {} edge from parent[draw=none] }
												child {
													node (node101) [treenode] {$101$} 
												}
										}
								}
						}
				}
			;
			\node (rootN) at (root) [treenode,xshift=\MYXSHIFT] {$\infty$};
			\draw[treeptr,densely dotted] (rootN) -- (node5);
			\MYLAB{.5,-3.25}
			\MYLOCKP{node99}
			\MYLOCK{node5}
			\MYLOCK{node42}
			\MYLOCK{node77}
			\begin{scope}[on background layer]
				\draw[listptr] (node5) edge[out=-0,in=180,looseness=2] (node42);
				\draw[listptr] (node42) edge[out=-0,in=180,looseness=2] (node77);
				\draw[listptr] (node77) edge[out=-0,in=180,looseness=2] (node99);
				\draw[listptr] (node99) edge[out=60,in=180,looseness=1.0] (rootN);
			\end{scope}
		\end{tikzpicture}
		\MYSTEP{b}
		\begin{tikzpicture}[flattrees, baseline=(node77.center)]
			\node (root) [] {}
				child {
					node (node5) [treenode] {$5$} edge from parent[draw=none]
						child { node {} edge from parent[draw=none] }
						child [treeptr] {
							node (node77) [treenode] {$77$}
								child {
									node (node42) [treenode] {$42$} 
								}
								child {
									node (node99) [treenode] {$99$} 
										child { node {} edge from parent[draw=none] }
										child {
											node (node101) [treenode] {$101$} 
										}
								}
						}
				}
			;
			\node (rootN) at (root) [treenode,xshift=\MYXSHIFT] {$\infty$};
			\draw[treeptr,densely dotted] (rootN) -- (node5);
			\MYLAB{1.0,-3.25}
			\MYLOCKP{node99}
			\MYLOCK{node5}
			\MYLOCK{node42}
			\MYLOCK{node77}
			\begin{scope}[on background layer]
				\draw[listptr] (node5) edge[out=-0,in=180,looseness=1.2] (node42);
				\draw[listptr] (node42) edge[out=0,in=180,looseness=2] (node77);
				\draw[listptr] (node77) edge[out=-0,in=180,looseness=2] (node99);
				\draw[listptr] (node99) edge[out=60,in=180,looseness=1.2] (rootN);
			\end{scope}
		\end{tikzpicture}
	\endgroup}
	\caption{%
		Sequence of events to show that inserting value $101$ into the tree overlay before inserting it into the logical ordering results in \code{contains} finding the value only if it does not experiences rotations, yielding non-linearizable results.
		\label{fig:lotree:bug-feldman}
	}
\end{figure}

To see linearizability violation, consider \Cref{fig:lotree:bug-feldman}.
In the first (leftmost) state, an \code{insert($101$)} has already linked value $101$ into the tree overlay but not yet into the logical ordering---note that this is the order proposed by \citet{DBLP:journals/pacmpl/FeldmanKE0NRS20} and differs from our one in \Cref{fig:lotree:impl}.
(The $\footnotesize\inslock$ lock is the one held by \code{insert($101$)}.)
In the second state, there is a \code{contains($101$)} underway and it reached node $42$.
For the last state, a rotation was executed.
(The $\footnotesize\dellock$ locks are the ones held by the rotation.)
As a consequence of this rotation, the \code{contains($101$)} currently at node $42$ will no longer be able to find $101$.
To obtain a linearizability violation, consider the following execution:
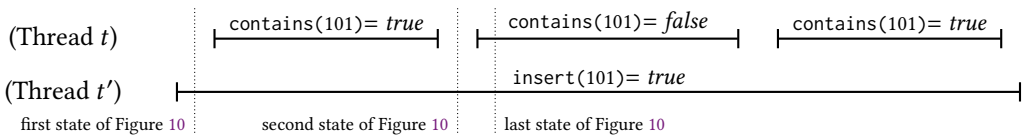
\[
	\begin{tikzpicture}
		\draw[|-|,thick] (.5,0) -- node[above,yshift=-.7mm] {\small\code{contains($101$)}$=\true$} (3.5,0);
		\draw[|-|,thick] (4.0,0) -- node[above,yshift=-.7mm] {\small\code{contains($101$)}$=\false$} (7.5,0);
		\draw[|-|,thick] (8,0) -- node[above,yshift=-.7mm] {\small\code{contains($101$)}$=\true$} (11,0);
		\draw[|-|,thick] (0,-.7) -- node[above,yshift=-.7mm] {\small\code{insert($101$)}$=\true$} (11.25,-.7);
		\node at (-1.5,0) {(Thread $t$)};
		\node at (-1.5,-.7) {(Thread $t'$)};
		\draw[densely dotted] (0.25,.4) -- (0.25,-1.25);
		\node at (0.25,-1.15) [anchor=east, align=right] {{\scriptsize first state of \Cref{fig:lotree:bug-feldman}}};
		\draw[densely dotted] (3.75,.4) -- (3.75,-1.25);
		\node at (3.75,-1.15) [anchor=east, align=left] {{\scriptsize second state of \Cref{fig:lotree:bug-feldman}}};
		\draw[densely dotted] (4.25,.4) -- (4.25,-1.25);
		\node at (4.25,-1.15) [anchor=west, align=left] {{\scriptsize last state of \Cref{fig:lotree:bug-feldman}}};
	\end{tikzpicture}
\]
Here, the first \code{contains($101$)} returns $\true$ as it can find $101$ via the tree overlay.
The second \code{contains($101$)} is the one depicted in \Cref{fig:lotree:bug-feldman}, which cannot find $101$ due to the rotation it experienced, and returns $\false$.
The last \code{contains($101$)} returns $\true$ as can find $101$ via the tree overlay, too.
We point out that the implementation of \code{contains} given by \citet{DBLP:journals/pacmpl/FeldmanKE0NRS20}, like ours from \Cref{fig:lotree:impl}, does not consider the logical ordering if node $101$ is found via the tree overlay and can thus return $\true$ indeed.
There are four possible linearizations for the above execution: \[
	\begin{array}{llll}
		\mymathtt{insert($77$)}=\true
			;& \mymathtt{contains($77$)}=\true
			;& \makeBug{\mymathtt{contains($77$)}=\false}
			;& \mymathtt{contains($77$)}=\true
		\\
		\mymathtt{contains($77$)}=\true
			;& \mymathtt{insert($77$)}=\true
			;& \makeBug{\mymathtt{contains($77$)}=\false}
			;& \mymathtt{contains($77$)}=\true
		\\
		\mymathtt{contains($77$)}=\true
			;& \makeBug{\mymathtt{contains($77$)}=\false}
			;& \mymathtt{insert($77$)}=\true
			;& \mymathtt{contains($77$)}=\true
		\\
		\mymathtt{contains($77$)}=\true
			;& \makeBug{\mymathtt{contains($77$)}=\false}
			;& \mymathtt{contains($77$)}=\true
			;& \mymathtt{insert($77$)}=\true
	\end{array}
\]
It is easy to see that all linearizations \makeBug{violate} the sequential specification of a set data type, meaning that the implementation is not linearizable.
Overall, this means that inserting into the tree overlay before inserting into the logical ordering as done by \citet{DBLP:journals/pacmpl/FeldmanKE0NRS20} is incorrect.
It is worth pointing out that the linearizability violation is independent of the order in which new nodes are inserted into the logical ordering ($\succsel$ first vs. $\predsel$ first).


\newcommand{\lockvar}{l}
\newcommand{\tid}{\mathit{tid}}
\newcommand{\cas}[3]{\mathsf{CAS}(#1, #2, #3)}

\section{A Control-Flow-Sensitive Generalization}\label{Section:Generalization}
Temporal interpolation derives information between $\weakpastof{\apred}$ and $\nowof{\apredp}$ from an abstraction of the program of interest, namely $\stmtof{\theInterference}$. 
This abstraction is control-flow insensitive, and there are situations in which it is too rough. 
Particularly problematic seem to be local computations and mutually exclusive accesses.
As for the mutual exclusion, consider a data structure in which a node's mark field is protected by the node's lock and may be set from false to true and from true to false. 
Imagine we find $\weakpastof{\apred}\cap\nowof{\apredp}$, where $\apred$ expresses, amongst other things, that we have the lock, $n.$\code{lock} $\mapsto 1$, and $\apredp$ says that the node is unmarked, $n.$\code{mark} $\mapsto 0$. 
The goal is to derive $\weakpastof{(\apred\cap\apredp)}$. 
The control-flow insensitive~\ruleref{temporal-interpolation} will not allow us to do so.
Predicate $\apred$ may refer to other fields of node $n$ that are not protected by the lock, and so the predicate will not be stable beyond the moment in the past where we find it.  
However, we will fail to conclude that the mark field was $0$ in that moment. 
The reason is that the self-interferences in $\stmtof{\theInterference}$ may aribtrarily release the lock held in $\apred$,
and then the mark field may experience arbitrary changes on the way to $\apredp$. 
With the control-flow sensitive version of temporal interpolation that we develop below, Rule~\ruleref{temporal-interpolation-cf},  we will be able to conclude that predicate $\apredp$ held true already in the moment we found $\apred$, and we thus have $\weakpastof{(\apred\cap\apredp)}$. 
The reasoning is as follows.
With the control flow at hand, we know that the thread of interest has not released the lock on the way from~$\apred$ to~$\apredp$.
This means no interference can modify the mark field.
We also know that the thread of interest has not modified the mark field. 
Together, the mark field was not changed on the way from $\apred$ to $\apredp$.

To incorporate control-flow information into our program abstraction, the idea is to modify the past predicate $\pastof{\apred}$ to a so-called \emph{history predicate} $\histof{\apred}{\astmt}$. 
The history predicate is meant to say that there has been a moment in the computation in which $\apred$ was true, and from that moment on the thread has executed a sequence of commands from program $\astmt$.  
This latter information is what will make temporal information control-flow sensitive. 
To formalize the semantics of the predicate, we need to adapt the separation algebra. 

Given a separation algebra $(\setstates, \statemult, \emp)$ and a (potentially infinite) set of commands $\setcom$, we define the separation algebra of \emph{histories} $\sethist\defeq\sethisteps.\setstates$ with $\sethisteps\defeq(\setstates^+.\setcom)^*.\setstates^*$.  Histories interleave non-empty sequences of states with commands. 
The intention behind this definition will become clear in a moment. 
The multiplication of histories is similar to the one for computations: we share the past and use the multiplication from the given separation algebra in the current state. 
It is defined, $\ahist_1.\astate_1\statemultdef \ahist_2.\astate_2$, if $\ahist_1=\ahist_2$ and $\astate_1\statemultdef\astate_2$, and in this case yields $\ahist_1.\astate_1\statemult\ahist_2.\astate_2\defeq \ahist_1.(\astate_1\mstar\astate_2)$. 
The set of units is $\emphist\defeq\sethisteps.\emp$.  
\begin{lemma}
If $(\setstates, \statemult, \emp)$ is a separation algebra, so is $(\sethist, \statemult, \emphist)$. 
\end{lemma}
The idea of a history $\ahist=\astateseq_1.\acom_1\astateseq_2\ldots\astateseq_n.\acom_n.\astateseq_{n+1}$ is to record the commands executed by the thread of interest. 
This means the state change from $\lastof{\astateseq_i}$ to $\firstof{\astateseq_{i+1}}$ is due to an execution of command $\acom_i$. 
We lift the semantics of commands $\acom\in\setcom$ to histories accordingly:
\begin{align*}
\csemof{\acom}(\ahist.\astate)\quad\defeq\quad \setcond{\ahist.\astate.\acom.\astate'}{\astate'\in\semof{\acom}(\ahist.\astate)}
\end{align*}
The state changes within the non-empty sequences of states $\astateseq_i$ are due to interferences from other threads. 
We do not record the command used in the interference, with the idea that the history predicate is meant to track thread-local information. 
The definition is as expected.

The \emph{history predicate} takes as input a state predicate $\apred\subseteq\setstates$ and a program $\astmt$ over $\setcom$:
\begin{align*}
\histof{\apred}{\astmt}\quad \defeq\quad &\sethisteps.\setcond{\astateseq_1.\acom_1.\astateseq_2\ldots \astateseq_{n}.\acom_n.\astateseq_{n+1}}{\firstof{\astateseq_1}\in\apred\wedge  \acom_1\ldots\acom_n\in\astmt}.
\end{align*}
Here, we understand $\astmt$ as a regular language and write $\acom_1\ldots\acom_n\in\astmt$ for membership, meaning the program is run to completion resp. a finite automaton for the language accepts the sequence of commands. 
We also write $\astmt_1\subseteq\astmt_2$ for the corresponding language inclusion.
As a special case, we may have the empty sequence of commands and $\apred$ holding in the current state. 
This means the history predicate has a weak understanding of the past, similar to $\weakpastof{\apred}$.  
\begin{lemma}
  The history predicate has the following properties.
  \begin{inparaenum}[(i)]
    \item[(i)] It is monotonic in both components: $\apred_1\subseteq\apred_2$ and $\astmt_1\subseteq\astmt_2$ imply $\histof{\apred_1}{\astmt_1}\subseteq\histof{\apred_2}{\astmt_2}$. 
    \item[(ii)] The interplay with commands is as expected: $\csemof{\acom}(\histof{\apred}{\astmt})\subseteq\histof{\apred}{\seqof{\astmt}{\acom}}$.
    \item[(iii)] It is interference-free: $\isInterferenceFreeOf[\theInterference]{\histof{\apred}{\astmt}}$. 
    \item[(iv)] With the purpose to track the execution of commands, it is not and should not be frameable.
    \item[(v)] If $\apred$ is intuitionistic, so is $\histof{\apred}{\astmt}$.  
  \end{inparaenum}
\end{lemma}

Neither the separation algebra of histories nor the history predicate require the state changes in a history to respect the semantics of commands and interferences.  
The reason we have not made this requirement, again, is that we do not know the interferences until we have built up a proof for the overall program.  
Fortunately, the set of governed computations provides the missing information.
The intersection $\histof{\apred}{\astmt}\cap\governeddefnb$ will keep from $\histof{\apred}{\astmt}$ only the histories in which the state changes are due to the commands and interferences. 
As before, our proofs will keep the intersection with $\governeddefnb$ implicit, which means we can think about $\histof{\apred}{\astmt}$ as having the expected semantics without having to add the notational overhead.

There is a technicality: we have to slightly redefine the set of governed computations: 
\begin{align*} 
\governeddefnb\quad&\defeq\quad \csemof{\stmtof{\theInterference}}_{\theInterference}(\Sigma)\downarrow
\ . 
\end{align*}
Recall that the program $\stmtof{\theInterference}$ has commands $\commandof{\acpred, \acom} = \mymathtt{atomic}\{\,\mymathtt{assume}(\acpred\mstar\true); \acom\,\}$, which are now recorded in the history. 
The program of interest, in turn, has plain commands~$\acom$. 
For the intersection with $\governeddefnb$ to be meaningful, the projection operation $\downarrow$ strips the atomic block and the assumption from $\commandof{\acpred, \acom}$, which has the effect of recording the block as $\acom$. 

We also lift the now and past predicates $\nowof{\apred}$ and $\pastof{\apred}$ to histories. 
The definition is as expected, and the \Cref{Lemma:PreciseIntuitionistic,Lemma:SLOperators,Lemma:InterplayCapStarPast} continue to hold. 

The analogue of Inclusion~\eqref{Equation:TIRelative} that we would like to use for temporal interpolation is
\begin{align}
\histof{\apred}{\astmt} \cap \nowof{\apredp}\cap\governeddefnb\quad \subseteq\quad \weakpastof{\apredpp}\ .\label{Equation:TICSF}
\end{align}
The hypothesis that, if true for the set $\theInterference$, justifies this inclusion is
\begin{align*}
\proghypof{\apred}{\astmt}{\apredp}{\apredpp}\quad\defeq\quad
\theInterferenceVar\semcalc\hoareOf{\nowof{\apred}}{\enrichof{\astmt}{\theInterferenceVar}}{\nowof{\apredp}\rightarrow\weakpastof{\apredpp}}\ .
\end{align*}
The hypothesis has the same pre- and postcondition as $\weakhypof{\apred}{\apredp}{\apredpp}$, but replaces the program $\stmtof{\theInterference}$ by $\enrichof{\astmt}{\theInterference}$. 
This enriched program uses the control-flow as recorded in $\astmt$, but enriches the commands by information about the states in which they are executed. 
Technically, function $\mathsf{enrich}$ turns every command $\acom$ into a choice over $\commandof{\acpred, \acom}$ with $(\acpred, \acom)\in \theInterference$, and preserves the remaining programming constructs:
\begin{alignat*}{5}
\enrichof{\acom}{\theInterference}\ &\defeq\ \sum_{(\acpred, \acom)\in\theInterference}\commandof{\acpred, \acom}&\qquad
\enrichof{\seqof{\astmt_1}{\astmt_2}}{\theInterference}\ &\defeq\ \seqof{\enrichof{\astmt_1}{\theInterference}}{\enrichof{\astmt_2}{\theInterference}}\\
\enrichof{\loopof{\astmt}}{\theInterference}\ &\defeq\ \loopof{\enrichof{\astmt}{\theInterference}}&\qquad
\enrichof{\choiceof{\astmt_1}{\astmt_2}}{\theInterference}\ &\defeq\ \choiceof{\enrichof{\astmt_1}{\theInterference}}{\enrichof{\astmt_2}{\theInterference}}\ .
\end{alignat*}
It is worth noting that the non-deterministic choice in $\enrichof{\acom}{\theInterference}$ can be avoided if we uniquely label each command in the program of interest (and therefore record a single interference for it).
The analogue of \Cref{Lemma:ATISound} that will guarantee soundness of the control-flow sensitive temporal interpolation rule is this. 

\begin{lemma}\label{Lemma:GeneralizedTISound}
If $\hypholdsof{\theInterference}{\proghypof{\apred}{\astmt}{\apredp}{\apredpp}}$, then $\histof{\apred}{\astmt}\cap \nowof{\apredp}\cap \governeddef\subseteq\weakpastof{\apredpp}$.
\end{lemma} 

The control-flow sensitive version of temporal interpolation is:
\begin{mathpar}
 	\inferH{temporal-interpolation-cf}{
 		\apred, \apredp\text{ intuitionistic}
        }{
		\set{\acpred\cap\pastof{\apredpp}}, \emptyset, \set{(\lastof{\acpred}, \cskip)}, \set{\proghypof{\apred}{\astmt}{\apredp}{\apredpp}}\semcalcti\hoareOf{\acpred\cap\histof{\apred}{\astmt}\cap\nowof{\apredp}}{\cskip}{\acpred\cap\pastof{\apredpp}}
    }
\end{mathpar}
The rule expects a history predicate $\histof{\apred}{\astmt}$ together with $\nowof{\apredp}$ and allows us to conclude $\pastof{\apredpp}$ after a $\cskip$ step, provided the hypothesis $\proghypof{\apred}{\astmt}{\apredp}{\apredpp}$ can be shown to hold for the final set of interferences. 
The rule is used together with the program logic in \Cref{Figure:ProgramLogicTI}. 
Soundness follows very closely the argumentation for \ruleref{temporal-interpolation} in \Cref{Lemma:HistoryTracking}. 

Temporal interpolation only ensure that the predicate $\apredpp$ has been true some time in the past.
For linearizability proofs, it is important to know that this happens while the method executes. 
The proof of the hypothesis $\hoareOf{\nowof{\apred}}{\astmt(\theInterference)}{\nowof{\apredp}\rightarrow \weakpastof{\apredpp}}$ with $\astmt(\theInterference)$ being $\stmtof{\theInterference}$ or $\enrichof{\astmt}{\theInterference}$ already guarantees this: the premise $\nowof{\apred}$ in particular contains the computation consisting of a single state in $\apred$, so if the program takes us to $\nowof{\apredp}\rightarrow\weakpastof{\apredpp}$, then $\apredpp$ can only have been true in between the two moments in time. 
To encode this knowledge into the logical reasoning, a simple way is to work with ghost flags that are raised by ghost commands upon method start or in the moment $\apred$ became true. 
There is a detail: as the past operator $\pastof{\apred}$ has no constraints except for the one state $\apred$, we need to make explicit that in all moments before $\apred$ was true, the flag was down.
This would be done with a predicate of the form $\apredppp^+$. 
We prefer to keep the mechanism of flags implicit, taking for granted that temporal interpolation ensures the existence of appropriate moments.

Temporal interpolation resembles the rule of conjunction, and an interesting question is whether this analogy may lead to a lighter formulation of our proof principle. 
To make the analogy explicit, we would mimic temporal interpolation by executing the following steps:
\begin{inparaenum}[(i)]
\item conduct a proof in which $\hoareof{\acpred\cap \nowof{\apred}}{\astmt}{\acpredp\cap \nowof{\apredp}}$ holds, 
\item conduct a proof $\hoareof{\nowof{\apred}}{\astmt}{\nowof{\apredp}\rightarrow \pastof{\apredpp}}$, and finally
\item conjoin this proof with the original one, resulting in  $\hoareof{\acpred\cap \nowof{\apred}}{\astmt}{\acpredp\cap \nowof{\apredp}\cap\pastof{\apredpp}}$. 
\end{inparaenum}
Unfortunately, the simpler formulation of temporal interpolation has problems the solutions to which lead to the development we have presented. 
First, we will not see $\apred$ in the proof conducted in (i), because it is typically not interference-free.
So we will have to record its occurrence in a past predicate~$\weakpastof{\apred}$. 
But then we need a mechanism to identify the program between  $\apred$ and $\apredp$. 
History predicates offer such a mechanism.   
Another aspect is that in many cases $\nowof{\apredp}\rightarrow \pastof{\apredpp}$ can already be derived for a coarse abstraction of the program.
Self-interferences form such a coarse abstraction that allows for concise subproofs. 
Finally, the subproof conducted in (ii) does not have available the knowledge derived in the outer proof.
We may add the predicate $\acpred$ to the precondition $\nowof{\apred}$, but this means repeating the outer proof.
With the enrichment $\enrichof{\astmt}{\theInterference}$, we add all knowledge from the outer proof, including intermediary assertions, and can focus on the implication to be derived.



\section{Details of Section~\ref{sec:Preliminaries}}\label{Appendix:Preliminaries}

The transition rules among configurations are as follows:
	\begin{mathpar}
		\infrule{}{
			\pcStepOf{\acom}{\acom}{\cskip}
		}
		\and
		\infrule{}{
			\pcStepOf{\seqof{\cskip}{\astmt}}{\cskip}{\astmt}
		}
		\and
		\infrule{}{
			\pcStepOf{\loopof{\astmt}}{\cskip}{\choiceof{\cskip}{\seqof{\astmt}{\loopof{\astmt}}}}
		}
		\\
		\infrule{i \in \set{1,2}}{
			\pcStepOf{\choiceof{\astmt_1}{\astmt_2}}{\cskip}{\astmt_i}
		}
		\and
		\infrule{
			\pcStepOf{\astmt_1}{\acom}{\astmt_1'}
		}{
			\pcStepOf{\seqof{\astmt_1}{\astmt_2}}{\acom}{\seqof{\astmt'_1}{\astmt_2}}
		}
		\and\normalfont
		\infrule{
			\pcStepOf{\astmt_1}{\acom}{\astmt_2}\\
			(\asharedseq_2, \alocalseq_2)\in\csemOf{\acom}(\asharedseq_1, \alocalseq_1)\\\\
			\apc_2=\setcond{j\mapsto(\alocalseq.\alocal.\alocal,\astmt)}{\apc_1(j)=(\alocalseq.\alocal,\astmt)}
		}{
			(\asharedseq_1, \apc_1[i\mapsto(\alocalseq_1, \astmt_1)])
			\progStepRel (\asharedseq_2, \apc_2[i\mapsto(\alocalseq_2, \astmt_2)])
		}
	\end{mathpar}

The initial, accepting, and reachable configurations are defined by:
\begin{align*}
	\initset{\acpred}{\astmt}
		~&\defeq~
		\setcond{(\asharedseq, \apc)}{
			\forall i \; \exists \alocalseq.~~
			\apc(i) = (\alocalseq,\astmt) \wedge (\asharedseq,\alocalseq)\in\acpred
		}
	\\
	\acceptset{\acpredp}
		~&\defeq~
		\setcond{(\asharedseq, \apc)}{
			\forall i,\alocalseq.~~
			\apc(i)=(\alocalseq, \cskip)\Rightarrow (\asharedseq, \alocalseq)\in\acpredp \,
		}
	\\
	\reachset{\aconfig}
		~&\defeq~
		\setcond{\aconfig'}{\aconfig \progStepRel^* \aconfig'}
	\ .
\end{align*}

\begin{definition}
$\subModels\hoareOf{\acpred}{\astmt}{\acpredp}$, if $\reachset{\initset{\acpred}{\astmt}}\subseteq\acceptset{\acpredp}$. 
\end{definition}

The proof system due to \citet{DBLP:journals/pacmpl/MeyerWW22} consists of the following rules:
\begin{mathpar}
	\inferH{com}{
		\csemof{\acom}(\acpred)\subseteq \acpredp
	}{
		\setcompact{\acpredp}, \set{(\acpred, \acom)}\semCalc\hoareOf{\acpred}{\acom}{\acpredp}
	}
	\and
	\inferH{consequence}{
		\thePredicates', \theInterference'\semCalc\hoareOf{\acpred'}{\astmt}{\acpredp'}
		\quad
		\acpred\subseteq \acpred'
		\quad
		\acpredp'\subseteq \acpredp
		\quad
		\thePredicates'\subseteq \thePredicates
		\quad
		\theInterference'\subseteq\theInterference
	}{
		\thePredicates, \theInterference\semCalc\hoareOf{\acpred}{\astmt}{\acpredp}
	}
	\and
	\inferH{frame}{
		\thePredicates, \theInterference\semCalc\hoareOf{\acpred}{\astmt}{\acpredp}
		\quad 
		\acpredpp\text{ frameable}
	}{
		\thePredicates\mstar\acpredpp, \theInterference\mstar\acpredpp\semCalc\hoareOf{\acpred\mstar \acpredpp}{\astmt}{\acpredp\mstar\acpredpp}
	}\and
	\inferH{seq}{
		\thePredicates_1, \theInterference_1\semCalc\hoareOf{\acpred}{\astmt_1}{\acpredp}\;\;\;
		\thePredicates_2, \theInterference_2\semCalc\hoareOf{\acpredp}{\astmt_2}{\acpredpp}
	}{
		\setcompact{\acpredp}\cup\thePredicates_1\cup\thePredicates_2, \theInterference_1\cup\theInterference_2\semCalc\hoareOf{\acpred}{\astmt_1;\astmt_2}{\acpredpp}
	}\and
	\inferH{loop}{
		\thePredicates, \theInterference\semCalc\hoareOf{\acpred}{\astmt}{\acpred}
	}{
		\setcompact{\acpred}\cup\thePredicates, \theInterference\semCalc\hoareOf{\acpred}{\loopof{\astmt}}{\acpred}
	}\and
	\inferH{choice}{
		\thePredicates_1, \theInterference_1\semCalc\hoareOf{\acpred}{\astmt_1}{\acpredp}\;\;\;
		\thePredicates_2, \theInterference_2\semCalc\hoareOf{\acpred}{\astmt_2}{\acpredp}
	}{
		\thePredicates_1\cup\thePredicates_2, \theInterference_1\cup\theInterference_2\semCalc\hoareOf{\acpred}{\choiceof{\astmt_1}{\astmt_2}}{\acpredp}
	}
\end{mathpar}


\section{Proofs of Section~\ref{Section:TemporalInterpolation}}

\begin{proof}[Proof (of Lemma~\ref{Lemma:ProofStrategy})]
We use the fact that $\aninvpred\cap \nowof{\apredp}\subseteq\weakpastof{\apredpp}$ implies $\aninvpred\subseteq \nowof{\apredp}\rightarrow \weakpastof{\apredpp}$.
\end{proof}
\begin{proof}[Proof (of Lemma~\ref{Lemma:InterplayCapStarPast})]
$\subseteq$\quad 
Consider $\acomp.\astate\in \acpredp\mstar\acpredpp\cap \pastof{\apredpp}$. \\
Then $\acomp=\acomp_1.\astatep.\acomp_2$ with $\astatep\in\apredpp$.\\
Moreover, $\astate = \astate_1\mstar \astate_2$ with $\acomp.\astate_1\in \acpredp$ and $\acomp.\astate_2\in\acpredpp$.\\
We thus have $\acomp.\astate_1=\acomp_1.\astatep.\acomp_2.\astate_1\in \pastof{\apredpp}\cap\acpredp$.\\
Moreover, $\acomp.\astate=\acomp.\astate_1\mstar\acomp.\astate_2 \in (\pastof{\apredpp}\cap\acpredp)\mstar\acpredpp$.   \\\\
$\supseteq$\quad We have $(\acpredp\cap \pastof{\apredpp})\mstar\acpredpp\subseteq \acpredp\mstar\acpredpp\cap \pastof{\apredpp}\mstar\acpredpp$.\\
Since $\pastof{\apredpp}$ is intuitionistic by Lemma~\ref{Lemma:PreciseIntuitionistic}, we have $\acpredp\mstar\acpredpp\cap \pastof{\apredpp}\mstar\acpredpp\subseteq \acpredp\mstar\acpredpp\cap \pastof{\apredpp}$.
\end{proof}
\begin{proof}[Proof (of Lemma~\ref{Lemma:FrameElimination})]
The implication from right to left is by definition.
For the implication from left to right, we proceed by Noetherian induction on the height of the derivation tree for $\thePredicates, \theInterference, \theHyp \semcalcti\hoareOf{\acpred}{\astmt}{\acpredp}$.\\
The height of the derivation tree is the maximal number of consecutive rule applications leading to the correctness statement.\\[0.2cm]
{\bfseries Base case}\quad \\[0.2cm]
{\bfseries Case~\ruleref{com-ti}}\quad 
We have 
\begin{align*}
\infer[\ruleref{com-ti}]{\csemof{\acom}(\acpred)\subseteq\acpredp}{
\infer[\ruleref{frame-ti}]{\set{\acpredp}, \set{(\acpred, \acom)}, \emptyset \semcalcti\hoareOf{\acpred}{\acom}{\acpredp}}{\set{\acpredp}\mstar\acpredpp, \set{(\acpred, \acom)}\mstar\acpredpp, \emptyset \semcalcti\hoareOf{\acpred\mstar\acpredpp}{\acom}{\acpredp\mstar\acpredpp}\ .}}
\end{align*}
The task is to find a derivation tree that does not use \ruleref{frame-ti}.\\
The observation is that \ruleref{com-ti} can deal with the framed predicate $\acpred\mstar\acpredpp$ right away.\\
By the locality of commands, $\csemof{\acom}(\acpred)\subseteq\acpredp$ entails $\csemof{\acom}(\acpred\mstar\acpredpp)\subseteq\acpredp\mstar\acpredpp$.\\
So we get
\begin{align*}
\infer[\ruleref{com-ti}]{\csemof{\acom}(\acpred\mstar\acpredpp)\subseteq\acpredp\mstar\acpredpp}{\set{\acpredp\mstar\acpredpp}, \set{(\acpred\mstar\acpredpp, \acom)}, \emptyset \semcalctinf\hoareOf{\acpred\mstar\acpredpp}{\acom}{\acpredp\mstar\acpredpp\ .}}
\end{align*}
Note that $\set{\acpredp\mstar\acpredpp}=\set{\acpredp}\mstar\acpredpp$ and $ \set{(\acpred\mstar\acpredpp, \acom)}= \set{(\acpred, \acom)}\mstar\acpredpp$. \\
This derivation is \ruleref{frame}-free.\\[0.2cm]
{\bfseries Case~\ruleref{temporal-interpolation}}\quad For intuitionistic $\apred, \apredp$, and some $\apredpp$, we have 
\begin{align*}
\infer[\ruleref{temporal-interpolation}]{\apred, \apredp\text{ intuitionistic}}{
\infer[\ruleref{frame-ti}]{\set{\acpred\cap \pastof{\apredpp}}, \set{(\acpred, \cskip)}, \set{\weakhypof{\apred}{\apredp}{\apredpp}} \semcalcti\hoareOf{\acpred\cap\weakpastof{\apred}\cap\nowof{\apredp}}{\cskip}{\acpred\cap \pastof{\apredpp}}}{\set{\acpred\cap \pastof{\apredpp}}\mstar\acpredpp, \set{(\acpred, \cskip)}\mstar\acpredpp, \set{\weakhypof{\apred}{\apredp}{\apredpp}} \semcalcti\hoareOf{(\acpred\cap\weakpastof{\apred}\cap\nowof{\apredp})\mstar\acpredpp}{\cskip}{(\acpred\cap\pastof{\apredpp})\mstar\acpredpp}\ .}}
\end{align*}
We apply \ruleref{temporal-interpolation} followed by \ruleref{consequence-ti}:
\begin{align*}
\infer[\ruleref{temporal-interpolation}]{\apred, \apredp\text{ intuitionistic}}{
\infer[\ruleref{consequence-ti}]{\set{\acpred\mstar\acpredpp\cap \pastof{\apredpp}}, \set{(\acpred\mstar\acpredpp, \cskip)}, \set{\weakhypof{\apred}{\apredp}{\apredpp}} \semcalcti\hoareOf{\acpred\mstar\acpredpp\cap\weakpastof{\apred}\cap\nowof{\apredp}}{\cskip}{\acpred\mstar\acpredpp\cap \pastof{\apredpp}}}{\set{\acpred\cap \pastof{\apredpp}}\mstar\acpredpp, \set{(\acpred, \cskip)}\mstar\acpredpp, \set{\weakhypof{\apred}{\apredp}{\apredpp}} \semcalcti\hoareOf{(\acpred\cap\weakpastof{\apred}\cap\nowof{\apredp})\mstar\acpredpp}{\cskip}{(\acpred\cap\pastof{\apredpp})\mstar\acpredpp}\ .}}
\end{align*}
The equality $\acpred\mstar\acpredpp\cap \pastof{\apredpp}=\set{\acpred\cap \pastof{\apredpp}}\mstar\acpredpp$ is Lemma~\ref{Lemma:PreciseIntuitionistic}.\\
It is also used in the postcondition.\\
For the interferences, $\set{(\acpred, \cskip)}\mstar\acpredpp=\set{(\acpred\mstar\acpredpp, \cskip)}$.\\
We indeed strengthen the precondition, as
\begin{align*}
(\acpred\cap\weakpastof{\apred}\cap\nowof{\apredp})\mstar\acpredpp \subseteq \acpred\mstar\acpredpp\cap\weakpastof{\apred}\mstar\acpredpp\cap\nowof{\apredp}\mstar\acpredpp\subseteq\acpred\mstar\acpredpp\cap\weakpastof{\apred}\cap\nowof{\apredp}.
\end{align*}
The second inclusion uses that $\weakpastof{\apred}$ and $\nowof{\apredp}$ are intuitionistic by Lemma~\ref{Lemma:PreciseIntuitionistic}.\\
The derivation is \ruleref{frame-ti}-free.\\[0.2cm]
{\bfseries Case~\ruleref{temporal-interpolation-unordered}}\quad Similar to the previous case.\\[0.2cm]
{\bfseries Induction step}\quad  We assume that for every correctness statement derived with a tree of height at most $n$, we have a derivation without \ruleref{frame-ti}.\\
We consider a correctness statement that is derived with a tree of height $n+1$ in which the last rule is~\ruleref{frame-ti}. \\
This means we have 
\begin{align*}
\infer[\ruleref{frame-ti}]{\thePredicates, \theInterference, \theHyp \semcalcti\hoareOf{\acpred}{\astmt}{\acpredp}}{\thePredicates\mstar\acpredpp, \theInterference\mstar\acpredpp, \theHyp \semcalcti\hoareOf{\acpred\mstar\acpredpp}{\astmt}{\acpredp\mstar\acpredpp}\ .}
\end{align*}
and the premise has a derivation of height $n$.\\
To eliminate this application of \ruleref{frame-ti}, we consider the rule application that lead to the premise.\\[0.2cm]
{\bfseries Case~\ruleref{frame-ti}}\quad Then for some $\acpredppp$ we have $\thePredicates=\thePredicates'\mstar\acpredppp$, $\theInterference=\theInterference'\mstar\acpredppp$, $\acpred=\acpred'\mstar\acpredppp$, and $\acpredp=\acpredp'\mstar\acpredppp$. \\
The derivation of height $n+1$ thus has the shape
\begin{align*}
\infer[\ruleref{frame-ti}]{\thePredicates', \theInterference', \theHyp \semcalcti\hoareOf{\acpred'}{\astmt}{\acpredp'}}{
\infer[\ruleref{frame-ti}]{\thePredicates'\mstar\acpredppp, \theInterference'\mstar\acpredppp, \theHyp\semcalcti\hoareOf{\acpred'\mstar\acpredppp}{\astmt}{\acpredp'\mstar\acpredppp}}{(\thePredicates'\mstar\acpredppp)\mstar\acpredpp, (\theInterference'\mstar\acpredppp)\mstar\acpredpp, \theHyp \semcalcti\hoareOf{(\acpred'\mstar\acpredppp)\mstar\acpredpp}{\astmt}{(\acpredp'\mstar\acpredppp)\mstar\acpredpp}\ .}}
\end{align*}
We frame $\acpredppp\mstar\acpredpp$ with a single application of \ruleref{frame-ti}:
\begin{align*}
\infer[\ruleref{frame-ti}]{\thePredicates', \theInterference', \theHyp \semcalcti\hoareOf{\acpred'}{\astmt}{\acpredp'}}{\thePredicates'\mstar(\acpredppp\mstar\acpredpp), \theInterference'\mstar(\acpredppp\mstar\acpredpp), \theHyp \semcalcti\hoareOf{\acpred'\mstar(\acpredppp\mstar\acpredpp)}{\astmt}{\acpredp'\mstar(\acpredppp\mstar\acpredpp)}\ .}
\end{align*}
Since separating conjunction is associative, this is the desired correctness statement.\\
The difference, however, is that now the derivation tree has height only $n$.\\
Thus, the induction hypothesis applies and yields a \ruleref{frame-ti}-free derivation. \\[0.2cm]
{\bfseries Case~\ruleref{loop-ti}}\quad Then the derivation tree of height $n+1$ ends with
\begin{align*}
\infer[\ruleref{loop-ti}]{\thePredicates', \theInterference, \theHyp \semcalcti\hoareOf{\acpred}{\astmt}{\acpred}}{
\infer[\ruleref{frame-ti}]{\set{\acpred}\cup \thePredicates', \theInterference, \theHyp \semcalcti\hoareOf{\acpred}{\loopof{\astmt}}{\acpred}}{(\set{\acpred}\cup \thePredicates')\mstar\acpredpp, \theInterference\mstar\acpredpp, \theHyp \semcalcti\hoareOf{\acpred\mstar\acpredpp}{\loopof{\astmt}}{\acpred\mstar\acpredpp}\ .}}
\end{align*}
Note that $\thePredicates$ from above is $\set{\acpred}\cup \thePredicates'$ and $\acpredp$ is $\acpred$. \\
We construct a different end of the derivation tree in which we first apply \ruleref{frame-ti} and then \ruleref{loop-ti}: 
\begin{align*}
\infer[\ruleref{frame-ti}]{\thePredicates', \theInterference, \theHyp \semcalcti\hoareOf{\acpred}{\astmt}{\acpred}}{
\infer[\ruleref{loop-ti}]{\thePredicates'\mstar\acpredpp, \theInterference\mstar\acpredpp, \theHyp \semcalcti\hoareOf{\acpred\mstar\acpredpp}{\astmt}{\acpred\mstar\acpredpp}}{\set{\acpred\mstar\acpredpp}\cup (\thePredicates'\mstar\acpredpp), \theInterference\mstar\acpredpp, \theHyp \semcalcti\hoareOf{\acpred\mstar\acpredpp}{\loopof{\astmt}}{\acpred\mstar\acpredpp}\ .}}
\end{align*}
Since $\set{\acpred\mstar\acpredpp}\cup (\thePredicates'\mstar\acpredpp)=(\set{\acpred}\cup \thePredicates')\mstar\acpredpp$, the result is the desired correctness statement.\\[0.2cm]
The application of \ruleref{frame-ti} in the rewritten proof occurs within a derivation tree of height $n$.\\
By the induction hypothesis, we get $\thePredicates'\mstar\acpredpp, \theInterference\mstar\acpredpp, \theHyp \semcalctinf\hoareOf{\acpred\mstar\acpredpp}{\astmt}{\acpred\mstar\acpredpp}$, 
meaning we can derive the intermediary correctness statement without \ruleref{frame-ti}.\\
Adding another application of \ruleref{loop-ti} keeps the derivation \ruleref{frame-ti}-free.\\[0.2cm]
{\bfseries Case~\ruleref{consequence-ti}}\quad Then the derivation tree of height $n+1$ ends with
\begin{align*}
\infer[\ruleref{consequence-ti}]{\thePredicates', \theInterference', \theHyp' \semcalcti\hoareOf{\acpred'}{\astmt}{\acpredp'}}{
\infer[\ruleref{frame-ti}]{\thePredicates, \theInterference, \theHyp \semcalcti\hoareOf{\acpred}{\astmt}{\acpredp}}{\thePredicates\mstar\acpredpp, \theInterference\mstar\acpredpp, \theHyp \semcalcti\hoareOf{\acpred\mstar\acpredpp}{\astmt}{\acpredp\mstar\acpredpp}\ ,}}
\end{align*}
where $\acpred\subseteq\acpred'$ and $\thePredicates'\subseteq\thePredicates$, $\theInterference'\subseteq\theInterference$, $\theHyp'\subseteq\theHyp$, and $\acpredp'\subseteq\acpredp$. \\
We construct a different end of the derivation tree with first \ruleref{frame-ti} and then \ruleref{consequence-ti}: 
\begin{align*}
\infer[\ruleref{frame-ti}]{\thePredicates', \theInterference', \theHyp' \semcalcti\hoareOf{\acpred'}{\astmt}{\acpredp'}}{
\infer[\ruleref{consequence-ti}]{\thePredicates'\mstar\acpredpp, \theInterference'\mstar\acpredpp, \theHyp'\semcalcti\hoareOf{\acpred'\mstar\acpredpp}{\astmt}{\acpredp'\mstar\acpredpp}}{\thePredicates\mstar\acpredpp, \theInterference\mstar\acpredpp, \theHyp \semcalcti\hoareOf{\acpred\mstar\acpredpp}{\astmt}{\acpredp\mstar\acpredpp}\ .}}
\end{align*}
To be able to apply \ruleref{consequence-ti}, we note that $\acpred\subseteq\acpred'$ entails $\acpred\mstar\acpredpp\subseteq \acpred'\mstar\acpredpp$, $\thePredicates'\subseteq\thePredicates$ entails $\thePredicates'\mstar\acpredpp\subseteq\thePredicates\mstar\acpredpp$, and similar for the other inclusions.\\[0.2cm]
The application of \ruleref{frame-ti} in the rewritten proof occurs within a derivation tree of height $n$.\\
By the induction hypothesis, we get $\thePredicates'\mstar\acpredpp, \theInterference'\mstar\acpredpp, \theHyp' \semcalctinf\hoareOf{\acpred'\mstar\acpredpp}{\astmt}{\acpredp'\mstar\acpredpp}$.\\
Adding another application of \ruleref{consequence-ti} keeps the derivation \ruleref{frame-ti}-free.\\[0.2cm]
{\bfseries Case~\ruleref{seq-ti}}\quad Then the derivation tree of height $n+1$ ends with
\begin{align*}
\infer[\ruleref{seq-ti}]{\thePredicates_1, \theInterference_1, \theHyp_1 \semcalcti\hoareOf{\acpred}{\astmt_1}{\acpredppp}
\quad
\thePredicates_2, \theInterference_2, \theHyp_2 \semcalcti\hoareOf{\acpredppp}{\astmt_2}{\acpredp}
}{
\infer[\ruleref{frame-ti}]{\set{\acpredppp}\cup \thePredicates_1\cup\thePredicates_2, \theInterference_1\cup\theInterference_2, \theHyp_1\cup \theHyp_2 \semcalcti\hoareOf{\acpred}{\seqof{\astmt_1}{\astmt_2}}{\acpredp}}{(\set{\acpredppp}\cup \thePredicates_1\cup\thePredicates_2)\mstar\acpredpp, (\theInterference_1\cup\theInterference_2)\mstar\acpredpp, \theHyp_1\cup \theHyp_2 \semcalcti\hoareOf{\acpred\mstar\acpredpp}{\seqof{\astmt_1}{\astmt_2}}{\acpredp\mstar\acpredpp}\ .}}
\end{align*}
So $\thePredicates$ from above is $\set{\acpredppp}\cup \thePredicates_1\cup\thePredicates_2$ and similar for the other components.\\
We construct a different end of the derivation tree in which we first apply \ruleref{frame-ti} to the correctness statements from the two branches and then \ruleref{seq-ti}: 
\begin{align*}
\infer{\infer[\ruleref{frame-ti}]{\thePredicates_1, \theInterference_1, \theHyp_1 \semcalcti\hoareOf{\acpred}{\astmt_1}{\acpredppp}}{\thePredicates_1\mstar\acpredpp, \theInterference_1\mstar\acpredpp,  \theHyp_1 \semcalcti\hoareOf{\acpred\mstar\acpredpp}{\astmt_1}{\acpredppp\mstar\acpredpp}}
\quad
\infer[\ruleref{frame-ti}]{\thePredicates_2, \theInterference_2, \theHyp_2 \semcalcti\hoareOf{\acpredppp}{\astmt_2}{\acpredp}}{\thePredicates_2\mstar\acpredpp, \theInterference_2\mstar\acpredpp, \theHyp_2 \semcalcti\hoareOf{\acpredppp\mstar\acpredpp}{\astmt_2}{\acpredp\mstar\acpredpp}}
}{\set{\acpredppp\mstar\acpredpp}\cup \thePredicates_1\mstar\acpredpp\cup\thePredicates_2\mstar\acpredpp, \theInterference_1\mstar\acpredpp\cup\theInterference_2\mstar\acpredpp, \theHyp_1\cup\theHyp_2 \semcalcti\hoareOf{\acpred\mstar\acpredpp}{\seqof{\astmt_1}{\astmt_2}}{\acpredp\mstar\acpredpp}\ .}
\end{align*}
Since $\set{\acpredppp\mstar\acpredpp}\cup (\thePredicates_1\mstar\acpredpp)\cup(\thePredicates_2\mstar\acpredpp)=(\set{\acpredppp}\cup \thePredicates_1\cup\thePredicates_2)\mstar\acpredpp$ and similar for the other components, the result is the desired correctness statement.\\[0.2cm]
The applications of \ruleref{frame-ti} in the rewritten proof occur within derivation trees of height at most~$n$.\\
By the induction hypothesis, we get $\thePredicates_1\mstar\acpredpp, \theInterference_1\mstar\acpredpp, \theHyp_1 \semcalctinf\hoareOf{\acpred\mstar\acpredpp}{\astmt_1}{\acpredppp\mstar\acpredpp}$ and $\thePredicates_2\mstar\acpredpp, \theInterference_2\mstar\acpredpp, \theHyp_2 \semcalctinf\hoareOf{\acpredppp\mstar\acpredpp}{\astmt_2}{\acpredp\mstar\acpredpp}$.\\
Adding another application of \ruleref{seq-ti} keeps the derivation \ruleref{frame-ti}-free.\\[0.2cm]
{\bfseries Case~\ruleref{choice-ti}}\quad Similar to the previous case.
\end{proof}
\begin{proof}[Proof (of Lemma~\ref{Lemma:HistoryTracking})]
We proceed by Noetherian induction on the height of \ruleref{frame-ti}-free derivations.\\[0.2cm]
{\bfseries Base case}\quad \\[0.2cm] 
{\bfseries Case~\ruleref{com-ti}}\quad Consider
\begin{align*}
\infer[\ruleref{com-ti}]{\csemof{\acom}(\acpred)\subseteq\acpredp}{
\set{\acpredp}, \set{(\acpred, \acom)}, \emptyset \semcalctinf\hoareOf{\acpred}{\acom}{\acpredp}\ .}
\end{align*}
Consider $\theInterference$ with $\set{(\acpred, \acom)}\subseteq\theInterference$. \\
We have $\csemof{\acom}(\acpred\cap\governeddef)\subseteq\csemof{\acom}(\acpred)\cap\governeddef\subseteq \acpredp\cap \governeddef$.\\
The latter inclusion is by the assumption $\csemof{\acom}(\acpred)\subseteq\acpredp$.\\
To see the former, consider $\acomp.\astate.\astate'\in \csemof{\acom}(\acpred\cap\governeddef)$.\\
Then $\acomp.\astate\in\governeddef$ by definition.\\
Moreover, we have $\acomp.\astate\in\acpred$.\\
Hence, the state change $\astate.\astate'$ is covered by the interference $\set{(\acpred, \acom)}$. \\
Since $\set{(\acpred, \acom)}\subseteq\theInterference$, we get $\acomp.\astate.\astate'\in\governeddef$ as required.\\[0.2cm]
The inclusion allows us to derive 
\begin{align*}
\infer[\ruleref{com}]{\csemof{\acom}(\acpred\cap\governeddef)\subseteq\acpredp\cap\governeddef}{
\infer[\ruleref{consequence}]{\set{\acpredp\cap\governeddef}, \set{(\acpred\cap\governeddef, \acom)} \semcalc\hoareOf{\acpred\cap\governeddef}{\acom}{\acpredp\cap\governeddef}}{
\set{\acpredp}\cap\governeddef, \theInterference \semcalc\hoareOf{\acpred\cap\governeddef}{\acom}{\acpredp\cap\governeddef}\ .}
}
\end{align*}
For \ruleref{consequence}, $\set{\acpredp\cap\governeddef} = \set{\acpredp}\cap\governeddef$ and $\set{(\acpred\cap\governeddef, \acom)}\subseteq\set{(\acpred, \acom)}\subseteq\theInterference$.\\[0.2cm] 
{\bfseries Case~\ruleref{temporal-interpolation}}\quad Consider 
\begin{align*}
\infer[\ruleref{temporal-interpolation}]{\apred, \apredp\text{ intuitionistic}}{
\set{\acpred\cap\pastof{\apredpp}}, \set{(\acpred, \cskip)}, \set{\weakhypof{\apred}{\apredp}{\apredpp}} \semcalctinf\hoareOf{\acpred\cap\weakpastof{\apred}\cap\nowof{\apredp}}{\cskip}{\acpred\cap\pastof{\apredpp}}\ .}
\end{align*}
Let $\theInterference$ be a set of interferences with $\set{(\acpred, \cskip)}\subseteq \theInterference$ so that $\hypholdsof{\theInterference}{\weakhypof{\apred}{\apredp}{\apredpp}}$. \\
To obtain an ordinary derivation, we first apply \ruleref{com} and get
\begin{align*}
\infer[\ruleref{com}]{\csemof{\cskip}(\acpred\cap\weakpastof{\apredpp}\cap \governeddef)\subseteq\acpred\cap\pastof{\apredpp}\cap \governeddef}{
\set{\acpred\cap\pastof{\apredpp}\cap \governeddef}, \set{(\acpred\cap\weakpastof{\apredpp}\cap \governeddef, \cskip)} \semcalc\hoareOf{\acpred\cap\weakpastof{\apredpp}\cap\governeddef}{\cskip}{\acpred\cap\pastof{\apredpp}\cap \governeddef}\ .}
\end{align*}
To see the inclusion in the premise, we have $\csemof{\cskip}(\acpred)\subseteq \acpred$, 
because $\acpred$ is frameable. \\
Then $\csemof{\cskip}(\weakpastof{\apredpp})\subseteq \pastof{\apredpp}$, since $\cskip$ adds an extra step. \\
We thus get $\csemof{\cskip}(\acpred\cap\weakpastof{\apredpp})\subseteq\acpred\cap\pastof{\apredpp}$. \\
We can add the governed computations with the same argument as in the previous case.\\\\
We now apply \ruleref{consequence} and get
\begin{align*}
\infer[\ruleref{consequence}]{
\set{\acpred\cap\pastof{\apredpp}\cap \governeddef}, \set{(\acpred\cap\weakpastof{\apredpp}\cap \governeddef, \cskip)}\semcalc\hoareOf{\acpred\cap\weakpastof{\apredpp}\cap\governeddef}{\cskip}{\acpred\cap\pastof{\apredpp}\cap \governeddef}}{
\set{\acpred\cap\pastof{\apredpp}}\cap \governeddef, \theInterference \semcalc\hoareOf{\acpred\cap\weakpastof{\apred}\cap\nowof{\apredp}\cap\governeddef}{\cskip}{\acpred\cap\pastof{\apredpp}\cap \governeddef}\ .}
\end{align*}
We have generalized the set of interferences and strengthened the precondition. \\
As for the latter, note that we have $\hypholdsof{\theInterference}{\weakhypof{\apred}{\apredp}{\apredpp}}$ by the assumption.\\
Hence, Lemma~\ref{Lemma:ATISound} applies and yields $\acpred\cap\weakpastof{\apred}\cap\nowof{\apredp}\cap\governeddef\subseteq \acpred\cap \weakpastof{\apredpp}\cap\governeddef$.\\
As for the former, we have $\set{(\acpred, \cskip)}\subseteq \theInterference$ by the assumption.\\
This implies $\set{(\acpred\cap\weakpastof{\apredpp}\cap \governeddef, \cskip)}\subseteq\theInterference$. \\[0.2cm]
{\bfseries Case~\ruleref{temporal-interpolation-unordered}}\quad The argumentation is similar to the previous case, but one has to show that $\hypholdsof{\theInterference}{\weakhypof{\apred}{\apredp}{\apredpp}}$ and $\hypholdsof{\theInterference}{\weakhypof{\apredp}{\apred}{\apredpp}}$ justify $\weakpastof{\apred}\cap\weakpastof{\apredp}\cap\governeddefnb\subseteq\weakpastof{\apredpp}$. 
To this end, consider a commputation $\acomp\in \weakpastof{\apred}\cap\weakpastof{\apredp}\cap\governeddefnb$. 
There has been a moment in which $\apred$ was true and a moment in which $\apredp$ was true.
Say $\apred$ was earlier. 
Then $\acomp=\acomp_1.\astate.\acomp_2.\astatep.\acomp_3$ with $\astate\in\apred$ and $\astatep\in\apredp$. 
We thus have $\acomp_1.\astate.\acomp_2.\astatep\in \weakpastof{\apred}\cap\nowof{\apredp}\cap\governeddefnb$. 
Lemma~\ref{Lemma:ATISound} applies and yields $\acomp_1.\astate.\acomp_2.\astatep\in\weakpastof{\apredpp}$. 
The weak past predicate does not change if we append $\acomp_3$, and so also $\acomp\in\weakpastof{\apredpp}$. \\[0.2cm]
{\bfseries Induction step}\quad  We assume that for every correctness statement derived with a tree of height at most $n$, potentially using temporal interpolation but not using \ruleref{frame-ti}, and for all larger sets of interferences that satisfy the hypotheses, we can give an ordinary derivation in which the pre- and postcondition are strengthend by an intersection with the corresponding set of governed computations.\\
We consider a correctness statement that is derived with a tree of height $n+1$ and perform an analysis along the last rule that has been applied.\\[0.2cm]
{\bfseries Case~\ruleref{loop-ti}}\quad Then the derivation tree of height $n+1$ ends with
\begin{align*}
\infer[\ruleref{loop-ti}]{\thePredicates', \theInterference', \theHyp \semcalctinf\hoareOf{\acpred}{\astmt}{\acpred}}{
    \set{\acpred}\cup\thePredicates', \theInterference', \theHyp \semcalctinf\hoareOf{\acpred}{\loopof{\astmt}}{\acpred}\ .
}
\end{align*}
So $\thePredicates$ from above is $\set{\acpred}\cup\thePredicates'$. \\
Consider $\theInterference$  with $\theInterference'\subseteq\theInterference$ and $\hypholdsof{\theInterference}{\theHyp}$.\\[0.2cm]
Since the derivation of $\thePredicates', \theInterference', \theHyp \semcalctinf\hoareOf{\acpred}{\astmt}{\acpred}$ has height $n$, the induction hypothesis yields
\begin{align*}
\thePredicates'\cap\governeddef, \theInterference \semcalc\hoareOf{\acpred\cap\governeddef}{\astmt}{\acpred\cap\governeddef}\ .
\end{align*}
An application of \ruleref{loop} yields
\begin{align*}
\infer[\ruleref{loop}]{\thePredicates'\cap\governeddef, \theInterference\semcalc\hoareOf{\acpred\cap\governeddef}{\astmt}{\acpred\cap\governeddef}
}{
    \set{\acpred\cap\governeddef}\cup(\thePredicates'\cap\governeddef), \theInterference \semcalc\hoareOf{\acpred\cap\governeddef}{\loopof{\astmt}}{\acpred\cap\governeddef}\ .
}
\end{align*}
Since $\set{\acpred\cap\governeddef}\cup(\thePredicates'\cap\governeddef)=(\set{\acpred}\cup\thePredicates')\cap\governeddef$, this is as desired.\\[0.2cm]
{\bfseries Case~\ruleref{consequence-ti}}\quad Then the derivation tree of height $n+1$ ends with
\begin{align*}
\infer[\ruleref{consequence-ti}]{\thePredicates', \theInterference'', \theHyp' \semcalctinf\hoareOf{\acpred'}{\astmt}{\acpredp'}}{
    \thePredicates, \theInterference', \theHyp \semcalctinf\hoareOf{\acpred}{\astmt}{\acpredp}\ ,
}
\end{align*}
where $\acpred\subseteq \acpred'$ and $\thePredicates'\subseteq \thePredicates$, $\theInterference''\subseteq\theInterference'$, $\theHyp'\subseteq\theHyp$, and $\acpredp'\subseteq\acpredp$.\\ 
Consider $\theInterference$ with $\theInterference'\subseteq\theInterference$ and $\hypholdsof{\theInterference}{\theHyp}$.\\[0.2cm]
Since the tree for $\thePredicates', \theInterference'', \theHyp' \semcalctinf\hoareOf{\acpred'}{\astmt}{\acpredp'}$ has height $n$, and since $\theInterference''\subseteq\theInterference$ with $\hypholdsof{\theInterference}{\theHyp'}$, the induction hypothesis yields
\begin{align*}
\thePredicates'\cap \governeddef, \theInterference \semcalc\hoareOf{\acpred'\cap \governeddef}{\astmt}{\acpredp'\cap \governeddef}.
\end{align*}
Since we have $\thePredicates'\subseteq \thePredicates$, we get $\thePredicates'\cap \governeddef\subseteq \thePredicates\cap \governeddef$, and similarly $\acpred\cap \governeddef\subseteq \acpred'\cap \governeddef$ and  $\acpredp'\cap \governeddef\subseteq \acpredp\cap \governeddef$.\\
This justifies an application of Rule~\ruleref{consequence} 
\begin{align*}
\infer[\ruleref{consequence}]{
\thePredicates'\cap \governeddef, \theInterference \semcalc\hoareOf{\acpred'\cap \governeddef}{\astmt}{\acpredp'\cap \governeddef}
}{
\thePredicates\cap \governeddef, \theInterference \semcalc\hoareOf{\acpred\cap \governeddef}{\astmt}{\acpredp\cap \governeddef}\ .
}
\end{align*}
The resulting correctness statement is as desired.\\[0.2cm]
{\bfseries Case~\ruleref{seq-ti}}\quad Then the derivation tree of height $n+1$ ends with
\begin{align*}
\infer[\ruleref{seq-ti}]{
\thePredicates_1, \theInterference_1, \theHyp_1 \semcalctinf\hoareOf{\acpred}{\astmt_1}{\acpredppp}
\quad
\thePredicates_2, \theInterference_2, \theHyp_2 \semcalctinf\hoareOf{\acpredppp}{\astmt_2}{\acpredp}
}{
\set{\acpredppp}\cup \thePredicates_1\cup\thePredicates_2, \theInterference_1\cup\theInterference_2, \theHyp_1\cup\theHyp_2 \semcalctinf\hoareOf{\acpred}{\seqof{\astmt_1}{\astmt_2}}{\acpredp}\ .}
\end{align*}
So $\thePredicates$ from above is $\set{\acpredppp}\cup \thePredicates_1\cup\thePredicates_2$ and similar for the other components.\\
Consider $\theInterference$ with $\theInterference_1\cup\theInterference_2\subseteq\theInterference$ so that $\hypholdsof{\theInterference}{(\theHyp_1\cup\theHyp_2)}$.\\[0.2cm]
Since $\theInterference_1\subseteq\theInterference$ with $\hypholdsof{\theInterference}{\theHyp_1}$ and similar for the second correctness statement, and since the derivation trees for these statements have height at most $n$, the induction hypothesis applies and yields 
\begin{align*}
\thePredicates_1\cap \governeddef, \theInterference &\semcalc\hoareOf{\acpred\cap\governeddef}{\astmt_1}{\acpredppp\cap\governeddef}\\
\thePredicates_2\cap \governeddef, \theInterference &\semcalc\hoareOf{\acpredppp\cap\governeddef}{\astmt_2}{\acpredp\cap\governeddef}\ .
\end{align*}
We use these correctness statements given by the hypothesis as a premise for sequential composition:
\begin{align*}
\infer[\ruleref{seq}]{
\thePredicates_1\cap \governeddef, \theInterference\semcalc\hoareOf{\acpred\cap\governeddef}{\astmt_1}{\acpredppp\cap\governeddef}\\
\thePredicates_2\cap \governeddef, \theInterference \semcalc\hoareOf{\acpredppp\cap\governeddef}{\astmt_2}{\acpredp\cap\governeddef}
}{
\set{\acpredppp\cap\governeddef}\cup (\thePredicates_1\cap\governeddef)\cup(\thePredicates_2\cap\governeddef), \theInterference \semcalc\hoareOf{\acpred\cap\governeddef}{\seqof{\astmt_1}{\astmt_2}}{\acpredp\cap\governeddef}\ .}
\end{align*}
Since $\set{\acpredppp\cap\governeddef}\cup (\thePredicates_1\cap\governeddef)\cup(\thePredicates_2\cap\governeddef)=(\set{\acpredppp}\cup \thePredicates_1\cup\thePredicates_2)\cap\governeddef$, the latter correctness statement is as desired.\\[0.2cm]
{\bfseries Case~\ruleref{choice-ti}}\quad Similar to the previous case.
\end{proof}
\begin{proof}[Proof (of Theorem~\ref{Theorem:Soundness})]
Consider $\thePredicates, \theInterference, \theHyp \semcalcti\hoareOf{\acpred}{\astmt}{\acpredp}$ with $\acpred\in\thePredicates$, $\isInterferenceFreeOf[\theInterference]{\thePredicates}$, and $\hypholdsof{\theInterference}{\theHyp}$. \\
We have to show $\thePredicates\cap \governeddef, \theInterference \semCalc\hoareOf{\acpred\cap\governeddef}{\astmt}{\acpredp\cap \governeddef}$ with  $\acpred\cap\governeddef\in\thePredicates\cap\governeddef$ and $\isInterferenceFreeOf[\theInterference]{(\thePredicates\cap\governeddef)}$. \\\\
For the derivation, Lemma~\ref{Lemma:FrameElimination} yields $\thePredicates, \theInterference, \theHyp \semcalctinf\hoareOf{\acpred}{\astmt}{\acpredp}$.\\
Now Lemma~\ref{Lemma:HistoryTracking} plus $\hypholdsof{\theInterference}{\theHyp}$ shows $\thePredicates, \theInterference \semcalc\hoareOf{\acpred\cap\governeddef}{\astmt}{\acpredp\cap\governeddef}$, as desired.\\\\
If $\acpred\in\thePredicates$, then $\acpred\cap\governeddef\in\thePredicates\cap\governeddef$ by the definition of $\thePredicates\cap\governeddef$.\\\\
For interference freedom, we have $\isInterferenceFreeOf[\theInterference]{\thePredicates}$ by the assumption.\\
Moreover, $\isInterferenceFreeOf[\theInterference]{\governeddef}$.\\
If we intersect two interference-free predicates, we obtain an interference-free predicate.\\ 
So the last point $\isInterferenceFreeOf[\theInterference]{(\thePredicates\cap\governeddef)}$ follows. 
\end{proof}

\begin{proof}[Proof (of \Cref{Lemma:GeneralizedTISound})]
Consider $\ahist\in\histof{\apred}{\astmt}\cap \governeddef$. 
We turn it into a history to which $\hypholdsof{\theInterference}{\proghypof{\apred}{\astmt}{\apredp}{\apredpp}}$ applies and allows us to conclude $\nowof{\apredp}\rightarrow\weakpastof{\apredpp}$.

Since $\ahist\in\histof{\apred}{\astmt}$, $\ahist=\ahist_1.\ahist_2$ for some $\ahist_1\in\sethisteps$ and  $\ahist_2=\astateseq_1.\acom_1.\astateseq_2\ldots \astateseq_{n}.\acom_n.\astateseq_{n+1}$ with $\acom_1\ldots\acom_n\in\astmt$.
We also have $\ahist\in\governeddef$, and so $\ahist_2\in\csemof{\stmtof{\theInterference}}_{\theInterference}(\Sigma)\downarrow$.
Hence, there is $\ahist_2'=\astateseq_1.\acom_1'.\astateseq_2\ldots \astateseq_{n}.\acom_n'.\astateseq_{n+1}\in\csemof{\stmtof{\theInterference}}_{\theInterference}(\Sigma)$ with $\acom_i' = \commandof{\acpred_i, \acom_i}$ and $(\acpred_i, \acom_i)\in\theInterference$. 
Together, we obtain $\acom_1'\ldots \acom_n'\in\enrichof{\astmt}{\theInterference}$.

By the definition of $\governeddef$, the state changes in the $\astateseq_i$ are due to the interferences in $\theInterference$. 
Moreover, also the state changes around the commands respect the semantics of the commands. 
Hence, 
\begin{align*}
\ahist_1.\ahist_2'\in \csemof{\enrichof{\astmt}{\theInterference}}_{\theInterference}(\ahist_1.\firstof{\astateseq_1})\ .
\end{align*}
By the definition of $\histof{\apred}{\astmt}$, we get $\firstof{\astateseq_1}\in\apred$. 
This means $\ahist_1.\ahist_2'\in \csemof{\enrichof{\astmt}{\theInterference}}_{\theInterference}(\nowof{\apred})$.

The assumption $\hypholdsof{\theInterference}{\proghypof{\apred}{\astmt}{\apredp}{\apredpp}}$ and \Cref{Lemma:SoundInclusion} justify the inclusion 
\begin{align*}
\csemof{\enrichof{\astmt}{\theInterference}}_{\theInterference}(\nowof{\apred})\quad \subseteq\quad \nowof{\apredp}\rightarrow\weakpastof{\apredpp}\ .
\end{align*}
Hence, $\ahist_1.\ahist_2'\in \nowof{\apredp}\rightarrow\weakpastof{\apredpp}$. 
Since the now and the weak past predicate only refer to the states in the computation, which coincide for $\ahist_1.\ahist_2'$ and $\ahist_1.\ahist_2$, we can conclude 
$\ahist_1.\ahist_2\in \nowof{\apredp}\rightarrow\weakpastof{\apredpp}$, as desired.
\end{proof}


\begin{proof}[Proof (of \Cref{Lemma:FrameElimination})]
To eliminate \ruleref{frame-ti}, we proceed by Noetherian induction on the height of the derivation tree for $\thePredicates, \theInterference, \theHyp \semcalcti\hoareOf{\acpred}{\astmt}{\acpredp}$. 
The height of the derivation tree is the maximal number of consecutive rule applications leading to the correctness statement.
We give here the base case of \ruleref{temporal-interpolation} followed by \ruleref{frame-ti}. 
Consider intuitionistic predicates $\apred, \apredp$, and a predicate $\apredpp$. 
We have 
\begin{align*}
\infer[\ruleref{temporal-interpolation}]{\apred, \apredp\text{ intuitionistic}}{
\infer[\ruleref{frame-ti}]{\set{\acpred\cap \pastof{\apredpp}}, \set{(\acpred, \cskip)}, \set{\weakhypof{\apred}{\apredp}{\apredpp}} \semcalcti\hoareOf{\acpred\cap\weakpastof{\apred}\cap\nowof{\apredp}}{\cskip}{\acpred\cap \pastof{\apredpp}}}{\set{\acpred\cap \pastof{\apredpp}}\mstar\acpredpp, \set{(\acpred, \cskip)}\mstar\acpredpp, \set{\weakhypof{\apred}{\apredp}{\apredpp}} \semcalcti\hoareOf{(\acpred\cap\weakpastof{\apred}\cap\nowof{\apredp})\mstar\acpredpp}{\cskip}{(\acpred\cap\pastof{\apredpp})\mstar\acpredpp}\ .}}
\end{align*}
For \ruleref{frame-ti}-free derivations, we apply \ruleref{temporal-interpolation} followed by \ruleref{consequence-ti}:
\begin{align*}
\infer[\ruleref{temporal-interpolation}]{\apred, \apredp\text{ intuitionistic}}{
\infer[\ruleref{consequence-ti}]{\set{\acpred\mstar\acpredpp\cap \pastof{\apredpp}}, \set{(\acpred\mstar\acpredpp, \cskip)}, \set{\weakhypof{\apred}{\apredp}{\apredpp}} \semcalcti\hoareOf{\acpred\mstar\acpredpp\cap\weakpastof{\apred}\cap\nowof{\apredp}}{\cskip}{\acpred\mstar\acpredpp\cap \pastof{\apredpp}}}{\set{\acpred\cap \pastof{\apredpp}}\mstar\acpredpp, \set{(\acpred, \cskip)}\mstar\acpredpp, \set{\weakhypof{\apred}{\apredp}{\apredpp}} \semcalcti\hoareOf{(\acpred\cap\weakpastof{\apred}\cap\nowof{\apredp})\mstar\acpredpp}{\cskip}{(\acpred\cap\pastof{\apredpp})\mstar\acpredpp}\ .}}
\end{align*}
The equality $\acpred\mstar\acpredpp\cap \pastof{\apredpp}=\set{\acpred\cap \pastof{\apredpp}}\mstar\acpredpp$ is \Cref{Lemma:PreciseIntuitionistic}. 
It is also used in the postcondition. 
For the interferences, we have $\set{(\acpred\mstar\acpredpp, \cskip)}=\set{(\acpred, \cskip)}\mstar\acpredpp$. 
We indeed strengthen the precondition, as
\begin{align*}
(\acpred\cap\weakpastof{\apred}\cap\nowof{\apredp})\mstar\acpredpp \subseteq \acpred\mstar\acpredpp\cap\weakpastof{\apred}\mstar\acpredpp\cap\nowof{\apredp}\mstar\acpredpp\subseteq\acpred\mstar\acpredpp\cap\weakpastof{\apred}\cap\nowof{\apredp}.
\end{align*}
The second inclusion uses that $\weakpastof{\apred}$ and $\nowof{\apredp}$ are intuitionistic by \Cref{Lemma:PreciseIntuitionistic}.
\end{proof}

\begin{proof}[Proof (of \Cref{Lemma:HistoryTracking})]
We again proceed by Noetherian induction on the height of \ruleref{frame-ti}-free derivations and consider the difficult base case of \ruleref{temporal-interpolation}:
\begin{align*}
\infer{\apred, \apredp\text{ intuitionistic}}{
\set{\acpred\cap\pastof{\apredpp}}, \set{(\acpred, \cskip)}, \set{\weakhypof{\apred}{\apredp}{\apredpp}} \semcalctinf\hoareOf{\acpred\cap\weakpastof{\apred}\cap\nowof{\apredp}}{\cskip}{\acpred\cap\pastof{\apredpp}}\ .}
\end{align*}
Let $\theInterference$ be a set of interferences with $\set{(\acpred, \cskip)}\subseteq \theInterference$ so that $\hypholdsof{\theInterference}{\weakhypof{\apred}{\apredp}{\apredpp}}$. 
To obtain an ordinary derivation, we first apply \ruleref{com} and get
\begin{align*}
\infer{\csemof{\cskip}(\acpred\cap\weakpastof{\apredpp}\cap \governeddef)\subseteq\acpred\cap\pastof{\apredpp}\cap \governeddef}{
\set{\acpred\cap\pastof{\apredpp}\cap \governeddef}, \set{(\acpred\cap\weakpastof{\apredpp}\cap \governeddef, \cskip)} \semcalc\hoareOf{\acpred\cap\weakpastof{\apredpp}\cap\governeddef}{\cskip}{\acpred\cap\pastof{\apredpp}\cap \governeddef}\ .}
\end{align*}
To see the inclusion in the premise, we have $\csemof{\cskip}(\acpred)\subseteq \acpred$, 
because $\cskip$ is the identity and $\acpred$ is frameable. 
Then $\csemof{\cskip}(\weakpastof{\apredpp})\subseteq \pastof{\apredpp}$, since $\cskip$ adds an extra step. 
For the governed computations, consider $\acomp.\astate.\astate\in \csemof{\cskip}(\acpred\cap\governeddef)$. 
Then $\acomp.\astate\in\governeddef$, meaning the state changes in $\acomp.\astate$ are governed by the interferences. 
Moreover, we have $\acomp.\astate\in\acpred$. 
Hence, the state change from $\astate$ to $\astate$ is covered by the interference $\set{(\acpred, \cskip)}$. 
Since $\set{(\acpred, \cskip)}\subseteq\theInterference$, we get $\acomp.\astate.\astate\in\governeddef$ as required. 

We now apply \ruleref{consequence} and get
\begin{align*}
\infer{
\set{\acpred\cap\pastof{\apredpp}\cap \governeddef}, \set{(\acpred\cap\weakpastof{\apredpp}\cap \governeddef, \cskip)}\semcalc\hoareOf{\acpred\cap\weakpastof{\apredpp}\cap\governeddef}{\cskip}{\acpred\cap\pastof{\apredpp}\cap \governeddef}}{
\set{\acpred\cap\pastof{\apredpp}}\cap \governeddef, \theInterference \semcalc\hoareOf{\acpred\cap\weakpastof{\apred}\cap\nowof{\apredp}\cap\governeddef}{\cskip}{\acpred\cap\pastof{\apredpp}\cap \governeddef}\ .}
\end{align*}
We have generalized the set of interferences and strengthened the precondition. 
As for the former, we have $\set{(\acpred, \cskip)}\subseteq \theInterference$ by the assumption. 
This implies $\set{(\acpred\cap\weakpastof{\apredpp}\cap \governeddef, \cskip)}\subseteq\theInterference$. 
As for the latter, note that we have $\hypholdsof{\theInterference}{\weakhypof{\apred}{\apredp}{\apredpp}}$ by the assumption. 
Hence, \Cref{Lemma:ATISound} applies and yields $\acpred\cap\weakpastof{\apred}\cap\nowof{\apredp}\cap\governeddef\subseteq \acpred\cap \weakpastof{\apredpp}\cap\governeddef$. 
\end{proof}


\newcommand{\seqspec}{\mathcal{S}}
\newcommand{\impl}{\mathcal{I}}
\newcommand{\datadomain}{\mathsf{D}}
\newcommand{\opcall}{\mathtt{call}}
\newcommand{\opret}{\mathtt{ret}}
\newcommand{\alloperations}{\mathsf{OP}}
\newcommand{\anoperation}{\mathsf{op}}
\newcommand{\anactpar}{\mathit{a}}
\newcommand{\aretval}{\mathit{v}}
\newcommand{\defcall}{\opcall\ \anoperation(\anactpar)}
\newcommand{\defret}{\aretval = \opret\ \anoperation(\anactpar)}
\newcommand{\dropof}[1]{#1\!\!\downarrow}
\newcommand{\alphabet}{\Gamma}
\newcommand{\wppref}{\mathit{w}'_{\mathit{pre}}}
\newcommand{\wpop}{\mathit{w}'_{\mathit{op}}}

\section{Meta-theory for proving linearizability}
Linearizability assumes to be given a \emph{sequential specification} of an object.
A sequential specification is a language over operation calls and returns in which (i) every operation call is decorated by the actual parameters, (ii) the return immediately follows the call, and (iii) the return is decorated by the return values for the actual parameters. 
Let $\alloperations$ be the set of all operations for accessing the object and for simplicity assume that every operation $\anoperation$ accepts a single parameter $\anactpar$ and returns a single value $\aretval$ from a domain $\datadomain$. 
With this, a sequential specification is a subset 
\begin{align*}
\seqspec\quad \subseteq\quad \setcond{\defcall.\defret}{\anoperation\in\alloperations, \anactpar, \aretval\in\datadomain}^*. 
\end{align*}

Search structures store sets of keys $\abscontent\subseteq \nat$ and their operations modify these sets. 
We can therefore give the sequential specification as a set of predicates $\acssup_{\anoperation}(\abscontent,\abscontentp,\anactpar,\aretval)$, one per operation, that specify this modification relative to a given actual parameter and return value.  
For example, an insertion of key $k$ with return value $\aretval$ would be captured by
\begin{align*}
\acssup_{\texttt{insert}}(\abscontent,\abscontentp,k,\aretval)\quad\defeq\quad \abscontentp = \abscontent \cup\set{k}\ \wedge\ \aretval \Leftrightarrow k\notin\abscontent. 
\end{align*}

With the predicates at hand, we define an automaton whose trace language is the sequential specification of the search structure. 
The automaton is $(\powerset{\nat}, \bigcup_{\anoperation\in\alloperations}E_{\anoperation})$, the states are all possible search structure contents, and we have a set of labeled edges $E_{\anoperation}$ per operation.
This set is defined to contain all transitions
\begin{align*}
\abscontent\xrightarrow{\opcall\ \anoperation(k).\aretval = \opret\ \anoperation(k)} \abscontentp,\quad\text{where }\acssup_{\anoperation}(\abscontent,\abscontentp, k, \aretval)\text{ holds}. 
\end{align*}
We use $\seqspec(\abscontent)$ for the trace language of this automaton when starting in $\abscontent$, and write $\seqspec$ for $\seqspec(\emptyset)$. 

A concurrent implementation of the search structure is a program of the form
\begin{align*}
\astmt\quad = \quad (\sum_{\anoperation\in\alloperations}\astmt_{\anoperation})^* \ .
\end{align*}
Every operation is represented by a piece of code $\astmt_{\anoperation}$, 
and a thread executing the implementation may exercise the operations in arbitrary order. 
The semantics is as defined in Appendix~\ref{Appendix:Preliminaries}. 
The implementation is executed by an arbitrary number of threads, each represented by an id~$i$, which modify a global state from $\setshared$ and a local state from $\setlocal$. 
For linearizability, we need a small addition. 
We assume the execution of the first command in $\astmt_{\anoperation}$, say by thread~$i$, makes visible the letter $\opcall_i\ \anoperation(k)$, the execution of the operation's return command yields $\aretval = \opret_i\ \anoperation(k)$ with $\aretval$ the return value, and the execution of commands inside the operation makes visible the thread id $i$.  
The transition system from Appendix~\ref{Appendix:Preliminaries} then yields a language over the alphabet $\alphabet$ of thread ids and calls and returns decorated with thread ids. 
We write $\impl(\initset{\acpred}{\astmt})$ for the trace language starting in a configuration from $\initset{\acpred}{\astmt}$. 
We simply write $\impl$ if the initial global and local heaps are empty. 
We focus on traces in which all operations execute to completion.

The words in $\impl$ interleave the operations executed by different threads. 
Linearizability admits the following rewriting of such an interleaving:
\begin{align*}
w.a.b.u\quad\rightsquigarrow\quad w.b.a.u,
\end{align*}
provided $a$ and $b$ stem from different threads and it is not the case that $a$ is a return and $b$ a call. 
This means we may arbitrarily order overlapping operations, but may not change the order of consecutive operations (the real-time order).
To make the link to sequential specifications, we define the partial function $\dropof{w}$ that drops thread ids as letters from calls and returns. 
The function is only defined if word $w$ is sequential, meaning $w$ decomposes into infixes of commands by thread $i$ leading from an invocation to the corresponding return.
\begin{definition}{\cite{DBLP:journals/toplas/HerlihyW90}}
A concurrent implementation $\astmt$ is \emph{linearizable} wrt. sequential specification~$\seqspec$, if for every $u\in\impl$ there is $v$ with $u\rightsquigarrow^* v$ so that $\dropof{v}\ \in \seqspec$. 
\end{definition} 

Towards a proof principle for linearizability, we now tie words over $\alphabet$ to runs of the automaton underlying the sequential specification. 
We consider words $\abscontent_1.a_1.\abscontent_2.a_2\ldots$ that interleave search structure contents and thread ids resp. decorated calls and returns.
We call an infix  $\abscontent.i.\abscontentp$ of such a word a \emph{command of thread $i$}. 
We call an infix $\opcall_i\ \anoperation(k)\ .\ u\ .\ \aretval = \opret_i\ \anoperation(k)$ an \emph{operation of thread $i$}, if $u$ does not contain any calls or returns by $i$. 
We call such a word a \emph{computation}, if the projection to every thread yields a sequence of operations of that thread.
Note that a computation does not have to stem from~$\astmt$ but the term applies more broadly.  
Our proof principle is this.
\begin{definition}\label{Definition:AdheresTo}
Operation $\opcall_i\ \anoperation(k)\ .\ u\ .\ \aretval = \opret_i\ \anoperation(k)$ \emph{adheres to the sequential specification}, if (1) it contains a command $\abscontent.i.\abscontentp$ of thread $i$, the \emph{linearization point}, with $\acssup_{\anoperation}(\abscontent,\abscontentp,k,\aretval)$, and (2) for all other commands $\abscontent.i.\abscontentp$ of $i$ we have $\abscontent = \abscontentp$. 
We say that a computation $w$ \emph{adheres to the sequential specification}, if this holds for every operation in $w$.  
We use $w_{\Gamma}$ for the projection to~$\Gamma$. 
\end{definition}
If the proof principle holds, the computation actually is a run of the automaton underlying the sequential specification. 
To see this, note that the commands in (2) do not alter the contents of the data structure, and so the sequential specification can stay in the same state. 
A linearization point may result in a contents modification, in which case the operation's predicate in the sequential specification is guaranteed to hold. 
Since the predicate defines the edges of the automaton underlying the sequential specification, 
the contents modification can be tracked in the automaton. 
\begin{theorem}
If computation $w$ adheres to the sequential specification, then $w_{\Gamma}$ is linearizable.  
\end{theorem}
\begin{proof}
Let $w$ be a computation that adheres to the sequential specification. 
We show that there is a computation $w'$ so that (1) $w'$ adheres to the sequential specification, (2) $w_{\Gamma}\rightsquigarrow^* w'_{\Gamma}$, (3)~$\dropof{w'_{\Gamma}}$ is defined, and (4) the last contents in $w$ and $w'$ is the same.
This is enough to establish linearizability of $w_{\Gamma}$.
Since $w'$ adheres to the sequential specification by (1), we have that $w'_{\Gamma}\in \seqspec$. 
This follows from the paragraph before the theorem, arguing that the automaton underlying the sequential specification has $w'$ as a run. 
We moreover have $w_{\Gamma}\rightsquigarrow^* w'_{\Gamma}$ by~(2) and $\dropof{w'_{\Gamma}}$ defined by (3).

We proceed by induction on the number of linearization points in the computation.  
In the base case of a single linearization point, there is nothing to do.
Assume the claim holds for computations with $n$ linearization points. 
Let computation $w$ have $n+1$ linearization points.
Then $w$ has the shape $w_1.a.w_2.\abscontent.i.\abscontentp.w_3.b.w_4$ so that $\abscontent.i.\abscontentp$ is the last linearization point and $a$ and $b$ are the call and return of the corresponding operation. 
We know that~$w_4$ does not contain calls, otherwise $\abscontent.i.\abscontentp$ would not be the last linearization point.
Moreover, $w_4$ will not contain $i$-commands. 
This allows us to move all commands from $w_4$ before $b$, resulting in $w_1.a.w_2.\abscontent.i.\abscontentp.w_3.w_4.b$. 
Relation $\rightsquigarrow$ allows us to move the commands of other threads out of $w_2$ and $w_3$. 
We do so from left to right in order to preserve the fact that we have a computation and the order of linearization points potentially present in $w_2$.
The result is $w'\defeq \wppref.\wpop$ with
\begin{align*}
\wppref\ \defeq\ w_1.w_2'.w_3'.w_4'\qquad\text{and}\qquad \wpop\ \defeq\ a.w_2''.\abscontent.i.\abscontentp.w_3''.b\ .
\end{align*} 
Here, $w_2'$ and $w_3'$ contain the commands from $w_2$ resp. $w_3$ that belong to threads different from $i$, and $w_2''$ and $w_3''$ contain the $i$-commands.
In $w_2'$ we maintain the memory contents we had in $w_2$. 
In $w_3'$, $w_4'$, and $w_2''$, we change the memory contents to $\abscontent$. 
Note that $w'$ is a computation and $w_{\Gamma}\rightsquigarrow^* w'_{\Gamma}$. 

We argue that $w'$ adheres to the specification, by showing that $\abscontent$ is the last contents in~$w_2'$. 
Let $\abscontent''.j.\abscontent'''$ be the last linearization point in $w_2$ before rewriting.
Since $w$ adheres to the specification, the subsequent commands will not modify the contents and $\abscontent'''=\abscontent$ has to hold. 
Since we move the commands out of $w_2$ from left to right, $\abscontent''.j.\abscontent$ will also be the last linearization point in $w_2'$. 

Since $w'$ is a computation that adheres to the specification and $\wpop$ is an operation, also $\wppref$ is a computation that adheres to the specification. 
Since it only has $n$ linearization points, the induction hypothesis applies to $\wppref$ and yields a computation $u$ with properties (1) to (4). 
We append the last operation and obtain $u' \defeq u.\wpop$. 
Then $u'$ is again a computation. 

We show that $u'$ has properties (1) to (4). 
To see (1), that $u'$ adheres to the sequential specification, note that $u$ adheres to the sequential specification by (1) from the induction hypothesis. 
Moreover, the last contents in $\wppref$ is $\abscontent$, and by (4) from the hypothesis this is also the last contents in~$u$. 
Since~$\abscontent$ is also the first contents in $\wpop$, and since $\wpop$ adheres to the specification, we have that $u'$ adheres to the specification. 
To see (3), note that $\dropof{u}$ is defined by (3) from the hypothesis, and $\wpop$ is a sequential operation, hence $\dropof{u'}$ is defined.
For (4), the last contents in $w$ and $u'$ is $\abscontentp$. 

It remains to show (2), namely $w_{\Gamma}\rightsquigarrow^* u'_{\Gamma}$.
We showed above $w_{\Gamma}\rightsquigarrow^* w'_{\Gamma}$ with $w'=\wppref.\wpop$.
By (2) from the hypothesis, we have ${\wppref}_{\Gamma} \rightsquigarrow^* u_{\Gamma}$. 
The rewriting relation is stable under contexts.
We can thus also execute this rewriting with $\wpop$ appended, yielding $w'_{\Gamma}={\wppref}_{\Gamma}.{\wpop}_{\Gamma} \rightsquigarrow^* u_{\Gamma}.{\wpop}_{\Gamma}=u'_{\Gamma}$.
\end{proof}

To apply the proof principle, we associate with every global state $\ashared$ reachable when executing the concurrent implementation of the search structure its contents $\abscontent(\ashared)\subseteq\nat$. 
It is defined as the unique $\abscontent\subseteq\nat$ for which $\ashared\models \acss(\abscontent)$. 
Recall that $\acss(\abscontent)\defeq \exists\somenodes.~\inv(\abscontent,\somenodes)$, the contents predicate is derived from the invariant, \cref{sec:lotree:locality}. 
Since the invariant is guaranteed to be maintained, $\abscontent(\ashared)$ is guaranteed to be defined. 
With this definition, we can understand the words $w\in \impl$ as interleavings $w=\abscontent(\ashared_1).a_1.\abscontent(\ashared_2).a_2\ldots$
It is readily checked that these interleavings form computations in the above sense.  
We say that $\astmt$ adheres to the sequential specification, if this holds for all $w \in\impl$ when seen as computations. 
\begin{corollary}\label{Corollary:ProofPrinciple}
If $\astmt$ adheres to the sequential specification $\seqspec$, then $\astmt$ is linearizable wrt. $\seqspec$. 
\end{corollary}
The proof rules in Figure~\ref{fig:linearizability-rules} implement the check that the execution of every operation adheres to the sequential specification, and thus Corollary~\ref{Corollary:ProofPrinciple} applies. 
To be precise, Rule~\ruleref{com-lin-void} checks that a command does not alter the search structure content, as required by Condition~(2) in Definition~\ref{Definition:AdheresTo}. 
Rule~\ruleref{com-lin-impure} explicitly checks that contents modification, actual parameter, and return value together respect the predicate specifying the operation.
This is one requirement of Condition~(1), but Definition~\ref{Definition:AdheresTo} requires more: there should be at most one linearization point. 
Uniqueness is guaranteed by the fact that the rule expects an $\anobl$ token in the precondition, produces a $\aful{v}$ token in the postcondition, and a $\aful{v}$ token cannot be transformed into an $\anobl$ token nor can an $\anobl$ token be produced by commands.  
We have argued here about the modification of the search structure contents on the level of rules.
Definition~\ref{Definition:AdheresTo} refers to computations, instead.  
The close correspondence between rules and program semantics is made precise in the program logic's soundness proof (proof of \Cref{Theorem:SoundnessComput}), which can be found in~\cite{DBLP:journals/corr/abs-2207-02355}.



\section{Impure Future-Dependent Linearization Points}
\label{sec:RDCSS}

\newcommand{\lm}{\mathit{lm}}
\newcommand{\rstate}{\mathsf{Rstate}}
\newcommand{\authrstate}{\rstate^\bullet}
\newcommand{\pactive}{\mathsf{Active}}
\newcommand{\inactive}{\mathsf{Inactive}}
\newcommand{\clock}{\mathsf{Clock}}
\newcommand{\proto}{\mathsf{Proto}}
\newcommand{\oblc}[1]{\OBLname_{#1}}
\newcommand{\fracpto}[3]{#1 \stackrel{\scriptsize{#2}}{\mapsto} #3}
\newcommand{\fracPerm}[2]{\text{\nicefrac{#1}{\!#2}}}
\newcommand{\ful}[1]{\FULname_{#1}}
\newcommand{\pful}[1]{\circ \FULname_{#1}}
\newcommand{\fulc}[2]{#1 \mapsto \ful{#2}}
\newcommand{\pfulc}[2]{#1 \mapsto \pful{#2}}

The rule given in \cref{sec:linearizability} for proving linearizability with retrospective reasoning is restricted to pure future-dependent linearization points. However, the approach can be generalized to handle impure future-dependent linearization points, i.e., those that modify the abstract state of the data structure.

In the presence of impure future-dependent linearization points, the
abstract state of the data structure at any given point in time of the
concurrent execution may depend on future thread interferences.
Rather than tracking a single abstract state in the proof, the idea is
to track a set of abstract states, one for each possible future.
This set of abstract states can be defined purely in terms of the
computation history. This idea is inspired
by the original proof of the Herlihy/Wing queue~\cite{DBLP:journals/toplas/HerlihyW90}. A similar idea has also been explored in~\cite{DBLP:conf/esop/KhyzhaDGP17}.

Each of the tracked abstract states carries its own obligation/fulfillment
token for each active operation.
When a thread changes the physical representation of the data
structure, the change may affect the abstract state for some but not
all possible futures.
For the affected abstract state, the proof obligation is to show that
the change is consistent with the sequential specification and that
the associated obligation token can be traded in for the fulfillment
token.
A modification of the data structure may also eliminate some of the
possible abstract states but it must not eliminate all.

At the return point of an operation the proof obligation is to show that the thread has indeed linearized for all possible abstract states at that point. This step can then make use of retrospective reasoning using temporal interpolation, similar to the rule \ruleref{com-lin-pure}.

This more general construction necessitates a helping protocol that governs the transfer of linearizability obligations between threads to handle cases where an impure linearization point of an operation lies in another thread. These proofs are therefore more difficult to automate than those involving only pure future-dependent linearization points.

We consider the automation of proofs involving impure future-dependent linearization points future work. However, to provide evidence that our logical is equipped to express such proofs, we here discuss a second case study: verifying the RDCSS data structure~\cite{DBLP:conf/wdag/HarrisFP02}. This case study involves impure future-dependent linearization points and helping. However, the data structure's abstract state is always uniquely determined by the computation history. So there is still no need to track sets of abstract states in the proof of this data structure.

\subsection{High-level Overview of RDCSS}

RDCSS, which stands for \emph{restricted double compare single swap}, is a data structure that implements a form of multi-word compare and swap operation. The data structure governs a memory location $r$ and its logical value $n$ by an abstract predicate $\rstate(r,n)$. It provides two operations, \lstinline+rdcss+ and \lstinline+get+, whose sequential specification is shown in \Cref{fig:rdcss-spec}. The operation \lstinline+get($r$,$n$)+ simply returns the current logical value $n$ of $r$. The operation \lstinline+rdcss($r$,$\ell$,$n_1$,$m_1$,$n_2$)+ takes a reference to a second memory location $\ell$ and only if the current value of $\ell$ is $m_1$ and the current value of $r$ is $n_1$, does it update $r$ to the new value $n_2$. Otherwise, it leaves $r$ unchanged. In all cases, the operation returns the old value of $r$.

\begin{figure}
  \begin{minipage}[t]{\linewidth}
    \begin{lstlisting}[gobble=4,language=SPL,mathescape=true,escapechar=@]
      $\annot{\rstate(r,n)}$ get($r$) $\annot{v.\; v = n \land \rstate(r,n)}$
      $\annot{\fracpto{\ell}{q}{m} \MSTAR \rstate(r,n)}$ rdcss($r$,$\ell$,$n_1$,$m_1$,$n_2$) $\annot{v.\; v = n \land \fracpto{\ell}{q}{m} \MSTAR \rstate(r,\ite{m=m_1 \land n=n_1}{n_2}{n})}$
    \end{lstlisting}
  \end{minipage}
  \caption{RDCSS data structure specification\label{fig:rdcss-spec}}
\end{figure}

\begin{figure}
  \begin{minipage}[t]{\linewidth}
    \begin{lstlisting}[gobble=4,language=SPL,mathescape=true,escapechar=@]
    datatype Descr = D($\ell$: Ref[Val], $n_1$: Val, $m_1$: Val, $n_2$: Val)
    datatype State = A($n$: Val) | I($\ell$: Ref[Descr])      

    method complete($r$: Ref[State], $d$: Ref[Descr]) {
      val D($\ell$, $n_1$, $m_1$, $n_2$) = !$d$ 
      val $m$ = !$\ell$ @\label{line:rdcss-complete-read-ell}@
      val $n'$ = $m$ == $m_1$ ? $n_2$ : $n_1$
      CmpX($r$, A($d$), I($n'$)) @\label{line:rdcss-complete-try-set-inactive}@
    }
    
    method get($r$: Ref[State]): Val {
      match !$r$ with {
        case I($v$) => return $v$
        case A($d$) =>
          complete($r$, $d$)
          return get($r$)
      }
    }

    method rdcss($r$: Ref[State], $\ell$: Ref[Val], $n_1$: Val, $m_1$: Val, $n_2$: Val): Val {
      val $d$ = new Ref(D($\ell$, $n_1$, $m_1$, $n_2$)) @\label{line:rdcss-alloc-descriptor}@
      val $s$ = CmpX($r$, I($n_1$), A($d$)) @\label{line:rdcss-try-make-active}@
      match $s$ with {
        case I($n$) =>
          if ($n$ == $n_1$) {
            complete($r$, $d$) @\label{line:rdcss-active-complete}@
            return $n_1$
          } else return $n$
        case A($d'$) =>
          complete($r$, $d'$)
          return rdcss($r$, $\ell$, $n_1$, $m_1$, $n_2$)
      }
    }    
  \end{lstlisting}
\end{minipage}
\caption{RDCSS data structure implementation\label{fig:rdcss}}
\end{figure}

An implementation of the data structure is shown in \Cref{fig:rdcss}.
The key challenge for the implementation is that the \lstinline+rdcss+ operation must read $r$ and $\ell$ in a single logically atomic step, even though two physical steps are required to read both locations.
So other threads may interfere and change the value of either location between the two reads.
In particular, the location $\ell$ is extraneous to the data structure and, hence, the client may concurrently update its value while an \lstinline+rdcss+ operation is in progress.
The data structure solves this challenge by maintaining two state modes. 
If the structure is in \emph{inactive} mode, indicated by storing the value $\m{I}(n)$ in $r$, then no \lstinline+rdcss+ operation is in progress and the logical value is $n$. In particular, a \lstinline+get+ operation can simply read out $n$ from the inactive state and return. If an \lstinline+rdcss($r$,$\ell$,$n_1$,$m_1$,$n_2$)+ operation starts, it first checks whether the structure is in inactive mode and whether its value is $n_1$. If yes, then it changes the state into \emph{active} mode by replacing $\m{I}(n_1)$ in $r$ with $\m{A}(d)$ where $d$ is a fresh location allocated on Line~\ref{line:rdcss-alloc-descriptor}. The location $d$ stores a \emph{descriptor} value $\m{D}(\ell,n_1,m_1,n_2)$ that remembers the actual arguments of this \lstinline+rdcss+ operation. The check and update are performed using a single atomic compare and exchange operation (\lstinline+CmpX+) on Line~\ref{line:rdcss-try-make-active}. The \lstinline+CmpX+ returns the old value $s$ of $r$ before the attempted update. If the update succeeded, the operation is completed by calling \lstinline+complete($r$,$d$)+ on Line~\ref{line:rdcss-active-complete}. The \lstinline+complete+ method then reads the value $m$ of $\ell$ (Line~\ref{line:rdcss-complete-read-ell}), and sets the state back to inactive, $\m{I}(n')$, for the new or old value $n' = (\ite{m = m_1}{n_2}{n_1})$ (Line~\ref{line:rdcss-complete-try-set-inactive}).

The correctness of the implementation hinges on the fact that the active state value $\m{A}(d)$ acts like a lock that gives the active \lstinline+rdcss+ operation exclusive access to the abstract state $\rstate(r,n_1)$. Excluding other \lstinline+rdcss+ operations from accessing the abstract state guarantees that at Line~\ref{line:rdcss-complete-read-ell}, $r$ still has the old logical value $n_1$ that it had on Line~\ref{line:rdcss-try-make-active}. Line~\ref{line:rdcss-complete-read-ell} must be the linearization point because it is the only point where one can guarantee that the logical value of $r$ is $n_1$ and, at the same time, the value of $\ell$ is $m$. Concurrent \lstinline+get+ operations are then still prevented from reading the old value $n_1$ between the linearization point and the point when the physical state of the data structure is updated to store the new value $n'$ on Line~\ref{line:rdcss-complete-try-set-inactive}.

A complication in the algorithm is that concurrent operations are not simply blocked while an \lstinline+rdcss+ operation is active. Instead, the implementation provides a fast path: a concurrent operation encountering an active state $\m{A}(d)$ will try to help complete the active \lstinline+rdcss+ operation using the information provided in the descriptor $d$. Consequently, there can be an unbounded number of threads that concurrently read the value $\ell$ on Line~\ref{line:rdcss-complete-read-ell} and then compete for setting $r$ back to the inactive state on 
Line~\ref{line:rdcss-complete-try-set-inactive}. Thus, only the thread who will ``win this race'' and execute the \lstinline+CmpX+ first should linearize the active \lstinline+rdcss+ at Line~\ref{line:rdcss-complete-read-ell}. This makes the linearization point of \lstinline+rdcss+ future-dependent.

\subsection{Linearizability Proof}
\label{sec:rdcss-lin-proof}

\citet{DBLP:journals/pacmpl/JungLPRTDJ20} provided a fully-mechanized proof of RDCSS, correcting a technical issue in an earlier pencil-and-paper proof by \citet{DBLP:phd/ethos/Vafeiadis08}. Both proofs share the same basic idea: one introduces a prophecy variable $p$ for each active \lstinline+rdcss+ operation. The prophecy $p$ predicts the sequence in which the helping threads will execute the \lstinline+CmpX+ on Line~\ref{line:rdcss-complete-try-set-inactive}. By case analysis on the value of $p$ at Line~\ref{line:rdcss-complete-read-ell}, a helping thread can then determine whether it will be the first thread to execute the \lstinline+CmpX+ and should therefore linearize the active \lstinline+rdcss+.

We here provide an alternative proof that uses temporal interpolation instead of prophecy reasoning. However, we note that our proof draws on ideas from~\cite{DBLP:journals/pacmpl/JungLPRTDJ20} to encode the ownership transfer of the linearization obligation and receipt resources between the active and helping threads via the shared data structure invariant.

The need for prophecy variables arises because the linearizability reasoning outlined in \Cref{sec:linearizability} demands that impure operations are committed at the actual linearization point, i.e., Line~\ref{line:rdcss-complete-read-ell} for \lstinline+rdcss+. If we take a closer look at the specification of the operation, we observe that it consists of two parts. The first part is pure and states that $m$ and $n$ are the values of $\ell$ and $r$ at the linearization point, which are then related to $m_1$ and $n_1$. The second part is impure in the case where the logical value of $r$ is updated to the value $n_2$. Without prophecies, the pure part can still be established at the linearization point. However, the impure part can only be established at the point when the winning thread updates the physical state on Line~\ref{line:rdcss-complete-try-set-inactive}. Establishing the two parts at different points in time is permissible if we can show that no other operation can have been linearized between the two points. In a sense, we can think of \lstinline+rdcss+ as having a \emph{linearization interval} rather than a linearization point. All concurrent operations on the data structure logically perceive this interval as a single point, which we identify with the beginning of the interval at Line~\ref{line:rdcss-complete-read-ell}.

To capture this argument formally, we extend the program logic from \Cref{sec:linearizability} for deriving linearizability judgments of the form $\thePredicates, \theInterference, \theHyp \semcalclin \hoareOf{\acpred}{\acom}{\acpredp}$.

We need to augment the program state with auxiliary ghost state for the relevant bookkeeping. First, we introduce a resource $\clock(r,c)$ for $c \in \nat$ that counts the number of \lstinline+get+ and \lstinline+rdcss+ operations that have already linearized. We will use this resource to express that no operations have linearized over some period of time. The underlying separation algebra is that of partial maps from references $r$ to clock values $c$ with disjoint union as composition. The resource is initialized to $\clock(r,0)$ when the instance $r$ is created and the clock is incremented each time a linearizability obligation resource $\oblc{y}$ is fulfilled.

Next, we change the separation algebra of obligation and receipt resources $\Sigma_{\mathsf{lin}}$ to allow a thread to linearize other threads. In particular, we track the two types of resources in different components of the ghost state and endow each with their own separation algebra. First, we introduce a separation algebra $\Sigma_\OBLname$ of multisets of $\oblc{y}$ values with separating conjunction defined as multiset union. The intuition for the multiset structure is that many operations with the same parameter values may be executing concurrently. So we need to track exactly how many such obligations are available at any time. In assertions, we will write $\oblc{y}$ to represent the singleton multiset containing $\oblc{y}$. 

For the receipt resources we give a two-layered construction. First, we introduce a separation algebra $\Sigma_\FULname$ of values of the form $\ful{y,v}$ and $\pful{y,v}$ where each $\pful{y,v}$ value is the unit of the value $\ful{y,v}$ and separating conjunction is undefined in all other cases. The intuition is that once we have obtained a fulfillment resource $\ful{y,v}$, we can snapshot it as a $\pful{y,v}$ to keep a persistent record of its existence even after $\ful{y,v}$ has been consumed by the postcondition of its associated operation. The second step is to lift this separation algebra to partial maps $\nat \pto \Sigma_\FULname$ in the expected way. That is, partial maps $h$ and $h'$ only compose if for every $c \in \nat$ either $h(c)$ or $h'(c)$ is undefined or $h(c)$ and $h'(c)$ compose. In assertions, we write $\fulc{c}{y,v}$ for the singleton map $\{\fulc{c}{y,v}\}$ and similarly for $\pfulc{c}{y,v}$.

We add a $\Sigma_\OBLname$ and a $\Sigma_\FULname$ component to both the global and local state.

Finally, to encode the helping mechanism, our data structure invariant will be a computation predicate rather than a state predicate. However, recall that in our linearizability proof rules, the predicate describing the abstract state of the data structure occurs below past operators in some of the rules. It must therefore be a state predicate. To circumnavigate this issue, we introduce an auxiliary ghost resource that tracks abstract predicates $\rstate(r,n)$ for all the existing RDCSS instances. The underlying separation algebra is that of partial maps from references $r$ to values $n$ with disjoint union as composition. The abstract predicate $\rstate(r,n)$ in assertions thus represents the singleton map $\{r \mapsto n\}$.  Our actual data structure invariant will then be of the form $\authrstate(r,n) = \rstate(r,n) \MSTAR \inv(r,n)$ where $\inv(r,n)$ is a computation predicate that ties $n$ to the physical state of $r$.

The derivation rules for the judgments $\thePredicates, \theInterference, \theHyp \semcalclin \hoareOf{\acpred}{\acom}{\acpredp}$ are appropriately updated to work with the new ghost state. For example, the rule for a pure linearization point (instantiated for \lstinline+get+) now looks like this:
\begin{mathpar}
  \inferhreftop{com-lin-pure}{com-lin-pure-rdcss}{
    \acpred \subseteq \pastOF{v=n\land\rstate(r,n)}\\
    \apred = \oblc{r}\mstar\clock(r,c)\\
    \apredp = \fulc{c}{r,v}\mstar\clock(r,c+1)
  }{
    \acpred\mstar\apredp, \{(\acpred \mstar \apred, \cskip + (\apred \leadsto \apredp))\}, \emptyset
    \semcalclin
    \hoareOf{\acpred\mstar\apred}{\cskip}{\acpred\mstar\apredp}
  }
\end{mathpar}
More interesting is the rule we will use for handling the linearization interval of the \lstinline+rdcss+ operation:
\begin{mathpar}
  \inferH{com-lin-mixed}{
    \thePredicates, \theInterference, \theHyp
    \semcalcti
    \hoareOf{\acpred}{\acom}{\acpredp}\\\\
    \acpred \subseteq \pastOF{\fracpto{\ell}{q}{m} \MSTAR \rstate(r,n) \MSTAR \clock(r,c)} \mstar \inv(r,n)\\\\
    \acpredp \subseteq \inv(r,n') \land n' = \ite{m=m_1 \land n = n_1}{n_2}{n}\\\\
    \apred = \rstate(r,n) \mstar \oblc{r,\ell,n_1,m_1,n_2}\mstar\clock(r,c)\\\\
    \apredp = \rstate(r,n') \mstar \fulc{c}{r,\ell,n_1,m_1,n_2,n}\mstar\clock(r,c+1)
  }{
    \thePredicates \mstar \apredp, \theInterference + (\apred \leadsto \apredp), \theHyp
    \semcalclin
    \hoareOf{\acpred\mstar\apred}{\acom}{\acpredp\mstar\apredp}
  }
\end{mathpar}
The premise of the rule states that we can show that $\acom$ changes the physical state of the data structure such that its logical value is changed from $n$ to $n'$ while preserving the invariant. Moreover, the new value $n'$ satisfies the postcondition of \lstinline+rdcss+. This captures the impure part of the specification. The additional precondition $\pastOF{\fracpto{\ell}{q}{m} \MSTAR \rstate(r,n) \MSTAR \clock(c)}$ then ensures that there was some past state at logical time $c$ when the value of $\ell$ was $m$ and the logical value of $r$ was $n$. This captures the pure part of the specification. Because the abstract state transition also happens at logical time $c$, the specification is logically satisfied at a single point in time.

\paragraph{Data structure invariant.} The invariant $\inv(r,n)$ of the RDCSS data structure that we use for our proof is shown in \Cref{fig:rdcss-invariant}. The disjunction $\inactive(r,n) \lor \pactive(r,n)$ keeps track of the resources associated with the inactive and active state modes of the data structure and ties the logical value $n$ to the physical state. The invariant also keeps track of the clock resource $\clock(r,c)$. Throughout the rest of this section, we just write $\clock(c)$ instead of $\clock(r,c)$ since we will always reason about a single fixed $r$. The final conjunct $\proto(r,c)$ stores some resources for each past operation that has already linearized before time $c$. In particular, it is used to encode the \emph{helping protocol}. That is, it governs the transfer of the fulfillment resource for a completed \lstinline+rdcss+ operation from the helping thread that linearized the operation back to the thread that performed the operation.

The predicate $\inactive(r,n)$ simply stores the resource $r \mapsto \m{I}(n)$, indicating that $r$ is in inactive mode. Likewise, $\pactive(r,n)$ stores $r \mapsto \m{A}(d)$ to indicate that $r$ is in active mode. The predicate additionally contains a fraction of the descriptor location $d$. The invariant ties the logical value $n$ to the value $n_1$ that is physically stored in $d$ (i.e., the value last stored in $r$ before $r$ became active). It also contains a fraction of the permission $\ell$ to ensure that helping threads can always safely dereference $\ell$. The final conjunct $\oblc{r,\ell,n_1,m_1,n_2}$ is the linearization obligation of the active \lstinline+rdcss+ operation. The winning thread will convert this resource into the linearization receipt when it linearizes the operation and then transfer it to $\proto(r,c)$. Likewise the permissions to $l$ and $d$ are transferred to $\proto(r,c)$ at this point. They need to remain in the invariant forever, even after the active operation has been completed, because helping threads may still read these locations afterwards.

The constraint $q_d > \nicefrac{1}{2}$ on the fractional permission of $d$ plays two important roles in the proof. First, the correctness of the implementation relies on the fact that the descriptors $d$ are never reused after an operation has completed. Otherwise, there is an ABA problem. The implementation assumes a garbage collected semantics. This allows the invariant to retain the permissions for descriptors that will no longer be accessed. The invariant ensures that more than half of $d$'s permission remains in $\proto(c)$ after $d$ has been used by a past \lstinline+rdcss+ operation. One can then conclude that $d$ cannot have been reused by the currently active operation, as this would also require more than half of the permission in $\pactive(r,n)$, exceeding the maximal full permission amount.

Similarly, the constraints on the permission amounts on $d$ are used to govern the ownership transfer of the linearizability receipt for the associated \lstinline+rdcss+ operation. The thread executing the active \lstinline+rdcss+ operation retains $\nicefrac{1}{4}$ of the permission on $d$ in its local state throughout its own execution of \lstinline+complete+. By the time that the call to \lstinline+complete+ has returned, some thread must have linearized the active operation, which will increase the clock value $c$. At the point when the clock is incremented, the predicate $\proto(c)$ in the invariant forces the helping thread to relinquish ownership of the receipt $\fulc{c}{r,\ell,n_1,m_1,n_2,n_1}$ and transfer it to the invariant. The active thread will then use the knowledge that the clock must have incremented to retrieve the receipt from the invariant by trading it in for its $\nicefrac{1}{4}$ permission on $d$. If another thread had already retrieved the receipt, then the invariant would already own more than $\nicefrac{3}{4}$ of the permission on $d$, contradicting the fact that the active thread still owns $\nicefrac{1}{4}$.

In the remainder of the section, we discuss the proof in some more detail.

\begin{figure}
  \begin{align*}
    \authrstate(r,n) \defeq {} \; &  \rstate(r, n) \MSTAR \inv(r,n)\\
    \inv(r,n) \defeq {} & \exists c .\; (\inactive(r, n) \lor \pactive(r, n)) \MSTAR \clock(c) \MSTAR \proto(r,c)  \\
    \inactive(r, n) \defeq {} \; & r \mapsto \m{I}(n)\\
    \pactive(r, n_1) \defeq {} \; & \exists\, d\, q_d\, \ell\, q_\ell\, m\, m_1\, n_2.\; r \mapsto \m{A}(d) \MSTAR \fracpto{d}{q_d}{\m{D}(\ell, n_1, m_1, n_2)} \land q_d > \nicefrac{1}{2} \\
   & \qquad \MSTAR \fracpto{\ell}{q_\ell}{m} \MSTAR \oblc{r,\ell,n_1,m_1,n_2}\\
    \proto(r,c) \defeq {} \; & \bigmstar{c' < c}.\; \exists\, d\, q_d\, \ell\, q_\ell\, m\, n_1\, m_1\, n_2.\; \pfulc{c'}{r,n_1} \\
    & \qquad \lor \fracpto{\ell}{q_\ell}{m} \MSTAR \fracpto{d}{q_d}{\m{D}(\ell, n_1, m_1, n_2)} \land q_d > \nicefrac{1}{2} \MSTAR \pastof{(\clock(c') \MSTAR r \mapsto \m{A}(d))}\\
    & \qquad \MSTAR \left(\fulc{c'}{r,\ell,n_1,m_1,n_2,n_1} \lor \fracpto{d}{\fracPerm{1}{4}}{\m{D}(\ell,n_1,m_1,n_2)}\right)
  \end{align*}
  \caption{RDCSS data structure invariant\label{fig:rdcss-invariant}}
\end{figure}

\paragraph{Proof of \lstinline+rdcss+.}
We start with the proof of the \lstinline+rdcss+ operation whose outline is shown in \Cref{fig:rdcss-rdcss-proof}. The precondition corresponds to the precondition of the sequential specification in \Cref{fig:rdcss-spec}, except that we have replaced the abstract predicate $\rstate(r,n)$ by the full invariant $\authrstate(r,n)$ and also added the linearization obligation $\oblc{r,\ell,n_1,m_1,n_2}$. After the allocation of the descriptor $d$, the thread gains the full permission $d \mapsto \m{D}(\ell,n_1,m_1,n_2)$ in its local state, leading to the interference-free assertion on Line~\ref{line:rdcss-proof-alloc-d}. Next, the thread tries to change the state of $r$ to active using the \lstinline+CmpX+. The proof then proceeds by case analysis on the returned old value of $r$.

\begin{figure}
  \begin{lstlisting}[gobble=4,language=SPL,mathescape=true,escapechar=@]
    $\annot{\boxedInline{\fracpto{\ell}{q}{m} \MSTAR \authrstate(r,n)} \MSTAR \oblc{r,\ell,n_1,m_1,n_2}}$
    method rdcss($r$: Ref[State], $\ell$: Ref[Val], $n_1$: Val, $m_1$: Val, $n_2$: Val): Val {
      val $d$ = new Ref(D($\ell$, $n_1$, $m_1$, $n_2$))
      $\annot{\boxedInline{\fracpto{\ell}{q}{m} \MSTAR \authrstate(r,n) \land r \mapsto u} \MSTAR d \mapsto \m{D}(\ell,n_1,m_1,n_2) \MSTAR \oblc{r,\ell,n_1,m_1,n_2}}$ @\label{line:rdcss-proof-alloc-d}@
      match CmpX($r$, I($n_1$), A($d$)) with {
        case I($n'$) =>
          if ($n'$ == $n_1$) {
            $\annotml{\boxedInline{\authrstate(r,n_1) \land (r \mapsto \m{A}(d) \MSTAR \fracpto{\ell}{q}{m} \MSTAR  \fracpto{d}{\fracPerm{5}{8}}{\m{D}(\ell,n_1,m_1,n_2)} \MSTAR \oblc{r,\ell,n_1,m_1,n_2}) \land \clock(c))} \\ \MSTAR \fracpto{d}{\fracPerm{1}{8}}{\m{D}(\ell,n_1,m_1,n_2)} \MSTAR \fracpto{d}{\fracPerm{1}{4}}{\m{D}(\ell,n_1,m_1,n_2)}}$ @\label{line:rdcss-proof-cmpx-success}@
            $\annotml{\boxedInline{\authrstate(r,n) \land \fracpto{\ell}{q}{m}} \land \weakpastof{(\boxedInline{r \mapsto \m{A}(d) \MSTAR \clock(c)})} \MSTAR  \fracpto{d}{q_d}{\m{D}(\ell,\_)} \MSTAR \fracpto{d}{\fracPerm{1}{4}}{\m{D}(\ell,n_1,m_1,n_2)}}$ @\label{line:rdcss-proof-cmpx-success-equal-before-complete}@
            complete($r$, $d$)
            $\annotml{\boxedInline{\authrstate(r,n) \land \clock(c') \land c' > c} \land \pastof{(\boxedInline{r \mapsto \m{A}(d) \MSTAR \clock(c)})} \MSTAR \fracpto{d}{\fracPerm{1}{4}}{\m{D}(\ell,n_1,m_1,n_2)}}$@\label{line:rdcss-proof-cmpx-success-equal-after-complete}@
            // Hypothesis @\color{green!60!black}$\forall r\,c\,c'.\;\weakhypof{\boxedInline{r \mapsto \m{A}(d) \MSTAR \clock(c)}}{\boxedInline{\pfulc{c}{r,n} \MSTAR \clock(c') \land c' > c}}{\false}$@
            // Hypothesis: @\color{green!60!black}$\forall r\,c\,d\,d'.\;\weakhypof{\boxedInline{r \mapsto \m{A}(d) \MSTAR \clock(c)}}{\boxedInline{r \mapsto \m{A}(d') \MSTAR \clock(c)}}{d = d'}$@
            $\annotml{\fracpto{\ell}{q}{m} \mstar \boxedInline{\authrstate(r,n)} \MSTAR \fulc{c}{r,\ell,n_1,m_1,n_2,n_1}}$ @\label{line:rdcss-proof-cmpx-success-equal-done}@
            return $n_1$
          } else {
            $\annot{\boxedInline{\;\fracpto{\ell}{q}{m} \MSTAR \authrstate(r,n) \land r \mapsto I(n) \MSTAR \clock(c)} \land n \neq n_1 \MSTAR \oblc{r,\ell,n_1,m_1,n_2}}$ @\label{line:rdcss-proof-cmpx-success-nonequal}@
            $\annot{\boxedInline{\;\fracpto{\ell}{q}{m} \MSTAR \authrstate(r,n) \land r \mapsto I(n) \MSTAR \clock(c+1)} \land n \neq n_1 \MSTAR \fulc{c}{r,\ell,n_1,m_1,n_2,n}}$
            $\annot{\boxedInline{\;\fracpto{\ell}{q}{m} \MSTAR \authrstate(r,n)} \MSTAR \fulc{c}{r,\ell,n_1,m_1,n_2,n}}$
            return $n$
          }
        case A($d'$) =>
          $\annotml{\boxedInline{\fracpto{\ell}{q}{m} \MSTAR \authrstate(r,n) \land  \fracpto{\ell'}{q_{\ell'}}{m'}} \land \weakpastof{(\boxedInline{r \mapsto \m{A}(d') \MSTAR \clock(c)})} \MSTAR \fracpto{d'}{q'}{\m{D}(\ell',\_)} \\ {} \MSTAR d \mapsto \m{D}(\ell,n_1,m_1,n_2) \MSTAR \oblc{r,\ell,n_1,m_1,n_2}}$ @\label{line:rdcss-proof-cmpx-fail-before-complete}@
          complete($d$)
          $\annot{\boxedInline{\fracpto{\ell}{q}{m} \MSTAR \authrstate(r,n)} \MSTAR \oblc{r,\ell,n_1,m_1,n_2}}$ @\label{line:rdcss-proof-cmpx-fail-after-complete}@
          return rdcss($r$, $\ell$, $n_1$, $m_1$, $n_2$)
      }
    }    
    $\annot{v.\;\boxedInline{\fracpto{\ell}{q}{m} \MSTAR \authrstate(r,n)} \MSTAR \fulc{c}{r,\ell,n_1,m_1,n_2,v}}$
  \end{lstlisting}
  \caption{Proof of \lstinline+rdcss+\label{fig:rdcss-rdcss-proof}}
\end{figure}

If the old value of $r$ was $I(n_1)$, the \lstinline+CmpX+ succeeded and we end up on Line~\ref{line:rdcss-proof-cmpx-success}. Here, we know that the new value of $r$ must now be $\m{A}(d)$. To show that the invariant is maintained we need to establish $\pactive(r,n_1)$. So the prove moves $\ell$ to the invariant, along with $\nicefrac{5}{8}$ of the permission on $d$ and the linearization obligation. This proves the invariant. Before we proceed we must inspect the specification of \lstinline+complete+, which is as follows:
\begin{lstlisting}[gobble=2,language=SPL,numbers=none,escapechar=@]
  $\annot{\boxedInline{\authrstate(r,n) \land \fracpto{\ell'}{q'}{m'}} \land \weakpastof{(\boxedInline{r \mapsto \m{A}(d) \MSTAR \clock(c)})} \MSTAR \fracpto{d}{q}{\m{D}(\ell', n_1, m_1, n_2)}}$
  complete($r$,$d$)
  $\annot{\boxedInline{\authrstate(r,n) \land \clock(c') \land c' > c} \MSTAR \pastof{(\boxedInline{r \mapsto \m{A}(d) \MSTAR \clock(c)})}}$
\end{lstlisting}
The proof splits the remaining $\nicefrac{3}{8}$ permission on $d$ still available in the local state into $\nicefrac{1}{8}$ that is used towards satisfying the precondition of \lstinline+complete+ and another $\nicefrac{1}{4}$ that will be framed across the call to \lstinline+complete+. In addition, it creates the weak past predicate $\weakpastof{(\boxedInline{r \mapsto \m{A}(d) \MSTAR \clock(c)})}$ which establishes the connection between the value of $r$ at the current clock time $c$ and the descriptor value $d$. The resulting interference-free assertion is shown on Line~\ref{line:rdcss-proof-cmpx-success-equal-before-complete}. It implies the precondition of \lstinline+complete+.

The postcondition of \lstinline+complete+ yields the interference-free assertion on Line~\ref{line:rdcss-proof-cmpx-success-equal-after-complete}. At this point, the proof uses the fact $\authrstate(r,n) \land \clock(c') \land c' > c$ to obtain the following fact from $\proto(r,c')$:
\begin{align*}
\pfulc{c}{r,n_1'} \lor {} & \arraycolsep=2pt \left(\begin{array}{l}
\fracpto{\ell'}{q_{\ell'}}{m} \MSTAR \fracpto{d'}{q_d}{\m{D}(\ell', n_1', m_1', n_2')} \land q_d > \nicefrac{1}{2} \MSTAR \pastof{(\clock(c) \MSTAR r \mapsto \m{A}(d'))}\\
    {} \MSTAR \left(\fulc{c}{r,\ell',n_1',m_1',n_2',n_1'} \lor \fracpto{d'}{\fracPerm{1}{4}}{\m{D}(\ell',n_1',m_1',n_2')}\right)
\end{array}\right)
\end{align*}
The first disjunct $\pfulc{c}{r,n_1'}$ would imply that a \lstinline+get+ operation linearized at time $c$. We can obtain a contradiction for this case using temporal interpolation with the hypothesis
\[\forall c\,c'.\;\weakhypof{\boxedInline{r \mapsto \m{A}(d) \MSTAR \clock(c)}}{\boxedInline{\pfulc{c}{r,n} \MSTAR \clock(c') \land c' > c}}{\false}\]
The hypothesis holds because $r \mapsto \m{A}(d) \MSTAR \clock(c)$ is invariant as long as the clock does not increase. It also yields a contradiction with $\clock(c') \land c' > c$. The interferences obtained from the code of \lstinline+get+ and \lstinline+rdcss+ that increase the clock require as precondition $r \mapsto \m{I}(\_)$, which is incompatible with $r \mapsto \m{A}(d)$. So we only need to consider the interferences from \lstinline+complete+. These produce a receipt resource $\pfulc{c}{r,\ell',n_1',m_1',n_2',n'}$ that is invariant under all interferences and contradicts $\pfulc{c}{r,n}$.

Thus, the second disjunct obtained from $\proto(r,c')$ must hold. The next step is to show that $d$ must equal $d'$. This is done by unordered temporal interpolation, instantiating the following hypothesis two times symmetrically for $d$ and $d'$:
\[\forall r\,c\,d\,d'.\;\weakhypof{\boxedInline{r \mapsto \m{A}(d) \MSTAR \clock(c)}}{\boxedInline{r \mapsto \m{A}(d') \MSTAR \clock(c)}}{d = d'}\]
This hypothesis holds because $r \mapsto \m{A}(d) \MSTAR \clock(c)$ implies the invariant $r \mapsto \m{A}(d) \MSTAR \clock(c) \lor \clock(c') \land c' > c$, which in turn implies $d=d'$ with $r \mapsto \m{A}(d') \MSTAR \clock(c)$.

Now the proof uses $\fracpto{d}{q_d}{\m{D}(\ell', n_1', m_1', n_2')} \mstar \fracpto{d}{\nicefrac{1}{4}}{\m{D}(\ell, n_1, m_1, n_2)}$ to conclude the equality
\[(\ell', n_1', m_1', n_2')=(\ell, n_1, m_1, n_2)\enspace.\]
The next step is to derive $\fulc{c}{r,\ell,n_1,m_1,n_2,n_1}$ using the reasoning about the fractional permissions on $d$ that we already outlined when discussing the invariant. This resource can be transferred out of the invariant into the local state provided we conversely transfer $\fracpto{d}{\fracPerm{1}{4}}{\m{D}(\ell,n_1,m_1,n_2)}$ back into the global state so that the invariant is maintained. This finally yields the assertion on Line~\ref{line:rdcss-proof-cmpx-success-equal-done} which implies the postcondition and completes this case.

The next case is when the return value of the \lstinline+CmpX+ is $\m{I}(n)$ for $n \neq n_1$. This yields the assertion on Line~\ref{line:rdcss-proof-cmpx-success-nonequal}. This case corresponds to the pure case of the operation's specification. Hence, the proof can directly linearize the operation at this point and complete this case.

The final case is when the return value of the \lstinline+CmpX+ is $\m{A}(d)$. That is, a concurrent \lstinline+rdcss+ operation is already active. To satisfy the precondition of the call to \lstinline+complete+, the proof proceeds as follows. First, the postcondition of \lstinline+CmpX+ gives us $r \mapsto \m{A}(d')$ for some descriptor $d'$. Hence, we obtain $\pactive(r,n) \land r \mapsto \m{A}(d')$ from the invariant. We can now transfer some fraction $q'$ of the permission $\fracpto{d}{q_d}{\m{D}(\ell, n_1, m_1, n_2)}$ into the local state that leaves enough in the global state to maintain the invariant, say $q' = \frac{q_d - \nicefrac{1}{2}}{2}$. We also transfer some fraction $q_{\ell'}$ of the permission on $\ell'$ into the local state. Finally, we derive the weak past predicate $\weakpastof{(\boxedInline{r \mapsto \m{A}(d') \MSTAR \clock(c)})}$ which establishes the connection between the value of $r$ at the current clock time $c$ and the descriptor value $d'$. The resulting interference-free assertion on Line~\ref{line:rdcss-proof-cmpx-fail-before-complete} implies the precondition of \lstinline+complete+. The postcondition of \lstinline+complete+ yields the assertion on Line~\ref{line:rdcss-proof-cmpx-fail-after-complete} which implies the precondition of the recursive call.

\paragraph{Proof of \lstinline+complete+ and \lstinline+get+.} The proofs of \lstinline+complete+ and \lstinline+get+ follow similar reasoning. Their outlines are shown in \Cref{fig:rdcss-complete-proof} and \Cref{fig:rdcss-get-proof}. We omit a detailed description but provide the key reasoning steps for \lstinline+complete+ inline. The proof of \lstinline+get+ closely follows that of \lstinline+rdcss+ but is simpler.

\begin{figure}
  \begin{lstlisting}[gobble=4,language=SPL,mathescape=true,escapechar=@]
    $\annot{\boxedInline{\authrstate(r,n) \land \fracpto{\ell'}{q'}{m'}} \land \weakpastof{(\boxedInline{r \mapsto \m{A}(d) \MSTAR \clock(c)})} \MSTAR \fracpto{d}{q}{\m{D}(\ell', n_1, m_1, n_2)}}$
    method complete($r$: Ref[State], $d$: Ref[Descr]) {
      $\annot{\boxedInline{\authrstate(r,n) \land \fracpto{\ell'}{q'}{m'}} \land \weakpastof{(\boxedInline{r \mapsto \m{A}(d) \MSTAR \clock(c)})} \MSTAR \fracpto{d}{q}{\m{D}(\ell', n_1, m_1, n_2)}}$
      val D($\ell$, $n_1$, $m_1$, $n_2$) = !$d$ 
      $\annot{\boxedInline{\authrstate(r,n) \land \fracpto{\ell}{q'}{m'} \MSTAR \clock(c')} \MSTAR \pastof{(\boxedInline{r \mapsto \m{A}(d) \MSTAR \clock(c)})} \MSTAR \fracpto{d}{q}{\m{D}(\ell, n_1, m_1, n_2)}}$
      val $m$ = !$\ell$
      // Hypothesis: @\color{green!60!black}$\forall r\,c\,c'.\;\weakhypof{\boxedInline{\clock(c)}}{\boxedInline{\clock(c')}}{c' \geq c}$@
      $\annotml{\boxedInline{\authrstate(r,n')} \MSTAR \pastof{\left(\boxedInline{\fracpto{\ell}{q'}{m} \MSTAR \rstate(r, n) \MSTAR \clock(c')} \MSTAR \fracpto{d}{q}{\m{D}(\ell, n_1, m_1, n_2)}\right)} \land c' \geq c \\ {} \MSTAR \pastof{(\boxedInline{r \mapsto \m{A}(d) \MSTAR \clock(c)})} \MSTAR \fracpto{d}{q}{\m{D}(\ell, n_1, m_1, n_2)}}$
      val $u$ = $m$ == $m_1$ ? $n_2$ : $n_1$
      $\annotml{\boxedInline{\authrstate(r,n') \land r \mapsto u \MSTAR \clock(c'')} \MSTAR \pastof{\left(\boxedInline{\fracpto{\ell}{q'}{m} \MSTAR \rstate(r,n) \MSTAR \clock(c')} \MSTAR \fracpto{d}{q}{\m{D}(\ell, n_1, m_1, n_2)}\right)} \land c' \geq c \\ {} \land u = (\ite{m = m_1}{n_2}{n_1})  \MSTAR \pastof{(\boxedInline{r \mapsto \m{A}(d) \MSTAR \clock(c)})} \MSTAR \fracpto{d}{q}{\m{D}(\ell, n_1, m_1, n_2)}}$
      val $v$ = CmpX($r$, A($d$), I($u$))
      // Case $v = \m{A}(d)$
        // First: show @\color{green!60!black}$c'' \geq c'$ and hence $c'' \geq c$@ 
        // Hypothesis: @\color{green!60!black}$\forall r\,c'\,c''.\;\weakhypof{\boxedInline{\clock(c')}}{\boxedInline{\clock(c'')}}{c'' \geq c'}$@
        // Now do case analysis on @\color{green!60!black}$c'' = c \lor c'' > c$@
        // Case @\color{green!60!black}$c'' > c$@
          // First show that @\color{green!60!black}$\pfulc{c}{r,n}$@ cannot hold
          // Hypothesis @\color{green!60!black}$\forall r\,c\,c''.\;\weakhypof{\boxedInline{r \mapsto \m{A}(d) \MSTAR \clock(c)}}{\boxedInline{\pfulc{c}{r,n} \MSTAR \clock(c'') \land c'' > c}}{\false}$@
          // Now @\color{green!60!black}$\proto(r,c'')$ in $\authrstate$ gives $\fracpto{d'}{q_d'}{\m{D}(\ell, n_1, m_1, n_2)} \land q_d' > \nicefrac{1}{2} \MSTAR \pastof{(\clock(c) \MSTAR r \mapsto \m{A}(d'))}$@
          // Note also that @\color{green!60!black}$r \mapsto \m{A}(d)$ gives us $\pactive(r,c'') \land r \mapsto \m{A}(d)$ from $\authrstate$@
          // In turn, this gives us @\color{green!60!black}$\fracpto{d}{q_d}{\m{D}(\ell, \_)} \land q_d > \nicefrac{1}{2}$.@
          // That is @\color{green!60!black}$\fracpto{d}{q_d}{\m{D}(\ell, \_)} \land q_d > \nicefrac{1}{2} \MSTAR \fracpto{d'}{q_d'}{\m{D}(\ell, n_1, m_1, n_2)} \land q_d' > \nicefrac{1}{2}$@
          // Now use unordered interpolation to infer @\color{green!60!black}$d=d'$@ 
          // Hypothesis: @\color{green!60!black}$\forall r\,c\,d\,d'.\;\weakhypof{\boxedInline{r \mapsto \m{A}(d) \MSTAR \clock(c)}}{\boxedInline{r \mapsto \m{A}(d') \MSTAR \clock(c)}}{d = d'}$@
          // @\color{green!60!black}Contradiction because $q_d + q_d' > 1$@
        // Case @\color{green!60!black}$c'' = c$ (and hence $c'' = c = c'$)@
          // Show @\color{green!60!black}$n = n_1$@ with unordered interpolation
          // Hypothesis:
          //   @\color{green!60!black}$\forall r\,c\,d\,n\,n_1.\;\weakhypof{\boxedInline{r \mapsto \m{A}(d) \MSTAR \clock(c)}}{\boxedInline{\rstate(r,n) \MSTAR \clock(c) } \MSTAR \fracpto{d}{q}{\m{D}(\ell, n_1, \_)}}{n = n_1}$@
          // Hypothesis:
          //   @\color{green!60!black}$\forall r\,c\,d\,n\,n_1.\;\weakhypof{\boxedInline{\rstate(r,n) \MSTAR \clock(c) } \MSTAR \fracpto{d}{q}{\m{D}(\ell, n_1, \_)}}{\boxedInline{r \mapsto \m{A}(d) \MSTAR \clock(c)}}{n = n_1}$@
          $\annotml{\boxedInline{\authrstate(r,u) \land r \mapsto \m{I}(n_2) \MSTAR \clock(c+1) \MSTAR \fulc{c}{r,\ell,n_1,m_1,n_2,n_1}} \MSTAR \pastof{(\boxedInline{r \mapsto \m{A}(d) \MSTAR \clock(c)})}}$
          $\annotml{\boxedInline{\authrstate(r,n) \land \clock(c') \land c' > c} \MSTAR \pastof{(\boxedInline{r \mapsto \m{A}(d) \MSTAR \clock(c)})}}$
      // Case $v \neq \m{A}(d)$
        // Hypothesis: @\color{green!60!black}$\forall r\,d\,c\,c'.\;\weakhypof{\boxedInline{r \mapsto \m{A}(d) \MSTAR \clock(c)}}{\boxedInline{r \mapsto v \MSTAR \clock(c') \land v \neq \m{A}(d)}}{c' > c}$@
        $\annot{\boxedInline{\authrstate(r,n) \land \clock(c') \land c' > c} \MSTAR \pastof{(\boxedInline{r \mapsto \m{A}(d) \MSTAR \clock(c)})}}$
    }
    $\annot{\boxedInline{\authrstate(r,n) \land \clock(c') \land c' > c} \MSTAR \pastof{(\boxedInline{r \mapsto \m{A}(d) \MSTAR \clock(c)})}}$
 \end{lstlisting}
  \caption{Proof of \lstinline+complete+\label{fig:rdcss-complete-proof}}
\end{figure}

\begin{figure}
  \begin{lstlisting}[gobble=4,language=SPL,mathescape=true,escapechar=@]
    $\annot{\boxedInline{\authrstate(r,n)} \MSTAR \oblc{r}}$
    method get($r$: Ref[State]): Val {
      $\annot{\boxedInline{\authrstate(r,n) \land \clock(c)} \MSTAR \oblc{r}}$
      $\annot{\boxedInline{\authrstate(r,n) \land r \mapsto u \MSTAR \rstate(r,n)}}$
      match !$r$ with {
        case I($v$) =>
          $\annot{\boxedInline{\authrstate(r,n')} \land \pastof{(\boxedInline{r \mapsto I(n) \land \rstate(r,n) \land v = n})}}$
          $\annot{\boxedInline{\authrstate(r,n')} \land \pastof{(\boxedInline{\rstate(r,n) \land v = n})}}$
          $\annot{\boxedInline{\authrstate(r,n') \land \clock(c+1)} \MSTAR \fulc{c}{r,v}}$
          return $n$
        case A($d$) =>
          $\annot{\boxedInline{\authrstate(r,n) \land r \mapsto \m{A}(d)} \land \weakpastof{(\boxedInline{r \mapsto \m{A}(d) \MSTAR \clock(c)})} \MSTAR \oblc{r}}$
          $\annot{\boxedInline{\authrstate(r,n) \land \fracpto{\ell}{q}{m}} \land \weakpastof{(\boxedInline{r \mapsto \m{A}(d) \MSTAR \clock(c)})} \MSTAR \fracpto{d}{q}{\m{D}(\ell,n_1,m_1,n_2)} \MSTAR \oblc{r}}$
          complete($r$, $d$)
          $\annot{\boxedInline{\authrstate(r,n)} \MSTAR \oblc{r}}$
          return get($r$)
        }
    }
    $\annot{v.\;\boxedInline{\authrstate(r,n)} \MSTAR \FULc{r,v}}$
  \end{lstlisting}
  \caption{Proof of \lstinline+get+\label{fig:rdcss-get-proof}}
\end{figure}


	}
\end{document}